\documentclass[journal,10pt,twocolumn]{IEEEtran}
\usepackage{amsmath,amsfonts,amsthm,amssymb,bbm}
\usepackage{cite}
\usepackage{balance}
\usepackage{mathrsfs}
\usepackage{xcolor}
\usepackage{graphicx}
\allowdisplaybreaks

\usepackage{dblfloatfix}
\newcounter{MYtempeqncnt}
\allowdisplaybreaks
\DeclareMathOperator*{\argmax}{arg\,max}
\usepackage{subcaption}
\theoremstyle{plain} 
\newtheorem{theorem}{Theorem}

\newtheorem{definition}{Definition}
\newtheorem{lemma}{Lemma}
\newtheorem{proposition}{Proposition}
\theoremstyle{definition} \newtheorem{remark}{Remark}
\theoremstyle{definition}

\DeclareMathOperator\erf{erf}
\DeclareMathOperator\erfi{erfi}
\usepackage[ colorlinks = true,
linkcolor = blue,
urlcolor = blue,
citecolor = blue,
anchorcolor = green,
]{hyperref}

\title{Unifying Privacy Measures via Maximal $(\alpha,\beta)$-Leakage ({M$\alpha$beL})}
\author{}

\begin{document}
\author{Atefeh Gilani, \IEEEmembership{Student Member,~IEEE}, Gowtham R. Kurri, \IEEEmembership{Member,~IEEE},  Oliver Kosut, \IEEEmembership{Senior Member,~IEEE}, and Lalitha Sankar, \IEEEmembership{Senior Member,~IEEE}
\thanks{This article was presented in part at the 2022 Information Theory Workshop (ITW)~\cite{GilaniKKS22}. This work is supported in part by NSF grants CIF-1901243, CIF-2312666, and CIF-2007688.}
\thanks{Atefeh Gilani, Oliver Kosut, and Lalitha Sankar are with the School of Electrical, Computer and Energy Engineering, Arizona State University, Tempe, AZ 85281 USA (e-mail: {\tt{agilani2@asu.edu}, \tt{okosut@asu.edu}, \tt{lalithasankar@asu.edu}}). 

Gowtham R. Kurri was with the School of Electrical, Computer and Energy Engineering at Arizona State University. He is now with the Signal Processing and Communications Research Centre at International Institute of Information Technology, Hyderabad, Telangana - 500032, India (e-mail: \tt{gowtham.kurri@iiit.ac.in}).}}

\maketitle
\begin{abstract}
    We introduce a family of information leakage measures called \emph{maximal $(\alpha,\beta)$-leakage} (M$\alpha$beL), parameterized by real numbers $\alpha$ and $\beta$ greater than or equal to 1. The measure is formalized via an operational definition involving an adversary guessing an unknown (randomized) function of the data given the released data. We obtain a simplified computable expression for the measure and show that it satisfies several basic properties such as monotonicity in $\beta$ for a fixed $\alpha$, non-negativity, data processing inequalities, and additivity over independent releases. We highlight the relevance of this family by showing that it bridges several known leakage measures, including maximal $\alpha$-leakage $(\beta=1)$, maximal leakage $(\alpha=\infty,\beta=1)$, local differential privacy (LDP) $(\alpha=\infty,\beta=\infty)$, and local R\'{enyi} differential privacy (LRDP) $(\alpha=\beta)$, thereby giving an operational interpretation to local R\'{e}nyi differential privacy. We also study a conditional version of M$\alpha$beL on leveraging which we recover differential privacy and R\'{e}nyi differential privacy. A new variant of LRDP, which we call \emph{maximal R\'{e}nyi leakage}, appears as a special case of M$\alpha$beL for $\alpha=\infty$ that smoothly tunes between maximal leakage ($\beta=1$) and LDP ($\beta=\infty$). Finally, we show that a vector form of the maximal R\'{e}nyi leakage relaxes differential privacy under Gaussian and Laplacian mechanisms.
\end{abstract}

\section{Introduction}\label{intro}
\IEEEPARstart{H}{ow} much information does an observation released to an adversary reveal/leak about correlated sensitive data? This fundamental question arises in many privacy problems whenever data about users is stored (e.g., social networks and cloud-based services) and a certain level of information leakage is unavoidable in exchange for certain services. Limiting such an information leakage is desirable. Quantifying such leakage is the first step towards limiting it. In an effort to quantify this leakage precisely, a variety of privacy measures have been proposed in computer science~\cite{Dwork2006,smith2009foundations,braun2009quantitative,kasiviswanathan2011can,duchi2013local,Alvimetal12,Alvimetal14} and in information theory~\cite{MerhavA99,CalmonF12,SchielerC14,IssaW15,asoodeh2015maximal,AsoodehDL2017,mironov2017renyi,IssaWK2020,LiaoKS20,RassouliDeniz,SaedianCOS21,saeidian2022pointwise}.

For any leakage measure, one of the key challenges is to associate an operational interpretation to it, so that a certain amount of leakage corresponds to a particular privacy guarantee. Only a few leakage measures possess such an operational meaning. For example, the works in \cite{IssaWK2020,LiaoKS20}, which pertain to the release of observation due to a side channel, measure privacy in terms of an adversary's gain in \emph{guessing} the sensitive data after observing the released data. In particular, Issa~\emph{et al.}~\cite{IssaWK2020} consider an adversary interested in guessing a possibly randomized function of $X$. They study the logarithm of the multiplicative increase, upon observing $Y$, of the probability of correctly guessing a randomized function of $X$, say $U$, maximized over all the random variables $U$ such that $U-X-Y$ forms a Markov chain. This maximization captures the scenario that the function of interest $U$ is unknown to the system designer. The resulting quantity is referred to as \emph{maximal leakage} (MaxL). Liao \emph{et al.}~\cite{LiaoKS20} later generalized maximal leakage to a family of leakages, \emph{maximal $\alpha$-leakage} (Max-$\alpha$L) that consider a family of losses, namely $\alpha$-loss, to quantify the adversarial gain. 
Thus, similar to MaxL, Max-$\alpha$L quantifies the maximal logarithmic gain in a monotonically increasing power function (dependent on $\alpha$) applied to the probability of correctly guessing. By doing so, \cite{LiaoKS20} presents an operational interpretation of leakage measures using adversarial loss functions.

Among leakage measures motivated by worst-case adversaries, \emph{differential privacy} (DP)~\cite{Dwork2006} has emerged as the gold standard. Relegating precise definitions to the sequel, we state that a differentially private algorithm guarantees that its outputs restrict the adversary from distinguishing between neighboring datasets (i.e., the datasets that differ only in a single data entry), where each dataset can be viewed as $n$ instantiations of $X$. An operational interpretation of DP in the framework of hypothesis testing is given by Kairouz \emph{et al.}\cite{kairouz2015composition}, where they show that it determines the trade-off between probabilities of false alarm and missed detection. When privacy guarantees have to be provided in a distributed setting, \emph{local differential privacy} (LDP)~\cite{kasiviswanathan2011can,duchi2013local} provides strong privacy guarantees between any two realizations of $X$. Issa~\emph{et al.}~\cite{IssaWK2020} extended their definition of maximal leakage to introduce a \emph{worst-case} measure via \emph{maximal realizable leakage} (MaxRL) by taking a maximum over all realizations of $Y$.
They show that maximizing MaxRL over all the distributions $P_X$ yields LDP, thereby providing an operational interpretation to the latter. In the context of composing DP outputs sequentially, \emph{R\'{e}nyi differential privacy} (RDP)~\cite{mironov2017renyi} has emerged as a better variant to compute tight DP guarantees over multiple compositions of differentially private algorithms. Specifically, RDP relaxes DP based on the R\'{e}nyi divergence~\cite{renyi1961measures}. One can define \emph{local R\'{e}nyi differential privacy} (LRDP) as a generalization of LDP based on the R\'{e}nyi divergence. All the aforementioned leakage measures find applications in many areas such as privacy utility trade-offs~\cite{LiaoKS20,SaraCOS21,Alvim12}, hypothesis testing~\cite{LiaoSCT17}, source coding~\cite{Liuetal21}, Census data~\cite{abowd2018us}, anomaly detection~\cite{aljably2019anomaly}, age of information~\cite{NityaYSM22}, membership inference~\cite{SaedianCOS21}, deep learning~\cite{abadi2016deep}, posterior sampling~\cite{geumlek2017renyi}, and mechanism design~\cite{McSherryT07}. 

No single measure of privacy/information leakage suits all the scenarios in practice. In spite of the existence of a large number of privacy measures in the literature, it is often challenging to make an informed choice of a measure for a particular application in view of the diversity and complexity of various privacy measures. This compels a need for a unification of privacy leakage measures, in general, via a principled approach. In this paper, motivated by \cite{IssaWK2020,LiaoKS20}, we propose a family of information leakage measures, called \emph{maximal $(\alpha,\beta)$-leakage} (M$\alpha$beL), unifying all the aforementioned leakage measures into a structured landscape of leakage measures in an operationally motivated manner. 
\subsection{Main Contributions.}
The main contributions of this paper are as follows.
\begin{itemize}
    \item We introduce M$\alpha$beL in the framework of a guessing adversary, which is parameterized by two real numbers $\alpha\in[1,\infty]$ and $\beta\in[1,\infty]$(Definition~\ref{def:maximal-alpha-beta-leakage}), and obtain a simplified computable expression for it (Theorem~\ref{theorem:alpha-beta-leakage})~\cite{GilaniKKS22}. We prove that it satisfies all the axiomatic properties of a measure of information leakage, including non-negativity, monotonicity in $\beta$ for a fixed $\alpha$, data-processing inequalities, and additivity over independent releases (Theorem~\ref{theorem:properties})~\cite{GilaniKKS22}. We also show that it is continuous at $(\alpha,\beta)\in [1,\infty]\times [1,\infty]$, with the exception of $\alpha=\beta=1$ (Theorem~\ref{thm:continuity}).
    \item We show that this family of measures encompasses a host of existing leakage measures: in particular, Max-$\alpha$L ($\beta=1$), MaxL ($\alpha \to \infty,\beta=1$), LDP ($\alpha\to \infty,\beta \to \infty$), LRDP ($\alpha=\beta$) (Proposition~\ref{prop1} and Figure~\ref{alpha-beta-leakage-relationships})~\cite{GilaniKKS22}. Theorem~\ref{theorem:properties} gives another proof that LDP satisfies both the post-processing and linkage inequalities \footnote{In the context of privacy, 'linkage inequality' is often used interchangeably with 'preprocessing inequality'.} unlike DP which does not satisfy the linkage inequality~\cite{Basciftietal}. Interestingly, M$\alpha$beL is defined apparently in terms of average-case analysis (in the spirit of MaxL and Max-$\alpha$L), and yet, it recovers the worst-case LDP and LRDP by exploiting the interplay between the parameters $\alpha$ and $\beta$. 
    \item We propose \emph{conditional} M$\alpha$beL which takes into account the side-information an adversary may have and obtain a simplified computable expression for it (Theorem~\ref{theorem:conditional-alpha-beta-leakage}). We prove that M$\alpha$beL upper bounds conditional M$\alpha$beL if the side-information is conditionally independent of the released data given the original data (Theorem~\ref{theorem:conditioning-reduces-leakage}). That is, minimizing M$\alpha$beL is still a reasonable objective for a situation in which an adversary has access to side information which is unknown to the system designer. We also show that conditional M$\alpha$beL is subadditive over multiple releases (Theorem~\ref{theorem:sub-additivity}). 

    \item We generalize the conditional M$\alpha$beL to a vector form which allows us to quantify the leakage associated with a change in only one entry of dataset with an assumption that an adversary has the knowledge of all the remaining entries (Definition~\ref{differential-leakage}). We show that this naturally recovers DP and RDP (Proposition~\ref{vector-leakage-cases}). An important consequence of our results is an operational interpretation to RDP and LRDP. We note that this subsumes an operational meaning of LDP given by Issa \emph{et al.}~\cite{IssaWK2020} via maximal realizable leakage.
    \item We introduce a reparameterization of M$\alpha$beL, called maximal $(\alpha,\tau$)-leakage, in terms of $\alpha$ and $\tau$ with $\alpha\geq 1$ and $\beta=\frac{\alpha\tau}{\alpha+\tau-1}$, where $\tau\geq 1$. We show that this new measure, in contrast to M$\alpha$beL, is monotonic in both orders (Lemma~\ref{lemma:alpha_tau_variational}) and maintains continuity at all points $(\alpha,\tau)\in[1,\infty]\times[1,\infty]$ (Theorem~\ref{thm:continuity}). A new information-theoretic quantity arises as a special case of this leakage measure when $\alpha\to 1$. We call this measure \emph{$\tau$-Shannon leakage} and show that it recovers KL-divergence and Shannon channel capacity when $\tau=\infty$ and $\tau=1$, respectively (Proposition~\ref{prop4} and Figure~\ref{fig:alpha-tau-leakage}). 
    \item A new variant of LRDP, which we call \emph{maximal R\'{e}nyi leakage}, appears as a special case of M$\alpha$beL when $\alpha=\infty$ (Definition~\ref{def:Renyi-leakage}) that smoothly tunes between maximal leakage ($\beta=1$) and LDP ($\beta=\infty$). Finally, we show that a vector form of maximal R\'{e}nyi leakage relaxes differential privacy under Laplacian and Gaussian mechanisms (Proposition~\ref{lemma:Laplacian-mechanism-differential}).
\end{itemize}
\subsection{Related Work}
There are different approaches to quantifying information leakage. The works \cite{smith2009foundations}, \cite{braun2009quantitative}, and \cite{Alvimetal12} quantify leakage similar to maximal leakage with an adversary interested in guessing $X$ itself rather than its randomized functions. A variant of maximal leakage capturing the amount of information leaked about $X$ due to disclosing a single outcome $Y=y$ rather than focusing on the \emph{average outcome} as in maximal leakage has been studied in \cite{saeidian2022pointwise} and \cite{kurri2022operational}. Several measures have been proposed to quantify information leakage, e.g., maximal correlation~\cite{Calmonetal13,Lig18}, probability of correctly guessing~\cite{AsoodehDL2017}, total variation distance~\cite{Rassouli20}, mutual information ~\cite{Shannon49,PrabhakaranR07,TyagiNG11,CalmonF12,SankarRP13,asoodeh2015maximal,Wangetal16,Kameletal19}.

 The notion of DP is known to be very strict and has limited applicability~\cite{McSherry10,ComasD13}. Approximate differential privacy is  proposed as a relaxation of DP to allow data releases with higher utility~\cite{dwork2006our}. Resorting to the fact that composition with RDP has a simple linear form compared to DP, the authors of \cite{abadi2016deep} developed a method called Moments Accountant (MA) where the privacy guarantees are obtained first in terms of RDP before translating them to those of DP.  The shuffle model for differential privacy, where a shuffler randomly permutes the (randomized) data of all the users before forwarding them to the (untrusted) server, is introduced in \cite{bittau2017prochlo} and \cite{cheu2019distributed}.  The authors of \cite{Feldman21} and \cite{girgis2021renyi} obtained privacy gurantees with local randomized mechanisms for approximate DP and RDP, respectively. The role of interactivity in LDP is studied by Joseph~\emph{et al.}~\cite{joseph2019role}. Though there has been a lot of work building up on (L)DP and (L)RDP in the literature, an operational interpretation to (L)RDP remained open so far which we settled by introducing M$\alpha$beL that subsumes (L)RDP as a special case for $\alpha=\beta$. For an extensive list of leakage measures see the surveys by Wagner and Eckhoff~\cite{wagner2018technical}, Bloch~\emph{et~al.}~\cite{Blochsurvey}, and Hsu~\emph{et al.}~\cite{hsuetal21}.
 
\subsection{Organization of the Paper}
The remainder of this paper is organized as follows. We review various relevant information leakage measures in Section~\ref{section:info-measures}. We introduce (conditional) M$\alpha$beL in Section~\ref{sec:Theorems-maximal alpha-beta leakage} and prove that it satisfies the axiomatic properties of a leakage measure. In Section~\ref{sec:relationship}, we show that M$\alpha$beL recovers several existing information leakage measures as special cases. We present our results on reparameterization of M$\alpha$beL in Section~\ref{sec:reparameterization}. We extend the notion of the leakage to continuous alphabets in Section~\ref{section:continuous-alphabets} and discuss its applications in Section~\ref{sec:Illustration of Results}.

\emph{Notation.} We use capital letters to denote random variables, e.g., $X$, and capital calligraphic letters to denote their corresponding alphabet, e.g., $\mathcal{X}$. We write $U-X-Y$ to denote that the random variables form a Markov chain. We use $\text{supp}(X):=\{x:P_X(x)>0\}$ to denote the support set of a discrete random variable $X$. We use $H(X)$, $I(X;Y)$, and $D(P_X\|Q_X)$ to denote entropy, mutual information, and relative entropy, respectively. Given two probability distributions $P_X$ and $Q_X$ over an alphabet $\mathcal{X}$, we write $P_X\ll Q_X$ to denote that $P_X$ is absolutely continuous with respect to $Q_X$. We also consider continuous random variables and use $f_X$ to denote the probability density function of $X$. We use $\log$ to denote the natural logarithm. Throughout the sequel, we use the terms \textit{privacy mechanisms} and \textit{conditional distributions} interchangeably. Finally, in our analyses, we employ the extended real number line, extending the real numbers to include positive and negative infinity.

\section{Overview of Existing Information Leakage Measures}\label{section:info-measures}
We review the definitions of some existing information leakage measures.
\begin{definition}[Maximal leakage~\cite{IssaWK2020}]
 Let $P_{XY}$ be a joint distribution on finite alphabet $\mathcal{X}\times \mathcal{Y}$, where $X$ and $Y$ represent the original data and the released data, respectively. The maximal leakage from $X$ to $Y$ is defined as
  \begin{align}
    \mathcal{L}(X \to Y):=
     &\sup_{U-X-Y}\log\frac{\max\limits_{P_{\hat{U}|Y}}\sum_{u,y}P_{UY}(u,y)P_{\hat{U}|Y}(u|y)}{\max\limits_{P_{\hat{U}}}\sum_uP_U(u)P_{\hat{U}}(u)},  
 \end{align}
 where $U$ represents any randomized function of $X$ that an adversary is interested in guessing and takes values in an arbitrary finite alphabet. Moreover, $\hat{U}$ is an estimator of $U$ with the same support as $U$.
\end{definition}
Liao \emph{et al.}~\cite{LiaoKS20} generalized maximal leakage by introducing a tunable leakage measure known as maximal $\alpha$-leakage.
\begin{definition}[Maximal $\alpha$-leakage~\cite{LiaoKS20}]
Given a joint distribution $P_{XY}$ on finite alphabet $\mathcal{X}\times \mathcal{Y}$, the maximal $\alpha$-leakage from $X$ to $Y$ is defined as 
\begin{align}\label{eqn:max-alpha-leakage-def}
   \nonumber &\mathcal{L}_\alpha^{\emph{max}}(X \to Y)\\&:=\frac{\alpha}{\alpha-1}
     \sup_{U-X-Y}\log\frac{\max\limits_{P_{\hat{U}|Y}}\sum_{u,y}P_{UY}(u,y)P_{\hat{U}|Y}(u|y)^{\frac{\alpha-1}{\alpha}}}{\max\limits_{P_{\hat{U}}}\sum_uP_U(u)P_{\hat{U}}(u)^{\frac{\alpha-1}{\alpha}}},
\end{align}
 for $\alpha \in (1,\infty)$ and by continuous extension of \eqref{eqn:max-alpha-leakage-def} for $\alpha=1$ and $\alpha=\infty$, where $U$ represents any randomized function of $X$ with an arbitrary finite alphabet, and $\hat{U}$ is an estimator of $U$ with the same support as $U$.
\end{definition}
Liao~\emph{et al.}~\cite{LiaoKS20} showed that 
 \begin{align}\label{eqn:maximalalpha}
     \mathcal{L}_\alpha^{\text{max}}(X\rightarrow Y)=\sup_{P_{\tilde{X}}}I_\alpha^{\text{S}}(\tilde{X};Y),
 \end{align}
 where the supremum is over all the probability distributions $P_{\tilde{X}}$ on the support of $P_X$ and $I_\alpha^\text{S}(\cdot;\cdot)$ is the Sibson mutual information of order $\alpha$~\cite{sibson1969information}. 
Maximal $\alpha$-leakage recovers Shannon channel capacity (and mutual information) and 
maximal leakage for $\alpha=1$ and $\alpha=\infty$, respectively.

Conditional versions of maximal leakage and maximal $\alpha$-leakage are also defined to quantify the leakage when the adversary has access to side-information \cite[Definition~6]{IssaWK2020}, \cite[Definition~3]{liao2019robustness}.
\begin{definition}[Local differential privacy~\cite{kasiviswanathan2011can,duchi2013local}]
Given a conditional distribution $P_{Y|X}$ with $X$ and $Y$ taking values in finite sets $\mathcal{X}$ and $\mathcal{Y}$, respectively, the local differential privacy (LDP) is defined as
\begin{align}
    \mathcal{L}^{\emph{LDP}}(X\rightarrow Y):=\max_{\substack{y\in\mathcal{Y},\\ x,x^\prime\in\mathcal{X}}}\log{\frac{P_{Y|X}(y|x)}{P_{Y|X}(y|x^\prime)}}.
\end{align}
\end{definition}
Let $x^n=(x_1,x_2,\dots,x_n)$ denote a dataset comprising $n$ points from $\mathcal{X}$. We say $x^n=(x_1,x_2,\dots,x_n)$ and $\tilde{x}^n=(\tilde{x}_1,\tilde{x}_2,\dots,\tilde{x}_n)$ are neighbouring datasets, denoted $x^n \sim \tilde{x}^n$, if the Hamming distance between them is $1$, i.e., $\sum_{i=1}^n\mathbbm{1}\{x_i\neq \tilde{x}_i\}=1$, or in other words if there exists a unique $i\in[1:n]$ such that $x_i\neq \tilde{x}_i$. The following notion of differential privacy captures the privacy incurred of a user in participating in a dataset.

\begin{definition}[Differential privacy~\cite{Dwork2006}]
Given a conditional distribution $P_{Y|X^n}$ with $X^n$ and $Y$ taking values in finite sets $\mathcal{X}^n$ and $\mathcal{Y}$, respectively, the differential privacy (DP) is defined as
\begin{align}
    \mathcal{L}^{\emph{DP}}(X\rightarrow Y):=\max_{\substack{y\in\mathcal{Y},x^n,\tilde{x}^n\in\mathcal{X}^n:\\x^n\sim\tilde{x}^n}}\log\frac{P_{Y|X^n}(y|x^n)}{P_{Y|X^n}(y|\tilde{x}^n)}.
\end{align}

\end{definition}
\begin{definition}[Maximal realizable leakage~\cite{IssaWK2020}] Given a joint distribution $P_{XY}$ on finite alphabets $\mathcal{X}$ and $\mathcal{Y}$, the maximal realizable leakage from $X$ to $Y$ is defined as
\begin{align}\label{eq:max-realizable-leakage}
    \mathcal{L}^r(X\to Y)=\displaystyle\sup_{U-X-Y}\  \log\frac{\displaystyle\max_y \; \max_u \;P_{U|Y}(u|y)}{\displaystyle\max_u P_U(u)}
\end{align}
where $U$ takes values in an arbitrary finite alphabet.
\end{definition}

In~\cite{IssaWK2020}, it has been shown that 
\begin{align}\label{eqn:relation-LDP-realizable}
    \mathcal{L}^{\text{LDP}}(X\to Y)=\sup_{P_X} \ 
    \mathcal{L}^r(X \to Y),
\end{align}
where the supremum is over all probability distribution $P_X$.
A natural relaxation of DP is introduced by Mironov~\cite{mironov2017renyi} based on the R\'{e}nyi divergence to allow stronger results for composition.
\begin{definition}[R\'{e}nyi differential privacy~\cite{mironov2017renyi}]
Given a conditional distribution $P_{Y|X^n}$ with $X^n$ and $Y$ taking values in finite sets $\mathcal{X}^n$ and $\mathcal{Y}$, respectively, the R\'{e}nyi differential privacy (RDP) of order $\alpha$ is defined as
\begin{align}
    \nonumber&\mathcal{L}^{\emph{RDP}}_{\alpha}(X\rightarrow Y)\\&:=\max_{\substack{x^n,\tilde{x}^n\in\mathcal{X}^n:\\x^n\sim\tilde{x}^n}}D_\alpha(P_{Y|X^n=x^n}\|P_{Y|X^n=\tilde{x}^n})\\
    &=\max_{\substack{x^n,\tilde{x}^n\in\mathcal{X}^n:\\x^n\sim\tilde{x}^n}}\frac{1}{\alpha-1}
 \log \sum_y P_{Y|X^n}(y|\tilde{x}^n)^{1-\alpha} P_{Y|X^n}(y|x^n)^\alpha.
\end{align}
 
\end{definition}
We may define local R\'{e}nyi differential privacy as a generalization of local differential privacy based on the R\'{e}nyi divergence~\cite{renyi1961measures}.
\begin{definition}[Local R\'{e}nyi differential privacy]
Given a conditional distribution $P_{Y|X}$ with $X$ and $Y$ taking values in finite sets $\mathcal{X}$ and $\mathcal{Y}$, respectively, the local R\'{e}nyi differential privacy (LRDP) of order $\alpha$ is defined as
\begin{align}
    \nonumber&\mathcal{L}^{\emph{LRDP}}_{\alpha}(X\rightarrow Y)
    \\&:=\max_{x,x^\prime\in\mathcal{X}}D_\alpha(P_{Y|X=x}\|P_{Y|X=x^\prime})\\
    &=\max_{x,x^\prime\in\mathcal{X}}\frac{1}{\alpha-1}
 \log \sum_y P_{Y|X}(y|x')^{1-\alpha} P_{Y|X}(y|x)^\alpha.
\end{align}
\end{definition}
As $\alpha\rightarrow \infty$, it can be verified using L'H\^{o}pital's rule that LRDP and RDP simplify to LDP and DP, respectively.

\section{A Unified Measure of Information Leakage}\label{sec:Theorems-maximal alpha-beta leakage}
In this section, we introduce a unified leakage measure, called maximal $(\alpha,\beta)$-leakage (M$\alpha$beL). The new leakage measure includes maximal leakage, maximal $\alpha$-leakage, local R\'{e}nyi differential privacy and local differential privacy as its special cases. As our unified measure includes these leakage measures, its definition naturally inherits some complexity, mirroring that of the definitions of these existing measures. However, in Theorem~\ref{theorem:alpha-beta-leakage}, we shed light on its inherent complexity and significantly simplify it, thereby allowing us to relate to a large class of privacy measures. Furthermore, to recover R\'{e}nyi differential privacy and differential privacy, we introduce conditional M$\alpha$beL,
and provide a simplified form for it in Theorem~\ref{theorem:conditional-alpha-beta-leakage}.
\subsection{M$\alpha$beL}
Before introducing our most general unified leakage measure, we start with a measure which smoothly transitions between maximal leakage and LDP. The following definition makes use of the similarity between the definitions of maximal leakage and maximal realizable leakage, and the fact that the latter is related to LDP via \eqref{eqn:relation-LDP-realizable}.
\begin{definition}[Maximal R\'{e}nyi leakage of order $\beta$]\label{def:Renyi-leakage}
Given a conditional distribution $P_{Y|X}$ (or $f_{Y|X}$) on alphabets $\mathcal{X}$ and $\mathcal{Y}$, maximal R\'{e}nyi leakage of order $\beta$ from $X$ to $Y$ for $\beta \in [1,\infty)$ is defined as
\begin{align}\label{def:varint-of-LRDP}
   \nonumber&\mathcal{L}_{\beta}(X\to Y):=\sup_{P_{X}}\ \sup_{U\to X\to Y}\\&  \
 \log \frac{\displaystyle \max_{P_{\hat{U}|Y}} \left[\mathbb{E}_Y\left[ \left(\sum_u P_{U|Y}(u|Y) P_{\hat{U}|Y}(u|Y)\right)^{\beta}\right]\right]^{1/\beta}}{\displaystyle \max_{P_{\hat{U}}} \sum_u P_U(u)P_{\hat{U}}(u)}.
\end{align}
where $\hat{U}$ represents an estimator taking values from the same arbitrary finite alphabet as $U$. It is defined by continuous extension for $\beta \to \infty$.
\end{definition}

There are two important aspects to this definition. First, we introduce a parameter $\beta$ in the numerator in \eqref{def:varint-of-LRDP} thereby allowing a continuous transition from a simple average over $y$ (at $\beta=1$) to a maximum over $y$ (at $\beta \to \infty$). Thus, ignoring for the moment the supremum over $P_X$, when $\beta\to\infty$ we recover maximal realizable leakage, and at $\beta=1$ we recover maximal leakage. Secondly, by introducing the supremum over $P_X$, we do not change the value at $\beta=1$, since maximal leakage depends on the distribution of $X$ only through its support, and at $\beta\to\infty$ we recover LDP due to \eqref{eqn:relation-LDP-realizable}.

As a next step, we combine the definition of maximal R\'enyi leakage with that of maximal $\alpha$-leakage, keeping both as special cases, and including both $\alpha$ and $\beta$ as independent parameters. Remarkably, this yields our most general unified measure which also recovers LDP and LRDP.
\begin{definition}[Maximal $(\alpha,\beta)$-leakage (M$\alpha$beL)]\label{def:maximal-alpha-beta-leakage}  Given a conditional distribution $P_{Y|X}$ (or $f_{Y|X}$) on supports $\mathcal{X}$ and $\mathcal{Y}$, the maximal $(\alpha,\beta)$-leakage from $X$ to $Y$ for $(\alpha,\beta)\in(1,\infty)\times[1,\infty)$ is defined as
\begin{align}\label{eqn:alpha,beta-leakage-original-def}
   \nonumber&\mathcal{L}_{\alpha,\beta}(X\to Y):=\sup_{P_X}\ \sup_{U\to X\to Y}\frac{\alpha}{\alpha-1}\\&  \
 \log \frac{\displaystyle \max_{P_{\hat{U}|Y}} \left[\mathbb{E}_Y\left[\left(\sum_u P_{U|Y}(u|Y) P_{\hat{U}|Y}(u|Y)^{\frac{\alpha-1}{\alpha}}\right)^{\beta}\right]\right]^{1/\beta}}{\displaystyle \max_{P_{\hat{U}}} \sum_u P_U(u)P_{\hat{U}}(u)^{\frac{\alpha-1}{\alpha}}}.
\end{align}
where $\hat{U}$ represents an estimator taking values from the same arbitrary finite alphabet as $U$. M$\alpha$beL is defined by its continuous extension for $(\alpha,\beta)\in\{1,\infty\}\times [1,\infty)\setminus\{(1,1)\}$ and $(\alpha,\beta)\in(1,\infty)\times \{\infty\}$. It is also defined by $\displaystyle\lim_{\beta\to \infty}\lim_{\alpha\to 1}\mathcal{L}_{\alpha,\beta}(X\to Y)$ and  $\displaystyle\lim_{\beta\to \infty}\lim_{\alpha\to \infty}\mathcal{L}_{\alpha,\beta}(X\to Y)$ for $(1,\infty)$ and $(\infty,\infty)$, respectively.
\end{definition}
We remark that the definition of M$\alpha$beL in \eqref{eqn:alpha,beta-leakage-original-def} recovers the definition of maximal $\alpha$-leakage from \eqref{eqn:max-alpha-leakage-def} when $\beta=1$. While at the outset this simplification does not appear to be the same as that of maximal $\alpha$-leakage in \eqref{eqn:max-alpha-leakage-def} (i.e., the latter does not include a supremum over $P_X$), maximal $\alpha$-leakage depends on the distribution of $X$ only through its support (see \eqref{eqn:maximalalpha}), and therefore, including the supremum over $P_X$ does not change its value.

We also observe that the definition of M$\alpha$beL specializes to the definition of maximal R\'{e}nyi leakage of order $\beta$ in \eqref{def:varint-of-LRDP} for $\alpha=\infty$.

In the following theorem, we present a simplification of the expression of M$\alpha$beL in \eqref{eqn:alpha,beta-leakage-original-def}. As a special case of $\alpha\to \infty$, it also includes a simplified form for the maximal R\'enyi leakage of order $\beta$ introduced in Definition \ref{def:Renyi-leakage}.

\begin{theorem}\label{theorem:alpha-beta-leakage}
Let $X$ and $Y$ take values from finite supports $\mathcal{X}$ and $\mathcal{Y}$, respectively. For $(\alpha,\beta)\in(1,\infty)\times [1,\infty)$, M$\alpha$beL defined in \eqref{eqn:alpha,beta-leakage-original-def} simplifies to
\begin{align}\label{eqn:thm-alpha-beta-leakage}
\nonumber&\mathcal{L}_{\alpha,\beta}(X\to Y)
=\max_{x'\in \mathcal{X}} \  \sup_{P_{\Tilde{X}}}\frac{\alpha}{(\alpha-1)\beta} \\& \  \log 
\sum_{y\in \mathcal{Y}} P_{Y|X}(y|x')^{1-\beta} \left(\sum_{x\in \mathcal{X}} P_{\Tilde{X}}(x) P_{Y|X}(y|x)^\alpha \right)^{\beta/\alpha},
\end{align}
where $P_{\Tilde{X}}$ is a probability distribution on the support of $P_{X}$. For $\alpha\to \infty$, since Definition \ref{def:maximal-alpha-beta-leakage} simplifies to the definition of maximal R\'{e}nyi leakage of order $\beta$ in \eqref{def:varint-of-LRDP}, \eqref{eqn:thm-alpha-beta-leakage} simplifies maximal R\'enyi leakage of order $\beta$ to
    \begin{align}
\nonumber&\mathcal{L}_{\beta}(X\to Y)\\&=\max_{x'\in \mathcal{X}} \ \frac{1}{\beta}  \log
\sum_{y\in \mathcal{Y}} P_{Y|X}(y|x')^{1-\beta} \max_{x\in \mathcal{X}} P_{Y|X}(y|x)^\beta.\label{eqn:unconditional-infty-beta-step2}
\end{align}
\end{theorem}
A detailed proof for Theorem~\ref{theorem:alpha-beta-leakage} is given in Appendix~\ref{proof:thm-alpha-beta-leakage}. 

For $\beta\le\alpha$, the quantity inside the log in \eqref{eqn:thm-alpha-beta-leakage} is concave in $P_{\tilde{X}}$; thus the supremum over $P_{\tilde{X}}$ can be efficiently solved using convex optimization techniques. As we will show in Section~\ref{sec:relationship}, for $\beta\ge\alpha$, the supremum over $P_{\tilde{X}}$ can be replaced by a maximum over $x\in\mathcal{X}$. Thus, in either case the quantity in \eqref{eqn:thm-alpha-beta-leakage} can be efficiently computed for finite alphabets.

\begin{remark}
  To achieve a finite value for M$\alpha$beL, it is necessary that $P_Y\ll P_{Y|X=x'}$ for each $x'\in \mathcal{X}$. Failure to satisfy this condition leads to infinite leakage when $\beta > 1$. 
\end{remark}

Like other leakage measures, M$\alpha$beL satisfies several basic properties such as non-negativity, data processing inequalities and additivity, as shown in the following theorem.
\begin{theorem}\label{theorem:properties}
Let $X$ and $Y$ take values from finite alphabets $\mathcal{X}$ and $\mathcal{Y}$, respectively. For $\alpha\in (1,\infty)$ and $\beta \in [1,\infty)$, M$\alpha$beL 
\begin{enumerate}
 \item is monotonically non-decreasing in $\beta$ for a fixed $\alpha$;
  \item satisfies data processing inequalities, i.e., for the Markov chain $X -Y - Z$:
    \begin{subequations}
            \begin{equation}\label{data-processing2}
        \mathcal{L}_{\alpha,\beta}(X\to Z) \leq \mathcal{L}_{\alpha,\beta}(X \to Y)
    \end{equation}
    \begin{equation}\label{data-processing1}
        \mathcal{L}_{\alpha,\beta}(X\to Z) \leq \mathcal{L}_{\alpha,\beta}(Y \to Z).
        \end{equation}
    \end{subequations}
     \item is non-negative, i.e.,
    \begin{align}
        \mathcal{L}_{\alpha,\beta}(X\to Y) \geq 0
    \end{align}
with equality if and only if $X$ and $Y$ are independent.
\item satisfies additivity: i.e., if $(X_i,Y_i)$ for $i=1,2,\ldots,n$ are independent, then
\end{enumerate}
\begin{align}\label{eqn:additivity}
  \mathcal{L}_{\alpha,\beta}(X_1,\ldots,X_n \to Y_1,\ldots,Y_n)=\sum_{i=1}^n\mathcal{L}_{\alpha,\beta}(X_i \to Y_i).  
\end{align}

\end{theorem}
A detailed proof of Theorem~\ref{theorem:properties} is in Appendix~\ref{proof:thm-properties}.
\begin{remark}\label{re:cont-mabel}
M$\alpha$beL is continuous at $(\alpha,\beta)\in [1,\infty]\times [1,\infty]$, with the exception of the point $(\alpha,\beta)=(1,1)$. The proof of this property relies on a reparameterization of M$\alpha$beL and is covered in detail in Section~\ref{sec:reparameterization}. 
\end{remark}
\subsection{Conditional M$\alpha$beL}
Analogously to the connection between M$\alpha$beL and maximal leakage, we define conditional M$\alpha$beL based on conditional maximal leakage as follows.
\begin{definition}[Conditional M$\alpha$beL] Let $Z$ be the knowledge of an adversary or third-party about $(X,Y)$. Given a conditional distribution $P_{Y|X,Z}$ (or $f_{Y|X,Z}$) and a marginal distribution $P_Z$ (or $f_Z$) on supports $\mathcal{X},\mathcal{Y}$ and $\mathcal{Z}$, the conditional M$\alpha$beL from $X$ to $Y$ given $Z$ for $(\alpha,\beta)\in(1,\infty)\times[1,\infty)$ is defined as
\begin{align}\label{def:conditional_alpha-beta_leakage}
\nonumber&\mathcal{L}_{\alpha,\beta}(X\to Y|Z):=\sup_{P_{X|Z}} \ \sup_{U\to X\to Y|Z} \frac{\alpha}{(\alpha-1)\beta} \\&  \log \ \frac{\displaystyle \max_{P_{\hat{U}|Z,Y}} \mathbb{E} \left(\sum_u P_{U|Z,Y}(u|Z,Y) P_{\hat{U}|Z,Y}(u|Z,Y)^{\frac{\alpha-1}{\alpha}}\right)^{\beta}}{\displaystyle \max_{P_{\hat{U}|Z}}\ \mathbb{E} \left( \sum_u P_{U|Z}(u|Z) P_{\hat{U}|Z}(u|Z)^{\frac{\alpha-1}{\alpha}}\right)^\beta }.
\end{align}
Here $\hat{U}$ represents an estimator taking values from the same arbitrary finite alphabet as $U$, and the expression $U-X-Y|Z$ represents the conditional Markov chain constraint where
\begin{align*}
    \nonumber &P_{UXY|Z}(u,x,y|z)\\&=P_{X|Z}(x|z) \ P_{U|XZ}(u|x,z) \  P_{Y|XZ}(y|x,z).
\end{align*}
Thus, the conditional Markov chain $U-X-Y|Z$ is equivalent to the Markov chain $U-(X,Z)-Y$. The continuous extensions can be defined analogously to Definition~\ref{def:maximal-alpha-beta-leakage}. 
\end{definition}
For a similar reason to that stated below Definition~\ref{def:maximal-alpha-beta-leakage}, the definition of conditional M$\alpha$beL in \eqref{def:conditional_alpha-beta_leakage} recovers the definition of conditional maximal $\alpha$-leakage (and thus conditional maximal leakage)  for $\beta=1$. The following theorem simplifies the expression of conditional M$\alpha$beL.
\begin{theorem}\label{theorem:conditional-alpha-beta-leakage}
Let $X$, $Y$, and $Z$ take values from finite supports $\mathcal{X}$, $\mathcal{Y}$, and $\mathcal{Z}$, respectively. For $(\alpha,\beta)\in(1,\infty)\times[1,\infty)$, conditional M$\alpha$beL defined in \eqref{def:conditional_alpha-beta_leakage} simplifies to
\begin{align}\label{eqn:thm-conditional-alpha-beta-leakage}
\nonumber&\mathcal{L}_{\alpha,\beta}(X \to Y|Z)\\\nonumber&=\max_{z\in \mathcal{Z}} \ \max_{x'\in \mathcal{X}} \ \sup_{P_{\Tilde{X}|Z}} \
  \frac{\alpha}{(\alpha-1)\beta}    \log \Bigg[ \displaystyle \sum_{y\in \mathcal{Y}} P_{Y|X,Z}(y|x',z)^{1-\beta}\\& \ \times  \left( \displaystyle\sum_{x\in \mathcal{X}} P_{Y|X,Z}(y|x,z)^\alpha P_{\Tilde{X}|Z}(x|z)\right)^{\frac{\beta}{\alpha}}\Bigg]
    \end{align}
where $P_{\Tilde{X}|Z}$ is a distribution on the support of $P_{X|Z}$.
\end{theorem}
A detailed proof of Theorem~\ref{theorem:conditional-alpha-beta-leakage} is in Appendix~\ref{proof:thm-conditional-alpha-beta-leakage}.  
\begin{remark}
  Interestingly, despite the fact that there is an expectation over $z$ in the definition of conditional M$\alpha$beL, the simplified form has a maximum over $z$. This is in contrast to some other conditional measures in~\cite{9517944,9611409,10206606}.   
\end{remark}
Under a specific Markov chain, the following theorem shows the effect of the side information $Z$ on leakage about any function $U$ of $X$ through $Y$.
\begin{theorem}\label{theorem:conditioning-reduces-leakage}
Let $X$, $Y$, and $Z$ take values from finite alphabets $\mathcal{X}$, $\mathcal{Y}$, and $\mathcal{Z}$, respectively. If $Z-X-Y$ holds, for $(\alpha,\beta)\in(1,\infty)\times[1,\infty)$, we have
\begin{align}\label{eqn:thm-conditioning-reduces-leakage}
    \mathcal{L}_{\alpha,\beta}(X\to Y|Z) \le \mathcal{L}_{\alpha,\beta}(X\to Y),
\end{align}
with equality if for some $z \in \emph{supp}(Z)$, $$\emph{supp}(X)=\emph{supp}(X|Z=z).$$
\end{theorem}
A detailed proof of Theorem~\ref{theorem:conditioning-reduces-leakage} is in Appendix~\ref{proof:thm-conditioning-reduces-leakage}. Therefore, minimizing $\mathcal{L}_{\alpha,\beta}(X\to Y)$ is still a reasonable objective for a situation in which 
an adversary has access to side information $Z$ which is unknown to the system designer.
The following theorem shows that successive releases increase the total leakage.
\begin{theorem}[Sub-additivity/Composition]\label{theorem:sub-additivity}
  Let $Z$ represent the side information of an adversary and $X$, $Y_1$, $Y_2$, and $Z$ take values from finite supports $\mathcal{X}$, $\mathcal{Y}_1$, $\mathcal{Y}_2$, and $\mathcal{Z}$, respectively. For $(\alpha,\beta)\in(1,\infty)\times[1,\infty)$, we have
\begin{align}\label{eqn:sub-additivity}
   \nonumber &\mathcal{L}_{\alpha,\beta}(X\to Y_1,Y_2|Z)\\&  \le \mathcal{L}_{\alpha,\beta}(X\to Y_1|Z)+\mathcal{L}_{\alpha,\beta}(X \to Y_2|Y_1,Z).
    \end{align}
\end{theorem}
A detailed proof of Theorem~\ref{theorem:sub-additivity} is in Appendix~\ref{proof:thm-sub-additivity}. 
\begin{remark}
In the scenario where no side information is available to an adversary, applying  Theorem~\ref{theorem:sub-additivity}, we can show that
\begin{align}
  \mathcal{L}_{\alpha,\beta}(X\to Y_1,Y_2)&\le \mathcal{L}_{\alpha,\beta}(X\to Y_1)+\mathcal{L}_{\alpha,\beta}(X\to Y_2|Y_1)\label{ineq:sub-add-general}. 
\end{align}
 Combining this result with Theorem~\ref{theorem:conditioning-reduces-leakage},
 we can conclude that if $Y_1-X-Y_2$ holds, then 
\begin{align}
\mathcal{L}_{\alpha,\beta}(X\to Y_1,Y_2)\le \mathcal{L}_{\alpha,\beta}(X\to Y_1)+\mathcal{L}_{\alpha,\beta}(X \to Y_2)\label{eq:sub-add}.
\end{align}
Equation \eqref{eq:sub-add} recovers Liao~\emph{et al.}'s result \cite{LiaoKS20} on the sub-additivity of maximal $\alpha$-leakage, and equation \eqref{ineq:sub-add-general}
 generalizes it.
\end{remark}
\begin{remark}
Repeated use of privacy mechanisms on the outcome of previous private releases requires computing the overall privacy guarantees, a problem known as composition. A related useful property of any privacy measure, namely, composability, identifies the ease of computing  this overall privacy.
 Let $X$ be a sensitive random variable, $\mathcal{M}_1$ and $\mathcal{M}_2$  be privacy mechanisms\footnote{As previously mentioned, we employ the terms \textit{privacy mechanisms} and \textit{conditional distributions} interchangeably.}, and $\mathcal{M}$ be their composition. $\mathcal{M}$ is called \textit{adaptive} if $\mathcal{M}(X)=\left(\mathcal{M}_1(X),\mathcal{M}_2(X,\mathcal{M}_1(X))\right)$, that is, the output of $\mathcal{M}_2$ depends on both $X$ and $\mathcal{M}_1(X)$. In contrast, $\mathcal{M}$ is called \textit{non-adaptive} if $\mathcal{M}(X)=\left(\mathcal{M}_1(X),\mathcal{M}_2(X)\right)$, that is, the output of $\mathcal{M}_2$ depends on $\mathcal{M}_1(X)$ only through the random variable $X$. In \eqref{eqn:sub-additivity}, let $\displaystyle P_{Y_1|X,Z}$, $P_{Y_2|X,Z,Y_1}$, and $P_{Y_1,Y_2|X,Z}$ be privacy mechanisms associated with privacy measures $\mathcal{L}_{\alpha,\beta}(X\to Y_1|Z)$, $\mathcal{L}_{\alpha,\beta}(X\to Y_2|Y_1,Z)$, and $\mathcal{L}_{\alpha,\beta}(X\to Y_1,Y_2|Z)$, respectively. Random variables $Y_1$ and $Y_2$ can be viewed as $\mathcal{M}_1(X)$ and $\mathcal{M}_2\left(X,\mathcal{M}_1(X)\right)$, respectively. Thus, this implies that M$\alpha$beL satisfies adaptive composition. In \eqref{eq:sub-add}, random variables $Y_1$ and $Y_2$ can be viewed as $\mathcal{M}_1(X)$ and $\mathcal{M}_2(X)$, respectively, leading to a non-adaptive composition result for M$\alpha$beL. It is known that DP and R\'enyi DP mechanisms also satisfy composability. In Section~\ref{sec:relationship}, we show that DP and R\'enyi DP can be recovered through conditional M$\alpha$beL.  Consequently, we can employ \eqref{eqn:sub-additivity}  to establish bounds for the adaptive composition of R\'enyi DP and DP mechanisms.
\end{remark}

\section{Relationships of Other Leakage Measures with (Conditional) M$\alpha$beL}\label{sec:relationship}
As mentioned earlier, (conditional) M$\alpha$beL recovers (conditional) maximal $\alpha$-leakage for $\beta=1$ which simplifies to (conditional) maximal leakage for $\beta=1$ and $\alpha \to \infty$. Moreover, in this section we show that M$\alpha$beL includes various other leakage measures, particularly, different notions of DP (see Fig.~\ref{alpha-beta-leakage-relationships}).  
\begin{figure*}[t!]
    \centering
    \begin{subfigure}[t]{0.45\textwidth}
        \centering
        \includegraphics[width=1\linewidth]{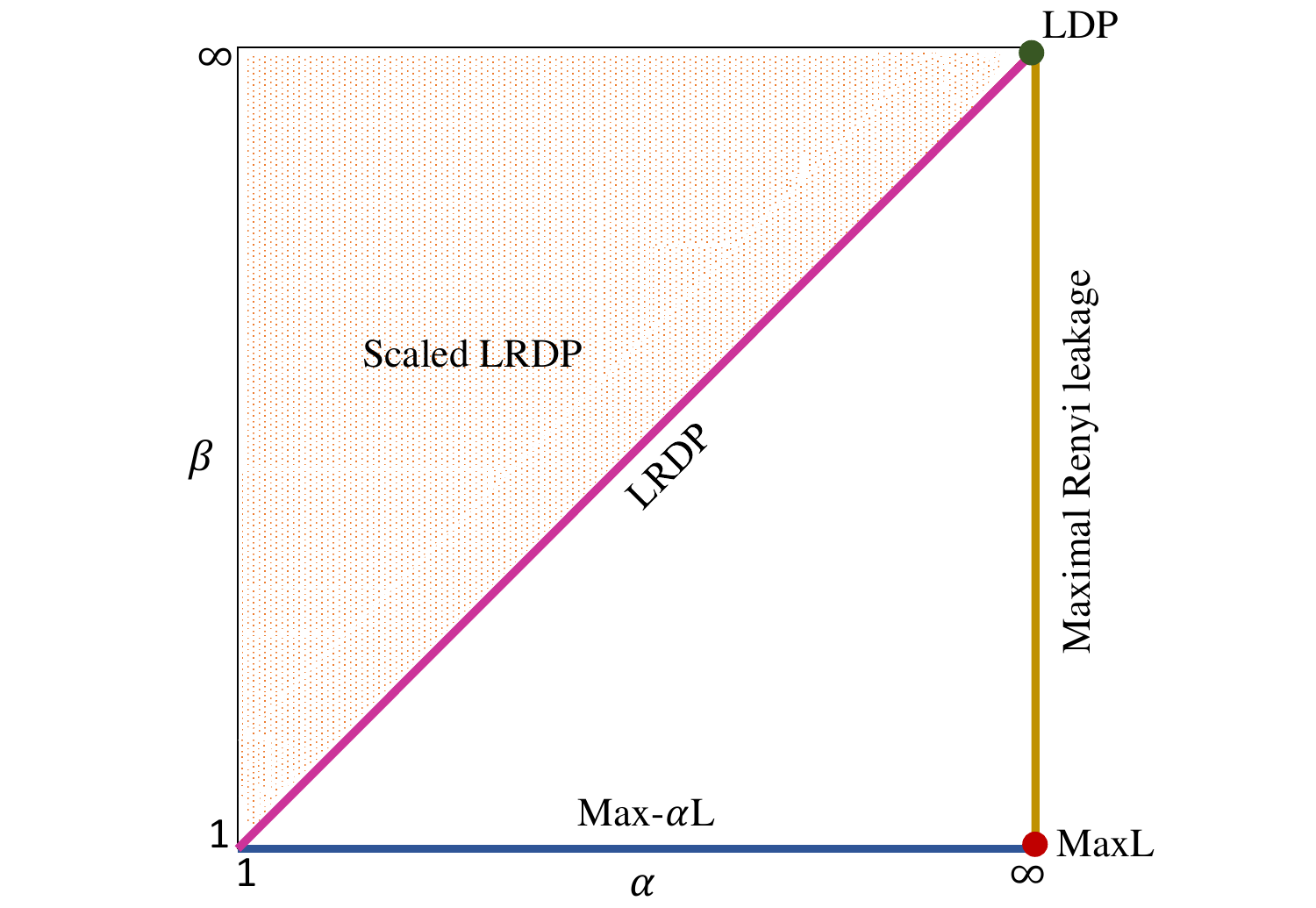}
        \caption{}
        \label{alpha-beta-leakage-relationships}
    \end{subfigure}%
    ~ 
    \begin{subfigure}[t]{0.45\textwidth}
        \centering
        \includegraphics[width=1\linewidth]{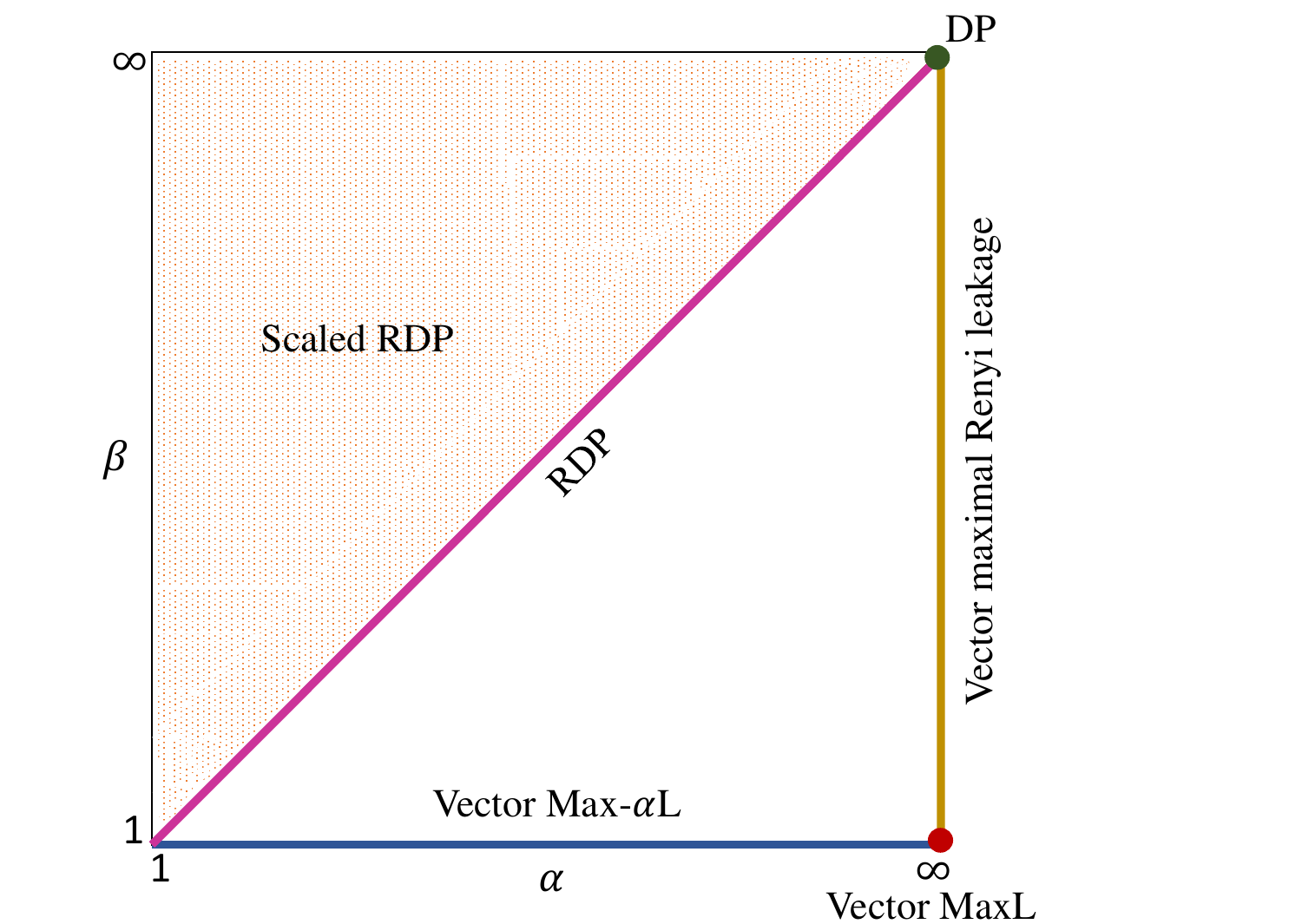}
        \caption{}
        \label{vector-alpha-beta-leakage-relationships}
    \end{subfigure}%
 
    \caption{Subplots \ref{alpha-beta-leakage-relationships} and \ref{vector-alpha-beta-leakage-relationships} show the relationships between existing leakage measures with M$\alpha$beL and vector M$\alpha$beL, respectively. }
    \label{fig:alpha_beta_leakage_and_other_measures}
\end{figure*}
\begin{proposition}\label{prop1}
    M$\alpha$beL can be simplified as follows:
    \begin{itemize}
        \item If $\alpha\leq \beta$, then $$\mathcal{L}_{\alpha,\beta}(X\to Y)=\displaystyle\frac{\alpha (\beta-1)}{(\alpha-1)\beta}\     \mathcal{L}^{\emph{LRDP}}_{\beta}(X\rightarrow Y).$$
         \item$ 
         \mathcal{L}_{\beta,\beta}(X\to Y)=  \mathcal{L}^{\emph{LRDP}}_{\beta}(X\rightarrow Y)$.
          \vspace{4mm}
           \item $\mathcal{L}_{\infty,\infty}(X\to Y)=\displaystyle\lim_{\beta\to \infty}\lim_{\alpha\to\infty}\mathcal{L}_{\alpha,\beta}(X\to Y)
           \\= \mathcal{L}^{\emph{LDP}}(X\rightarrow Y).$
    \end{itemize}
\end{proposition}
A detailed proof of Proposition~\ref{prop1} is in Appendix~\ref{App:prop1-proof}. The key result of Proposition \ref{prop1} is captured in Fig. \ref{alpha-beta-leakage-relationships} which shows that LDP (top right corner point in Fig. \ref{alpha-beta-leakage-relationships}) is the limiting point of M$\alpha$beL as $(\alpha\to \infty,\beta\to\infty)$.
In particular, M$\alpha$beL reduces to maximal R\'{e}nyi leakage of order $\beta$ (defined in \eqref{def:varint-of-LRDP}) for $\alpha=\infty$ (illustrated by the gold vertical line in Fig. \ref{alpha-beta-leakage-relationships}). 
\begin{proposition}\label{prop1-conditional}
    Conditional M$\alpha$beL can be simplified as follows:
    \begin{itemize}
        \item If $\alpha\leq \beta$, then 
        \begin{align*}
        &\mathcal{L}_{\alpha,\beta}(X\to Y|Z)=\displaystyle\max_{z,x',x} \ 
  \frac{\alpha}{(\alpha-1)\beta} \\& \ \log \displaystyle \sum_y P_{Y|X,Z}(y|x',z)^{1-\beta} P_{Y|X,Z}(y|x,z)^\beta.
  \end{align*}
        \item If $\alpha =\beta$, then \begin{align*}&\mathcal{L}_{\beta,\beta}(X\to Y|Z)=\displaystyle\max_{z,x',x} \ 
  \frac{1}{\beta-1} \\ & \ \log \displaystyle \sum_y P_{Y|X,Z}(y|x',z)^{1-\beta} P_{Y|X,Z}(y|x,z)^\beta.\end{align*}
           \item $\mathcal{L}_{\infty,\infty}(X\to Y|Z)=\displaystyle\lim_{\beta\to\infty}\lim_{\alpha\to\infty}\mathcal{L}_{\alpha,\beta}(X\to Y|Z)\\\vspace{3mm}=\displaystyle\max_{z,x',y,x} \log \frac{P_{Y|X,Z}(y|x,z)}{P_{Y|X,Z}(y|x',z)}$.
           \vspace{1mm}
           \item $\mathcal{L}_{\infty,\beta}(X\to Y|Z)=\displaystyle\max_{z,x'}\ \frac{1}{\beta} \\ \ \log\sum_y P_{Y|X,Z}(y|x',z)^{1-\beta}\ \max_x P_{Y|X,Z}(y|x,z)^\beta$.
    \end{itemize}
\end{proposition}
The proof of Proposition~\ref{prop1-conditional} follows from Theorem~\ref{theorem:conditional-alpha-beta-leakage} and similar steps as those in the proofs of Proposition~\ref{prop1} and equation \eqref{eqn:unconditional-infty-beta-step2}.
\subsection{Vector M$\alpha$beL}
In this section, we show that conditional M$\alpha$beL can be used to derive a general version of DP and RDP rather than their local versions (see Fig.~\ref{vector-alpha-beta-leakage-relationships}). Specifically, given a dataset with $n$ entries, we define a vector form of M$\alpha$beL as when the adversary has access to all but one of the entries.
 \begin{definition}[Vector M$\alpha$beL]\label{differential-leakage}
  Let $X^n=(X_1,X_2,\ldots,X_n)$ and $X_{-i}$ represent
 a dataset with $n$ entries and all entries except the $i$th, respectively. 
 Vector M$\alpha$beL is defined as
 \begin{align}\label{def:diff-alpha-beta-leakage}
     \mathcal{L}_{\alpha,\beta}^{\text{vec}}(X^n\to Y):=\max_i \  \mathcal{L}_{\alpha,\beta}(X_i\to Y|X_{-i}).
 \end{align}
  \end{definition}
 \begin{proposition}\label{vector-leakage-cases}
   For finite alphabets, vector M$\alpha$beL defined in \eqref{def:diff-alpha-beta-leakage} simplifies to
       \begin{align}
 \nonumber&\displaystyle \mathcal{L}_{\alpha,\beta}^{\text{vec}}(X^n\to Y)= \max_{i,x_{-i},x_i'} \sup_{P_{\Tilde{X}_i|X_{-i}}}
  \frac{\alpha}{(\alpha-1)\beta} \\\nonumber& \ \log\Bigg[  \displaystyle \sum_y P_{Y|X_i,X_{-i}}(y|x'_i,x_{-i})^{1-\beta}\\&\ \left( \displaystyle\sum_{x_i} P_{Y|X_i,X_{-i}}(y|x_i,x_{-i})^\alpha P_{\Tilde{X}_i|X_{-i}}(x_i|x_{-i})\right)^{\frac{\beta}{\alpha}}\Bigg],\label{eq:diff-simplified}
    \end{align}
    where  $P_{\Tilde{X}_i|X_{-i}}$ is a distribution on the support of $P_{X_i|X_{-i}}$. Moreover, it recovers
    \begin{itemize}
    \item vector maximal $\alpha$-leakage (and thus vector maximal leakage) when $\beta=1$, that is,
    \begin{align}\label{eq:vec-max-alpha-leak}
 \displaystyle&\nonumber \mathcal{L}_{\alpha,1}^{\text{vec}}(X^n\to Y)= \max_{i,x_{-i}} \sup_{P_{\Tilde{X}_i|X_{-i}}}
  \frac{\alpha}{\alpha-1}  \\& \log  \displaystyle \sum_y \Bigg[ \displaystyle\sum_{x_i} P_{Y|X_i,X_{-i}}(y|x_i,x_{-i})^\alpha P_{\Tilde{X}_i|X_{-i}}(x_i|x_{-i})\Bigg]^{\frac{1}{\alpha}}
    \end{align}
 
        \item a scaled RDP of order $\beta$ when $\alpha\le \beta$, that is,
         \begin{align}\label{eq:scaled-RDP}
\nonumber&\mathcal{L}_{\alpha,\beta}^{\text{vec}}(X^n\to Y)=      
  \max_{x^n\sim x'^n}\frac{\alpha}{(\alpha-1)\beta}\\& \ \log \displaystyle \sum_y P_{Y|X^n}(y|x'^n)^{1-\beta} P_{Y|X^n}(y|x^n)^\beta;
    \end{align}
        \item RDP of order $\alpha=\beta$ when $\alpha=\beta$, that is,
          \begin{align}
      \nonumber & \mathcal{L}_{\alpha=\beta}^{\text{vec}}(X^n\to Y)=      
  \max_{x^n\sim x'^n}\frac{1}{\beta-1} \\&\log \ \displaystyle \sum_y P_{Y|X^n}(y|x'^n)^{1-\beta} P_{Y|X^n}(y|x^n)^\beta;
    \end{align}
        \item DP when $\alpha,\beta\to \infty$;
        \item a variant of RDP of order $\beta$ for $\alpha \to \infty$ and an arbitrary $\beta$,
         which we call \emph{vector maximal R\'{e}nyi leakage}. That is, 
        \begin{align}\label{eqn:centralizedRendileakage}
 \nonumber&\mathcal{L}_{\infty,\beta}^{\text{vec}}(X^n\to Y)\\\nonumber&=\max_{i,x_i',x_{-i}}  \frac{1}{\beta} \  \log\Bigg[
\sum_y P_{Y|X_i,X_{-i}}(y|x_i',x_{-i})^{1-\beta} \\&\quad \times\ \max_{x_i} P_{Y|X_i,X_{-i}}(y|x_i,x_{-i})^\beta\Bigg].
\end{align}
    \end{itemize}
  \end{proposition}  
 We remark that vector maximal R\'{e}nyi leakage defined in \eqref{eqn:centralizedRendileakage} differs from RDP of order $\beta$ mainly in that the max over $x_i$ is inside the summation over $y$ rather than outside. A detailed proof of Proposition~\ref{vector-leakage-cases} can be found in Appendix~\ref{proof:vector-cases}.

\section{Maximal $(\alpha,\tau)$-leakage: A Reparameterization of M$\alpha$beL}\label{sec:reparameterization}
Maximal $(\alpha,\beta)$-leakage is not uniquely defined at $\alpha=\beta=1$. As $\alpha=\beta\rightarrow 1$, it can be verified using L'H\^{o}pital's rule that maximal $(\alpha=\beta)$-leakage, i.e., local R\'enyi DP, simplifies to $$\displaystyle\max_{x,x'}\ D_{KL}\left(P_{Y|X}(y|x)\|P_{Y|X}(y|x')\right)$$ whereas the limit of maximal $(\alpha,\beta=1)$-leakage, i.e., maximal $\alpha$-leakage, gives Shannon channel capacity as $\alpha \to 1$~\cite{LiaoKS20}. In this section, we consider a reparameterization which
leads to a new measure. The new measure is uniquely defined in all its endpoints and it is monotonic in both orders (unlike maximal $(\alpha,\beta)$-leakage which is monotonic in only one of its orders).
\begin{figure*}[t!]
    \centering
    \begin{subfigure}[t]{0.475\textwidth}
        \centering
        \includegraphics[width=1\linewidth]{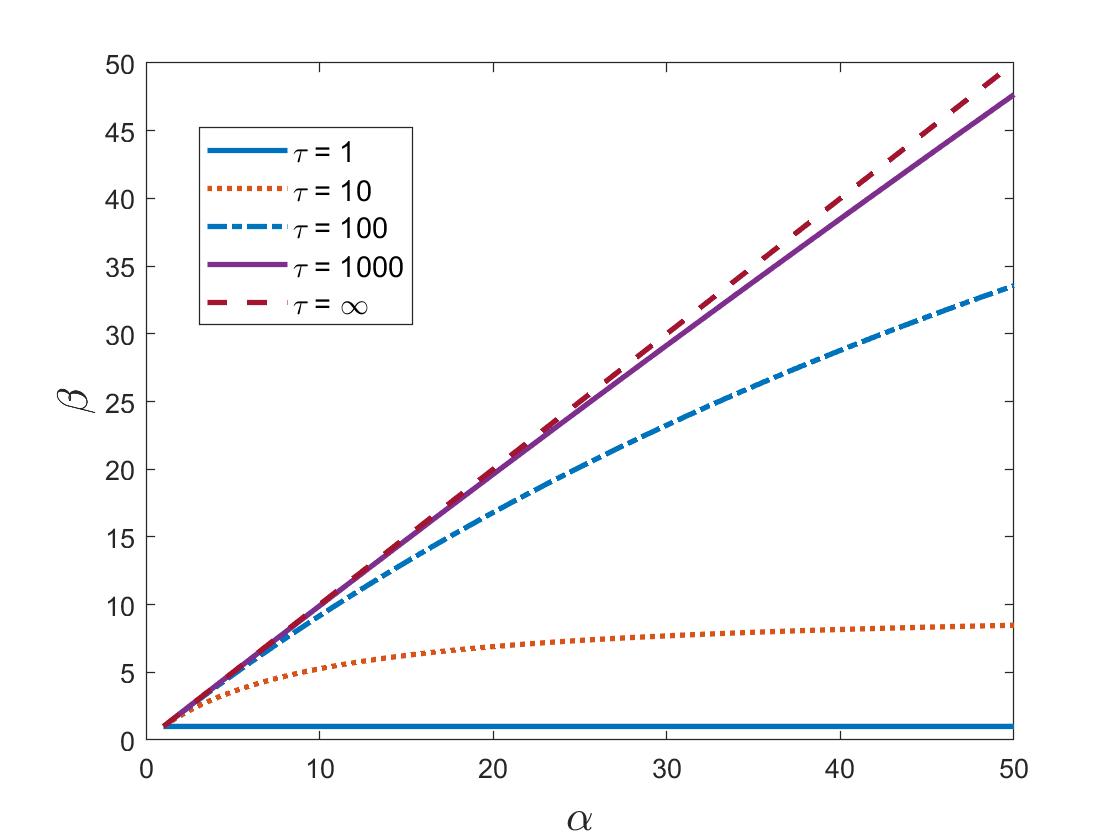}
        \caption{}
           \label{alpha-beta-curve}
    \end{subfigure}%
    ~ 
    \begin{subfigure}[t]{0.44\textwidth}
        \centering
        \includegraphics[width=1\linewidth]{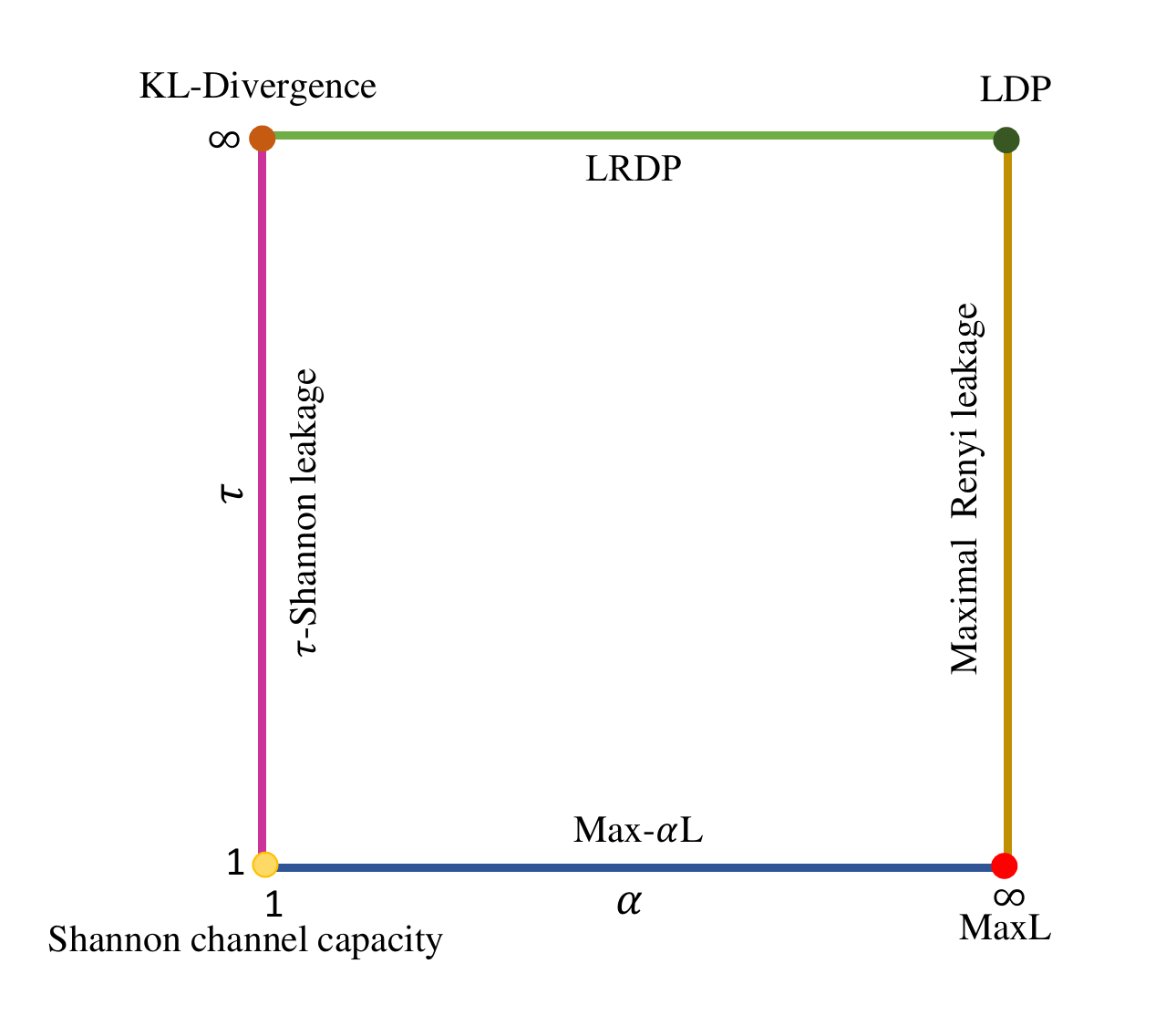}
        \caption{}
        \label{fig:alpha-tau-leakage}
    \end{subfigure}%
 
    \caption{Subplot \ref{alpha-beta-curve} shows $\beta$ vs $\alpha$ curves across different $\tau$ values. Subplot \ref{fig:alpha-tau-leakage} shows  relationship between maximal $(\alpha,\tau)$-leakage and other leakage measures as a function of $\alpha$ and $\tau$.} 
\end{figure*}
Let $\tau \in [1,\infty)$ and $\beta=\frac{\alpha   \tau}{\tau+\alpha-1}$, where $\alpha\in(1,\infty)$. We may re-write the expression of maximal $(\alpha,\beta)$-leakage in \eqref{eqn:thm-alpha-beta-leakage} in terms of $\alpha$ and $\tau$, as follows.
\begin{align}\label{alpha_tau_version}
&\nonumber\mathcal{L}_{\alpha,\tau}(X\to Y)
\\&= \max_{x'}\ \sup_{P_{\Tilde{X}}} \;  \left(\frac{1}{\alpha-1}+\frac{1}{\tau}\right) \log\Bigg[ \sum_y  P_{Y|X}(y|x')^{\frac{(\tau-1)(1-\alpha)}{\tau+\alpha-1}} \nonumber \\
& \ \times \left(\sum_x P_{\Tilde{X}}(x)P_{Y|X}(y|x)^\alpha\right)^{\frac{\tau}{\tau+\alpha-1}}\Bigg].
\end{align}
This measure is defined by its continuous extension for $(\alpha,\tau)\in\{1,\infty\}\times [1,\infty)$ and $(\alpha,\tau)\in(1,\infty)\times \{\infty\}$. It is also defined by  $\displaystyle\lim_{\tau\to \infty}\lim_{\alpha\to \infty}\mathcal{L}_{\alpha,\tau}(X\to Y)$ and $\displaystyle\lim_{\tau\to\infty} \lim_{\alpha\to 1} \mathcal{L}_{\alpha,\tau}(X\to Y)$ for $(\alpha=\infty,\tau=\infty)$ and $(\alpha=1,\tau=\infty)$, respectively. We call the quantity in \eqref{alpha_tau_version} \textit{maximal $(\alpha,\tau)$-leakage}. It is important to note that $\beta$ is non-decreasing in $\tau$ for a fixed $\alpha$ and so $\alpha>\beta$ (if $\tau\to\infty$, then $\beta\to \alpha$, see Fig.~\ref{alpha-beta-curve}).
\begin{lemma}\label{lemma:alpha_tau_variational} For $(\alpha,\tau)\in(1,\infty)\times [1,\infty)$, maximal $(\alpha,\tau)$-leakage can be represented by
\begin{align}\label{alpha_tau_variational}
&\nonumber\mathcal{L}_{\alpha,\tau}(X\to Y)
\\\nonumber&=\max_{x'}\,\sup_{P_{\Tilde{X}}}\,\inf_{Q_Y}\ \frac{1}{\alpha-1} \log\Bigg[\sum_{x,y} P_{\Tilde{X}}(x) P_{Y|X}(y|x)^\alpha \\& \ \times\left(Q_Y(y)^{\frac{1}{\tau}} P_{Y|X}(y|x')^{1-\frac{1}{\tau}}\right)^{1-\alpha}\Bigg],
\end{align}
and it is non-decreasing in $\tau$  and $\alpha$ for a fixed $\alpha$   and $\tau$, respectively. 
\end{lemma}
A detailed proof is in Appendix~\ref{proof:remark-reparameterization}.
\begin{remark}
Some of the relationships to other measures become clear from this lemma. Namely, if $\tau=1$, then we see the expression of maximal $\alpha$-leakage in terms of Sibson mutual information, that is, 
\begin{align}
\sup_{P_{\tilde{X}}} \ \inf_{Q_Y} D_\alpha(P_{\tilde{X}} \times P_{Y|X}\|P_{\tilde{X}}\times Q_Y)=\sup_{P_{\tilde{X}}} I_{\alpha}^S(\tilde{X};Y).
\end{align}
If $\tau\to \infty$, then we see the definition of LRDP as
\begin{align}
\max_{x,x'} D_\alpha(P_{Y|X=x}\|P_{Y|X=x'}).
\end{align}
\end{remark}
Note that $\frac{\tau}{\tau+\alpha-1}<1$ for $\tau\in [1,\infty)$ and $\alpha \in (1,\infty)$, and so the quantity inside the logarithm in \eqref{alpha_tau_version} is concave in $P_{\tilde{X}}$. Now we consider some endpoints of $\alpha$ and $\tau$ values (see Fig.~\ref{fig:alpha-tau-leakage}). If $\tau=1$, then we recover maximal $\alpha$-leakage which reduces to maximal leakage when $\alpha=\infty$.

 If $\tau=\infty$, then we recover local R\'enyi differential privacy of order $\alpha$ which simplifies to local differential privacy for $\alpha=\infty$. If $\alpha=\infty$, then we recover maximal R\'{e}nyi leakage of order $\tau$, that is,
\begin{align}
&\nonumber\mathcal{L}_{\alpha=\infty,\tau}(X\to Y)\\&=
\max_{x'} \; \frac{1}{\tau}  \log \sum_y  P_{Y|X}(y|x')^{1-\tau} \ \max_x  P_{Y|X}(y|x)^\tau.
\end{align}
Following similar steps as those in the proof of Proposition~\ref{prop1}, we can also show that 
\begin{align}
\nonumber\mathcal{L}_{\alpha=\infty,\tau=\infty}(X\to Y)&=\displaystyle\lim_{\tau\to \infty} \lim_{\alpha\to \infty}\mathcal{L}_{\alpha,\tau}(X\to Y)\\&=\mathcal{L}^{\text{LDP}}(X\to Y).
\end{align}
\begin{proposition}\label{prop4}
    For $\alpha\to 1$, maximal $(\alpha,\tau)$-leakage simplifies to
\begin{align}\label{eqn:tau-shannon-leakage}
   \nonumber&\mathcal{L}_{\alpha=1,\tau}(X\to Y)= \max_{x'}\ \sup_{P_{\Tilde{X}}}\Bigg[ \frac{1}{\tau}I(\tilde{X};Y)\\&+\left(1-\frac{1}{\tau}\right)D_{KL}\left(P_{Y|X}(y|x)\|P_{Y|X}(y|x')|P_{\tilde{X}}(x)\right)\Bigg].
\end{align}
\end{proposition}
We call the quantity in \eqref{eqn:tau-shannon-leakage} the \emph{$\tau$-Shannon leakage}. A detailed proof of Proposition~\ref{prop4} is in Appendix~\ref{proof:alpha-equals-one-tau-leakage}. 
\begin{remark}
    For $\tau\to \infty$, $\tau$-Shannon leakage simplifies to  
        \begin{align}
  \nonumber&\mathcal{L}_{\alpha=1,\tau=\infty}(X\to Y)\\&=\displaystyle\lim_{\tau\to\infty} \lim_{\alpha\to 1} \mathcal{L}_{\alpha,\tau}(X\to Y)\\
    &= \max_{x'}\ \sup_{P_{\Tilde{X}}}\ D_{KL}\left(P_{Y|X}(y|x)\|P_{Y|X}(y|x')|P_{\tilde{X}}(x)\right)\\
    &=\max_{x,x'}\ D_{KL}\left(P_{Y|X}(y|x)\|P_{Y|X}(y|x')\right)\label{eqn:linear in p_tilde_x}.
\end{align}
 \eqref{eqn:linear in p_tilde_x} follows because $D_{KL}\left(P_{Y|X}(y|x)\|P_{Y|X}(y|x')|P_{\tilde{X}}(x)\right)$ is linear in $P_{\tilde{X}}$ and so the supremum is attained at an extreme point. This quantity is KL divergence. Also, for $\tau=1$, $\tau$-Shannon leakage is given by 
\begin{align}
   \mathcal{L}_{\alpha=1,\tau=1}(X\to Y)=\sup_{P_{\tilde{X}}} \ I(\tilde{X};Y)
\end{align}
which is Shannon channel capacity. So $\tau$-Shannon leakage smoothly tunes between KL divergence ($\tau=\infty$) and Shannon channel capacity ($\tau=1$).
\end{remark} \begin{theorem}\label{thm:continuity} Let $P_Y\ll P_{Y|X=x'}$ for each $x'\in \mathcal{X}$.
Maximal $(\alpha,\tau)$-leakage is continuous in $(\alpha,\tau)$ for all $(\alpha,\tau)\in[1,\infty]\times[1,\infty]$, and M$\alpha$beL is continuous in $(\alpha,\beta)$ for all $(\alpha,\beta)\in[1,\infty]\times[1,\infty]\setminus \{(1,1)\}$.    
\end{theorem}
The proof of Theorem~\ref{thm:continuity} can be found in Appendix~\ref{proof:continuity}. 
Similar to the definition of vector M$\alpha$beL, maximal $(\alpha,\tau)$-leakage can be generalized to a vector form.
\section{Continuous Alphabets}\label{section:continuous-alphabets}
In this section, we generalize Theorem~\ref{theorem:alpha-beta-leakage} and Theorem~\ref{theorem:conditional-alpha-beta-leakage}  to continuous alphabets.
\begin{theorem}\label{thm:continuous-alpha-beta-leakage}
Let $X$ and $Y$ be continuous random variables having a continuous joint pdf $f_{XY}$. M$\alpha$beL defined in \eqref{eqn:alpha,beta-leakage-original-def} simplifies to
\begin{align}\label{eqn:continousalphabet}
\nonumber&\mathcal{L}_{\alpha,\beta}(X\to Y)
=\max_{x': f_X(x')>0} \  \sup_{f_{\Tilde{X}}}\ \frac{\alpha}{(\alpha-1)\beta} \\&\  \log 
\int_{\mathcal{Y}} f_{Y|X}(y|x')^{1-\beta} \left(\int_{\mathcal{X}} f_{\Tilde{X}}(x) f_{Y|X}(y|x)^\alpha dx\right)^{\beta/\alpha}dy,
\end{align}
where $f_{\tilde{X}}$ is a pdf on $\mathcal{X}$.
\end{theorem}
The proof of Theorem~\ref{thm:continuous-alpha-beta-leakage} and expressions similar to \eqref{eqn:continousalphabet} for the other cases, i.e., for discrete $X$ and continuous $Y$, and continuous $X$ and discrete $Y$ can be found in Appendix~\ref{proof:thm-continuous-alpha-beta-leakage}.
\begin{theorem}\label{thm:continuous-conditional-alpha-beta-leakage}
Let $X$, $Y$, and $Z$ be continuous random variables having a continuous joint pdf $f_{XYZ}$. The conditional M$\alpha$beL defined in \eqref{def:conditional_alpha-beta_leakage} simplifies to
\begin{align}\label{eq:continuous-conditional-simplified-expression}
   &\nonumber\mathcal{L}_{\alpha,\beta}(X \to Y|Z)\\\nonumber&=\max_{z} \ \max_{x'} \ \sup_{f_{\Tilde{X}|Z=z}} 
  \frac{\alpha}{(\alpha-1)\beta}  \log\Bigg[  \displaystyle \int_{\mathcal{Y}} f_{Y|X,Z}(y|x',z)^{1-\beta}\\&\ \times\left( \displaystyle\int_{\mathcal{X}} f_{Y|X,Z}(y|x,z)^\alpha f_{\Tilde{X}|Z=z}(x)\ dx \right)^{\frac{\beta}{\alpha}}dy\Bigg]
    \end{align}
where $f_{\tilde{X}|Z=z}$ is a pdf on the support of $f_{X|Z=z}$ for any $z$ that $f_Z(z)>0$.
\end{theorem}
  The proof of Theorem~\ref{thm:continuous-conditional-alpha-beta-leakage} follows a similar approach to the proof of Theorem~\ref{thm:continuous-alpha-beta-leakage}, and expressions similar to \eqref{eq:continuous-conditional-simplified-expression} for the other cases, i.e., for discrete $X$ and continuous $Y$, and continuous $X$ and discrete $Y$, can be derived similarly. Moreover, by applying similar steps to the proof of Proposition \ref{prop1}, we can recover RDP, DP, and vector maximal R\'{e}nyi leakage as special cases of vector M$\alpha$beL
  for continuous alphabets.
\subsection{Results for known mechanisms}\label{sec:Illustration of Results}
\begin{figure*}[t!]
    \centering
    \begin{subfigure}[t]{0.45\textwidth}
        \centering
        \includegraphics[width=1\linewidth]{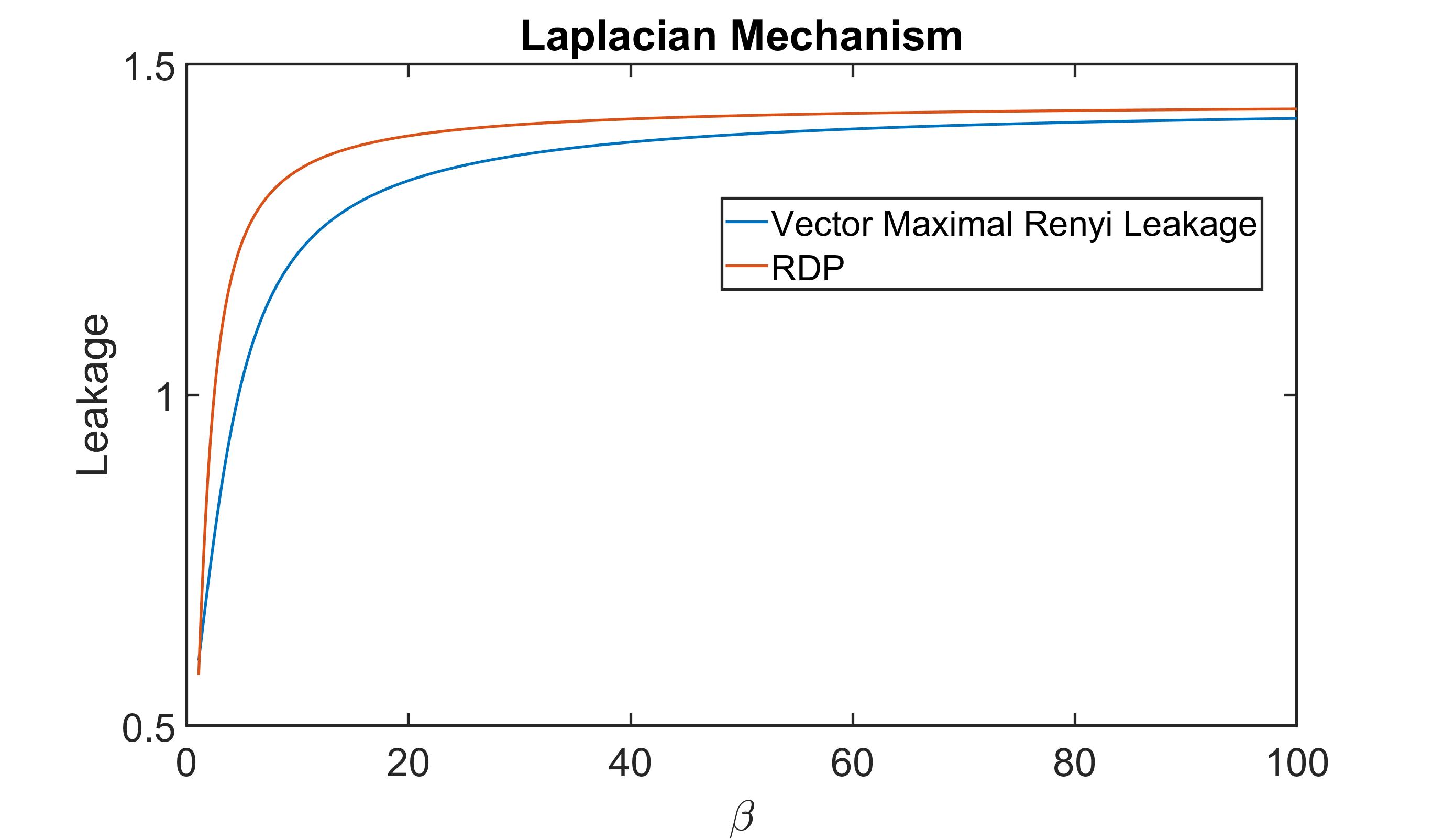}
        \caption{}
        \label{fig:lap-b=1}
    \end{subfigure}%
    ~ 
    \begin{subfigure}[t]{0.45\textwidth}
        \centering
        \includegraphics[width=1\linewidth]{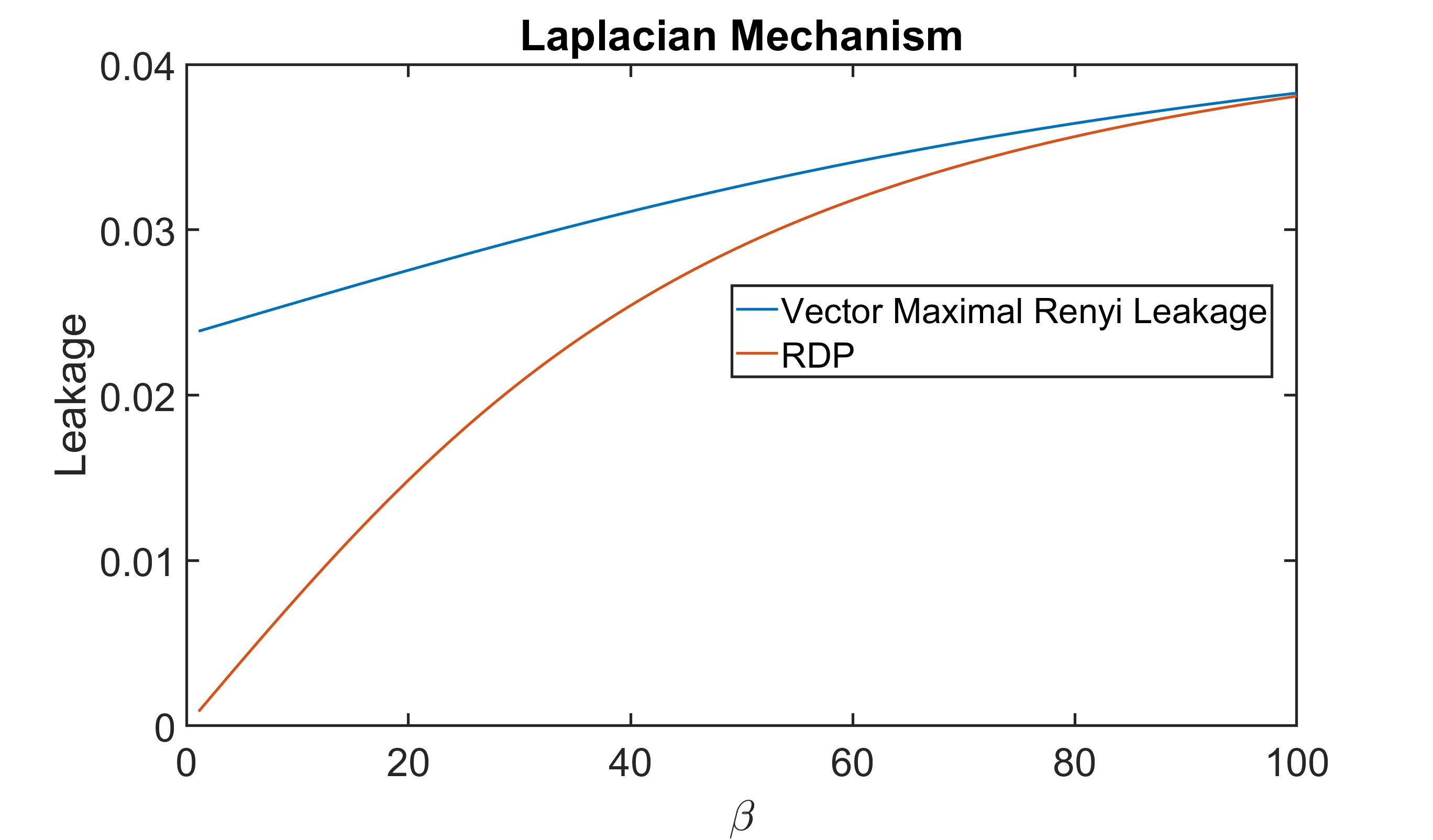}
        \caption{}
        \label{fig:lap-b=30}
    \end{subfigure}%
    \\
    \begin{subfigure}[t]{0.45\textwidth}
        \centering
        \includegraphics[width=1\linewidth]{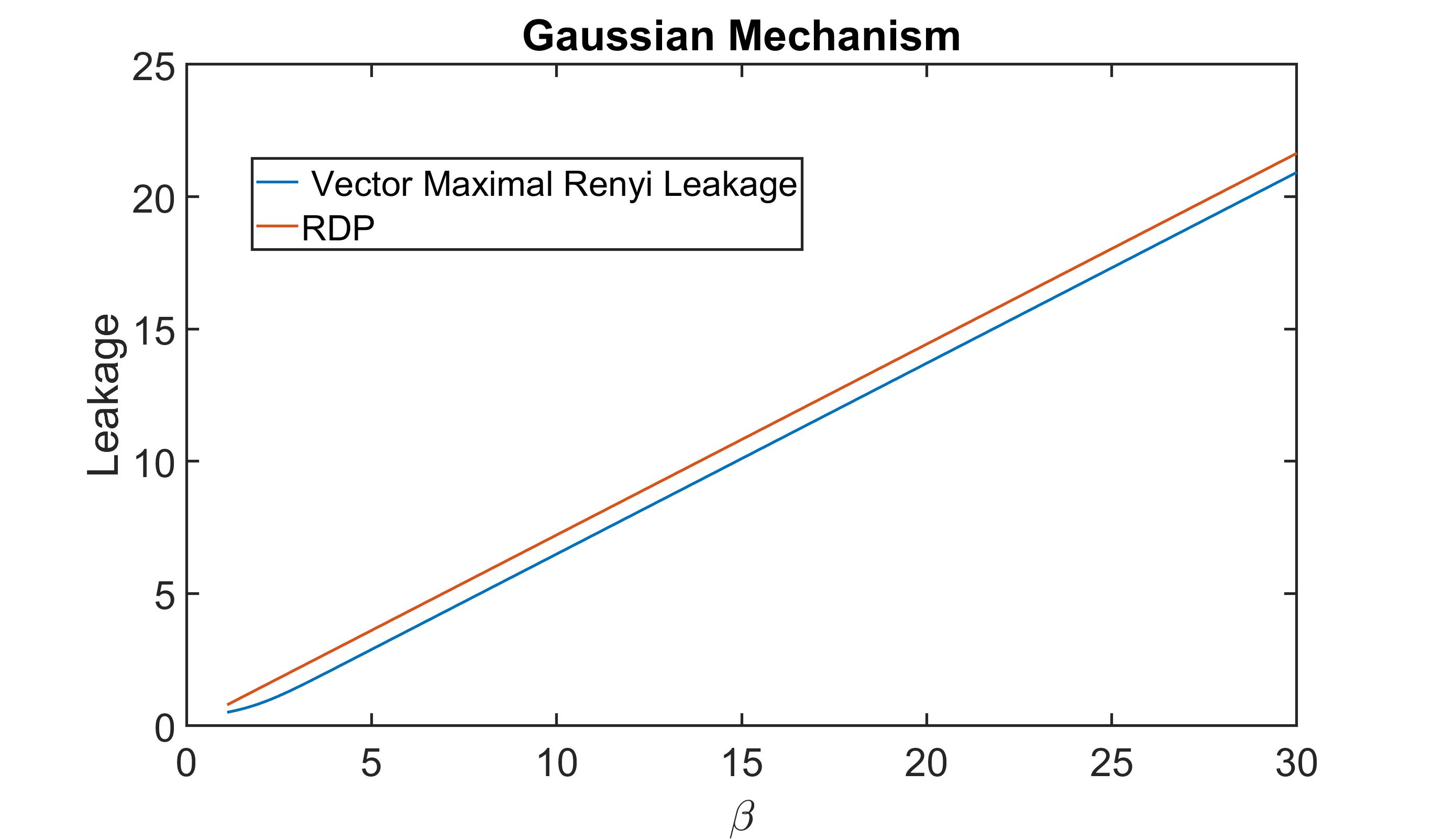}
         \caption{}
         \label{fig:gau-sig=1}
    \end{subfigure}%
    ~ 
    \begin{subfigure}[t]{0.45\textwidth}
        \centering
        \includegraphics[width=1\linewidth]{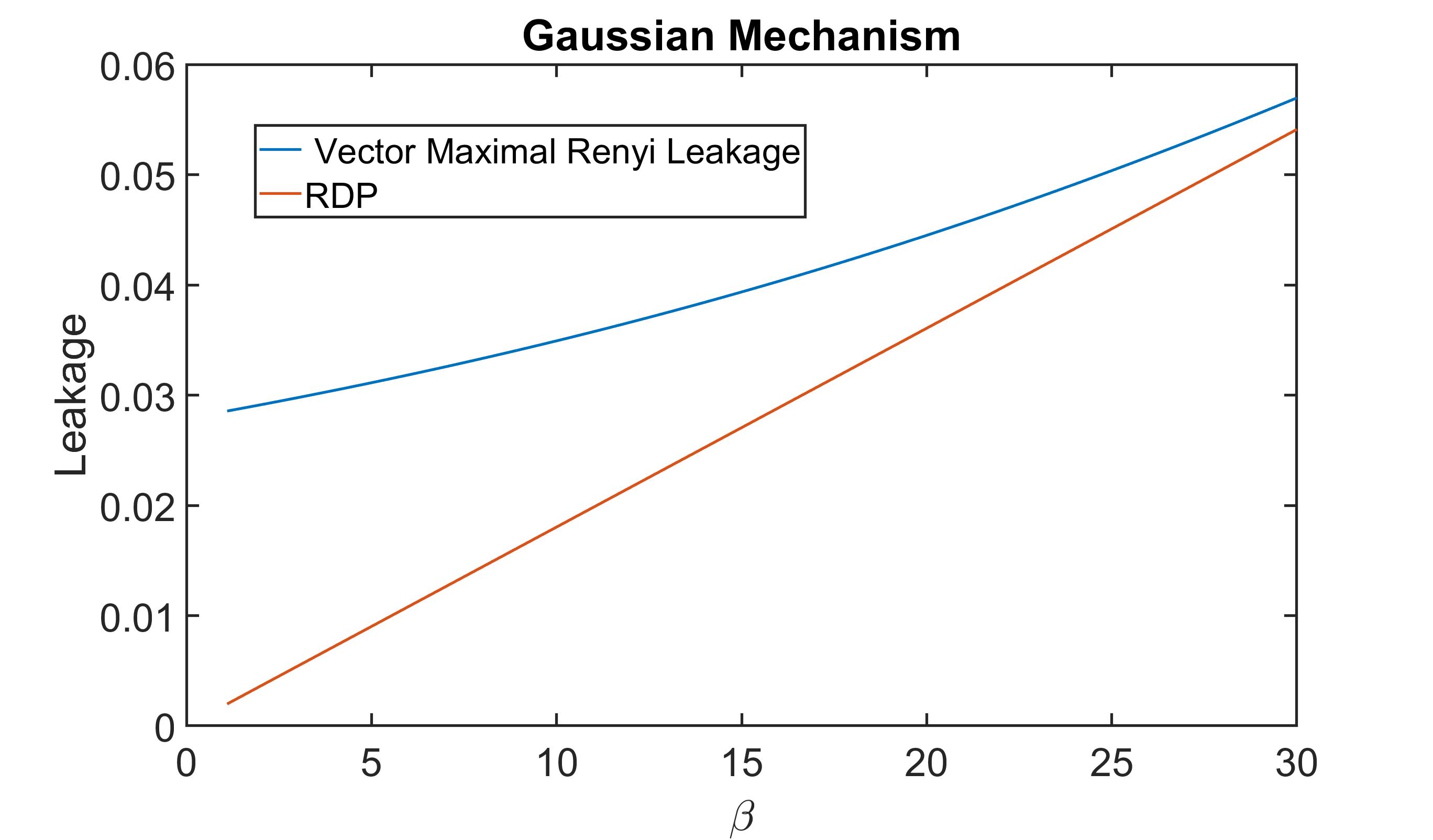}
         \caption{}
         \label{fig:gau-sig=20}
    \end{subfigure}
    \caption{Subplots \ref{fig:lap-b=1} and \ref{fig:lap-b=30} compare vector maximal Renyi leakage and Renyi differential privacy under Laplacian mechanisms with $b=1$ and $b=30$, respectively. Subplots \ref{fig:gau-sig=1} and \ref{fig:gau-sig=20} compare vector maximal Renyi leakage and Renyi differential privacy under Gaussian mechanisms with $\sigma=1$ and $\sigma=20$, respectively. In all subplots, we consider a situation in which the upper bounds \eqref{eq:lemma-Laplacian} and \eqref{eq:lemma-Gaussian} are achieved with equality, and the sensitivity parameter is 1.}
    \label{fig:compare-RDP}
\end{figure*}
In this section, we show how vector M$\alpha$beL relaxes differential privacy 
through vector maximal R\'{e}nyi leakage 
under Gaussian and Laplacian mechanisms (see Fig.~\ref{fig:compare-RDP}). The proofs for this section can be found in Appendix~\ref{proof:lemma-Laplacian-mechanism}. 
\begin{proposition}\label{lemma:Laplacian-mechanism-differential}
 Let $h :\mathcal{X}^n \to \mathbb{R}$ be a real-valued function such that for all $i,x_{-i},x_{i},\tilde{x}_i$ we have $|h(x_{-i},x_i)-h(x_{-i},\tilde{x}_i)|\leq \delta$, where $x^n=\{x_{-i},x_i\}$ and $\tilde{x}^n=\{x_{-i},\tilde{x}_i\}$ are neighboring datasets, and $\delta$ is a sensitivity parameter. For $\beta \in (1,\infty)$,
 \begin{itemize}
     \item if $\mathcal{M}(x^n)=h(x^n)+N$ where $N \sim Lap\;(0,b)$, then 
     \end{itemize}
     \begin{align}\label{eq:lemma-Laplacian}
   &\nonumber\mathcal{L}_{\infty,\beta}^{\text{vec}}(X^n\to {M}(X^n))\leq \frac{1}{\beta}\\&  \ \log \left[  \frac{1}{2}-\frac{1}{2(\beta-1)}+\left(\frac{1}{2}+\frac{1}{2(\beta-1)}\right) \exp \left(\frac{(\beta-1)\delta}{b}\right)\right];
\end{align} 
\begin{itemize}
\item if $\mathcal{M}(x^n)=h(x^n)+N$ where $N \sim \mathcal{N}(0,\sigma^2)$, then
\begin{align}\label{eq:lemma-Gaussian}
   &\nonumber\mathcal{L}_{\infty,\beta}^{\text{vec}}(X^n\to {M}(X^n))\\\nonumber&\leq \frac{1}{\beta} \ \log \Bigg[  \frac{1}{2}+\frac{1}{2\sqrt{\beta-1}} \erfi\left(\sqrt{\frac{\beta-1}{2\sigma^2}} \ \delta\right)\\&\ +\frac{1}{2} \exp{\left(\frac{\beta (\beta-1) \delta^2}{2 \sigma^2}\right)}\left(1+\erf\left(\frac{(\beta-1)\delta}{\sqrt{2}\sigma}\right)\right)\Bigg],
\end{align} 
where $\erf$ indicates the error function, that is, 
    $\erf(x)=\frac{2}{\sqrt{\pi}}\int_0^x e^{-t^2}\ dt$, 
and $\erfi$ indicates the imaginary error function, that is, $\erfi(x)=\frac{2}{\sqrt{\pi}}\int_0^x e^{t^2}\ dt.$
 \end{itemize}
Both upper bounds are achieved with equality if there exist $i$ and $x_{-i}$ such that the function $h(x_{-i},x_i)$ is surjective in $x_i$ and $\displaystyle\max_{x_{i},\tilde{x}_i}|h(x_{-i},x_i)-h(x_{-i},\tilde{x}_i)|= \delta$.
\end{proposition}
For $n=1$, the upper bounds on vector maximal R\'{e}nyi leakage in \eqref{eq:lemma-Laplacian} and \eqref{eq:lemma-Gaussian} collapse to  upper bounds on maximal R\'{e}nyi leakage under Laplacian and Gaussian mechanisms, respectively.
\section{Conclusion}
In this paper, we have introduced a new measure of information leakage called maximal $(\alpha,\beta)$-leakage that bridges several existing leakage measures, including maximal $\alpha$-leakage $(\beta=1)$, maximal leakage $(\alpha=\infty,\beta=1)$, (local) differential privacy $(\alpha=\infty,\beta=\infty)$, (local) R\'{enyi} differential privacy $(\alpha=\beta)$, and a variant of RDP, which we call vector maximal R\'{e}nyi leakage ($\alpha=\infty$). This provides a much-needed operational interpretation to (local) RDP. We believe that our work has taken a step towards identifying the common characteristics of various information leakage measures despite their diversity. For example, our formulation allows us to smoothly transition from average-case leakage measures to worst-case leakage measures by exploiting the interplay between the parameters $\alpha$ and $\beta$. Finally, we posit that the unification provided by our guessing framework allows us to tailor the proposed leakage measure to study privacy-utility tradeoffs under different settings depending on the context.

\appendices
\section{Proof of Theorem~\ref{theorem:alpha-beta-leakage}}\label{proof:thm-alpha-beta-leakage}
For $\alpha\in (1,\infty)$ and $\beta \in [1,\infty)$, we first bound $\mathcal{L}_{\alpha,\beta}(X \to Y)$ from above and then, give an achievable scheme.\\
\textbf{Upper Bound:}
Consider the optimization in the denominator of \eqref{eqn:alpha,beta-leakage-original-def}:
\begin{align}
\max_{P_{\hat{U}}} \sum_u P_U(u)P_{\hat{U}}(u)^{\frac{\alpha-1}{\alpha}}.
\end{align}
This is solved by
\begin{equation}
P_U(u) P_{\hat{U}}(u)^{-1/\alpha}=\nu
\end{equation}
for some constant $\nu$. So we have
\begin{align}
P_{\hat{U}}(u)=\frac{P_U(u)^\alpha}{\sum_{u'} P_U(u')^\alpha}.
\end{align}
Thus the denominator becomes
\begin{align}\label{eqn:optimization-result}
\sum_u P_U(u) \left(\frac{P_U(u)^\alpha}{\sum_{u'} P_U(u')^\alpha}\right)^{\frac{\alpha-1}{\alpha}}=\left(\sum_u P_U(u)^\alpha\right)^{\frac{1}{\alpha}}.
\end{align}
Similarly, the numerator becomes
\begin{align}
\left[\sum_y P_Y(y) \left(\sum_u P_{U|Y}(u|y)^\alpha\right)^{\beta/\alpha}\right]^{1/\beta}.
\end{align}
Thus, the logarithmic term in \eqref{eqn:alpha,beta-leakage-original-def} reduces to
\begin{align}
& \log\frac{\left[\sum_y P_Y(y) \left(\sum_u P_{U|Y}(u|y)^\alpha\right)^{\beta/\alpha}\right]^{1/\beta}}
{\left(\sum_u P_U(u)^\alpha\right)^{1/\alpha}}
\\&= \log
\frac{\left[
\sum_y P_Y(y)^{1-\beta} \left(\sum_u P_{U,Y}(u,y)^\alpha\right)^{\beta/\alpha}
\right]^{1/\beta}}
{\left(\sum_u P_U(u)^\alpha\right)^{1/\alpha}}
\\\label{eq:simplified-leakage}
&=\frac{1}{\beta}\log
\sum_y P_Y(y)^{1-\beta} \left(\frac{\sum_u P_U(u)^\alpha P_{Y|U}(y|u)^\alpha}{\sum_u P_U(u)^\alpha}\right)^{\frac{\beta}{\alpha}}.
\end{align}
Using Jensen's inequality and the Markov chain $U-X-Y$, we have
\begin{align}
P_{Y|U}(y|u)^\alpha &= \left(\sum_{x} P_{X|U}(x|u)P_{Y|X}(y|x)\right)^\alpha\\
&\le \sum_{x} P_{X|U}(x|u) P_{Y|X}(y|x)^\alpha.
\end{align}
So M$\alpha$beL may be bounded from above by
\begin{align}
&\nonumber\mathcal{L}_{\alpha,\beta}(X \to Y)\\&\nonumber\leq
\sup_{P_{X}}\sup_{U\to X\to Y} \frac{\alpha}{(\alpha-1)\beta}
\log
\sum_y P_Y(y)^{1-\beta} \\&\ \times \left(\frac{\displaystyle\sum_{u,x} P_{U}(u)^\alpha P_{X|U}(x|u) P_{Y|X}(y|x)^\alpha}{\sum_u P_U(u)^\alpha} \right)^{\frac{\beta}{\alpha}}
\\\nonumber
&\leq \sup_{P_{X}}\ \sup_{P_{\Tilde{X}}} \frac{\alpha}{(\alpha-1)\beta}\log
\sum_y P_Y(y)^{1-\beta}\\& \ \times  \left(\sum_{x} P_{\Tilde{X}}(x) P_{Y|X}(y|x)^\alpha \right)^{\frac{\beta}{\alpha}}\label{eq:upper-bound}
\end{align}
where
\begin{equation}
P_{\Tilde{X}}(x)=\frac{\sum_{u} P_{U}(u)^\alpha P_{X|U}(x|u)}{\sum_u P_U(u)^\alpha}.
\end{equation}
\textbf{Lower Bound:} The proof is based on the expression  in \eqref{eq:simplified-leakage} as well as ``shattering'' method. Consider  a random variable $U$ such that $U\to X \to Y$ form a Markov chain and $H(X|U)=0$. For each $x$, let $\mathcal{U}_{x}$ be a finite set such that  $U=u \in \mathcal{U}_{x}$ if and only if $X=x$ and $\mathcal{U}=\bigcup_{x \in \mathcal{X}} \mathcal{U}_{x}$. Moreover, given $X=x$ let $U$ be uniformly distributed on $\mathcal{U}_{x}$.  That is, 
\begin{align}
 P_{U|X}(u|x)=\begin{cases} \displaystyle \frac{1}{|\mathcal{U}_{x}|} & \text{for all}~ u \in \mathcal{U}_{x} \\ 0 & \text{otherwise,}\end{cases}   
\end{align}
and so
\begin{align}
P_{Y|U}(y|u)=\begin{cases} P_{Y|X}(y|x) & \text{for all}~ u \in \mathcal{U}_{x} \\ 0 & \text{otherwise.}\end{cases}
\end{align}
Therefore, we have
\begin{align}
   &\frac{\sum_u P_U(u)^\alpha P_{Y|U}(y|u)^\alpha}{\sum_u P_U(u)^\alpha}\\
   &=\frac{\sum_{x \in \mathcal{X}}\sum_{u\in \mathcal{U}_{x}} \left( \displaystyle\frac{P_X(x) P_{U|X}(u|x)}{P_{X|U}(x|u)}\right)^\alpha P_{Y|U}(y|u)^\alpha}{\sum_{x \in \mathcal{X}}\sum_{u\in \mathcal{U}_{x}}\left( \displaystyle\frac{P_X(x) P_{U|X}(u|x)}{P_{X|U}(x|u)}\right)^\alpha}\\
   &=\frac{\sum_{x} |\mathcal{U}_{x}|^{1-\alpha} P_{X}(x)^\alpha P_{Y|X}(y|x)^\alpha}{\sum_{x} |\mathcal{U}_x|^{1-\alpha} P_{X}(x)^\alpha}.
\end{align}
So we may bound M$\alpha$beL from below by
\begin{align}
\nonumber&\mathcal{L}_{\alpha,\beta}(X\to Y)\\\nonumber&\geq
\sup_{P_{X}} \ \sup_{\mathcal{U}_x}\frac{\alpha}{(\alpha-1)\beta}
\ \log \sum_y P_Y(y)^{1-\beta}\\& \ \times  \left(\frac{\sum_{x} |\mathcal{U}_{x}|^{1-\alpha} P_{X}(x)^\alpha P_{Y|X}(y|x)^\alpha}{\sum_{x} |\mathcal{U}_x|^{1-\alpha} P_{X}(x)^\alpha}\right)^{\frac{\beta}{\alpha}}\\\nonumber
&=\sup_{P_X}\ \sup_{P_{\Tilde{X}}} \frac{\alpha}{(\alpha-1)\beta}
\log
\sum_y P_Y(y)^{1-\beta} \\&\ \times \left(\sum_{x} P_{\Tilde{X}}(x) P_{Y|X}(y|x)^\alpha \right)^{\frac{\beta}{\alpha}}\label{eq:result-lower-finite}
\end{align}
where
\begin{equation}
P_{\Tilde{X}}(x)=\frac{|\mathcal{U}_{x}|^{1-\alpha}P_{X}(x)^\alpha}{\sum_{x} |\mathcal{U}_x|^{1-\alpha} P_{X}(x)^\alpha},
\end{equation}
and we have used the fact that any distribution $P_{\Tilde{X}}(x)$ can be reached with appropriate choice of $|\mathcal{U}_{x}|$, assuming $P_{X}(x)>0$ for all $x$; this condition can be assumed because any $P_{X}$ is arbitrarily close to a distribution with full support. Thus, combining \eqref{eq:upper-bound} and \eqref{eq:result-lower-finite}, we have
\begin{align}\label{result thm1}
\nonumber&\mathcal{L}_{\alpha,\beta}(X\to Y)=\sup_{P_X}\ \sup_{P_{\Tilde{X}}} \frac{\alpha}{(\alpha-1)\beta}\\& \ 
\log
\sum_y P_Y(y)^{1-\beta} \left(\sum_{x} P_{\Tilde{X}}(x) P_{Y|X}(y|x)^\alpha \right)^{\frac{\beta}{\alpha}}.
\end{align}
Note that the choice of $P_X$ only impacts $P_Y$, and the quantity inside the log is convex in $P_Y$. Since the supremum of a convex function is attained at an extreme point, we may simplify \eqref{result thm1} as follows.
\begin{align}
&\nonumber \mathcal{L}_{\alpha,\beta}(X\to Y)
=\max_{x'} \  \sup_{P_{\Tilde{X}}}\frac{\alpha}{(\alpha-1)\beta}\\& \ \log \sum_y P_{Y|X}(y|x')^{1-\beta} \left(\sum_{x} P_{\Tilde{X}}(x) P_{Y|X}(y|x)^\alpha \right)^{\beta/\alpha}.
\end{align}
We now obtain the expression of maximal R\'{e}nyi leakage. We first bound maximal R\'{e}nyi leakage from above as follows. 
\begin{align}
\nonumber &\mathcal{L}_{\alpha,\beta}(X\to Y)
\\\nonumber&=\max_{x'} \  \sup_{P_{\Tilde{X}}} \ \frac{\alpha}{(\alpha-1)\beta}  \  \log 
\sum_y P_{Y|X}(y|x')^{1-\beta} \\& \ \times  \left(\sum_{x} P_{\Tilde{X}}(x)\ P_{Y|X}(y|x)^\alpha \right)^{\beta/\alpha}\label{proof:max-Renyi-start}\\
&\nonumber\le 
\max_{x'} \  \sup_{P_{\Tilde{X}}} \ \frac{\alpha}{(\alpha-1)\beta}  \ \log 
\sum_y P_{Y|X}(y|x')^{1-\beta}\\& \ \times \left(\sum_{x} P_{\Tilde{X}}(x)\ \max_x P_{Y|X}(y|x)^\alpha \right)^{\beta/\alpha}\\\nonumber
&= 
\max_{x'} \  \sup_{P_{\Tilde{X}}} \ \frac{\alpha}{(\alpha-1)\beta} \ \log 
\sum_y P_{Y|X}(y|x')^{1-\beta}\\& \ \times \left( \max_x P_{Y|X}(y|x)^\alpha \sum_{x} P_{\Tilde{X}}(x) \right)^{\beta/\alpha}\\\nonumber
&= 
\frac{\alpha}{(\alpha-1)\beta} \ \max_{x'} \ \log 
\sum_y P_{Y|X}(y|x')^{1-\beta} \\& \ \times \max_x P_{Y|X}(y|x)^\beta.
\end{align}
So
\begin{align}
    &\nonumber\displaystyle\lim_{\alpha \to \infty}\mathcal{L}_{\alpha,\beta}(X\to Y)\\&\leq \frac{1}{\beta} \ \max_{x'} \ \log 
\sum_y P_{Y|X}(y|x')^{1-\beta} \ \max_x P_{Y|X}(y|x)^\beta.\label{eq:UB-Renyi-leakage}
\end{align}
We now provide an achievable scheme. We have
\begin{align}
    &\nonumber\mathcal{L}_{\alpha,\beta}(X\to Y)\\\nonumber
&=\max_{x'} \  \sup_{P_{\Tilde{X}}} \ \frac{\alpha}{(\alpha-1)\beta} \ \log 
\sum_y P_{Y|X}(y|x')^{1-\beta} \\& \ \times\left(\sum_{x} P_{\Tilde{X}}(x) \ P_{Y|X}(y|x)^\alpha \right)^{\beta/\alpha}\\
&\nonumber\geq \max_{x'} \  \sup_{P_{\Tilde{X}}} \ \frac{\alpha}{(\alpha-1)\beta} \ \log 
\sum_y P_{Y|X}(y|x')^{1-\beta} \\& \ \times\left( P_{\Tilde{X}}(x^*_y) \ P_{Y|X}(y|x^*_y)^\alpha \right)^{\beta/\alpha}\\
&\nonumber\geq \frac{\alpha}{(\alpha-1)\beta}\ \max_{x'} \ \log 
\sum_y P_{Y|X}(y|x')^{1-\beta} \\& \ \times\ |\mathcal{X}|^{-\frac{\beta}{\alpha}} \ P_{Y|X}(y|x^*_y)^\beta\\\nonumber
&\nonumber= \frac{\alpha}{(\alpha-1)\beta}\ \max_{x'} \ \Bigg[\log |\mathcal{X}|^{-\frac{\beta}{\alpha}}  \\&\ +
\log \sum_y P_{Y|X}(y|x')^{1-\beta} \ P_{Y|X}(y|x^*_y)^\beta\bigg]\\
&\nonumber=  -\frac{1}{\alpha-1} \log |\mathcal{X}|+ \frac{\alpha}{(\alpha-1)\beta} \\& \ \max_{x'}\
\log \sum_y P_{Y|X}(y|x')^{1-\beta} \ P_{Y|X}(y|x^*_y)^\beta\\
&\nonumber=  -\frac{1}{\alpha-1} \log |\mathcal{X}| + \frac{\alpha}{(\alpha-1)\beta}\\& \ \max_{x'}\
\log \sum_y P_{Y|X}(y|x')^{1-\beta} \ \max_x P_{Y|X}(y|x)^\beta
\end{align}
where $x^*_y= \displaystyle\argmax_x P_{Y|X}(y|x)$ for $y\in\mathcal{Y}$. So
\begin{align}
    &\nonumber\displaystyle\lim_{\alpha\to \infty}\mathcal{L}_{\alpha,\beta}(X\to Y)\\&\geq \frac{1}{\beta}\ \max_{x'}\
\log \sum_y P_{Y|X}(y|x')^{1-\beta} \max_x P_{Y|X}(y|x)^\beta.\label{eq:LB-Renyi-leakage}
\end{align}
Combining \eqref{eq:UB-Renyi-leakage} and \eqref{eq:LB-Renyi-leakage}, we get
\begin{align}
   \nonumber& \displaystyle\lim_{\alpha\to \infty}\mathcal{L}_{\alpha,\beta}(X\to Y)\\&= \frac{1}{\beta}\ \max_{x'}\
\log \sum_y P_{Y|X}(y|x')^{1-\beta} \max_x P_{Y|X}(y|x)^\beta.\label{proof:max-Renyi-end}
\end{align}
\section{Proof of Theorem~\ref{theorem:properties}}\label{proof:thm-properties}
\textbf{Monotonicity in $\beta$:} For $\alpha\in (1,\infty)$, $\beta_1,\beta_2 \in [1,\infty)$ and $\beta_2 > \beta_1$, consider the argument of the logarithm in \eqref{eqn:thm-alpha-beta-leakage}:
\begin{align}
& 
\sum_y P_{Y|X}(y|x')^{1-\beta_1}\left(\sum_{x} P_{\Tilde{X}}(x) P_{Y|X}(y|x)^\alpha\right)^{\frac{\beta_1}{\alpha}}\\\nonumber
 &=
 \sum_{y}P_{Y|X}(y|x')\Bigg( P_{Y|X}(y|x')^{-\alpha}\\&\ \times \sum_{x} P_{\Tilde{X}}(x
    ) P_{Y|X}(y|x)^{\alpha}\Bigg)^{\frac{\beta_2\beta_1}{\alpha\beta_2}}\\\nonumber
 &\leq \Bigg(
\sum_{y} P_{Y|X}(y|x')\bigg(P_{Y|X}(y|x')^{-\alpha}\\&\ \times \sum_{x} P_{\Tilde{X}}(x
    ) P_{Y|X}(y|x)^{\alpha}\bigg)^{\frac{\beta_2}{\alpha}}\Bigg)^{ \frac{\beta_1}{\beta_2}}\\
 &= \left(
\sum_{y} P_{Y|X}(y|x')^{1-\beta_2}\bigg(\sum_{x} P_{\Tilde{X}}(x
    ) P_{Y|X}(y|x)^{\alpha}\bigg)^{\frac{\beta_2}{\alpha}}\right)^{ \frac{\beta_1}{\beta_2}}
    \end{align}
    where the inequality results from applying Jensen's inequality to the concave function $f:\;x\to x^{p}\;(x \geq 0,\; p< 1)$. For $\alpha \in (1,\infty)$ and $\beta \in [1,\infty)$, the function $f: t \to \frac{\alpha}{(\alpha-1)\beta} \log t$ is increasing in $t>0$. Therefore, we have
    \begin{align}
&\nonumber\frac{\alpha}{(\alpha-1)\beta_1}  \log \sum_y P_{Y|X}(y|x')^{1-\beta_1} \\& \  \times\left(\sum_{x} P_{\Tilde{X}}(x)  P_{Y|X}(y|x)^\alpha \right)^{\frac{\beta_1}{\alpha}}
\\\nonumber\leq & \frac{\alpha}{(\alpha-1)\beta_2}\log\sum_{y} P_{Y|X}(y|x')^{1-\beta_2}\\&\ \times\left( \sum_{x} P_{\Tilde{X}}(x
    ) P_{Y|X}(y|x)^{\alpha}\right)^{\frac{\beta_2}{\alpha}}.
\end{align}
Taking the maximum over $x'$ and supremum over $P_{\Tilde{X}}$ completes the proof. Another way to prove this property is to consider the numerator in \eqref{eqn:alpha,beta-leakage-original-def} as the $\beta$-norm of a random variable. Since the $\beta$-norm of a random variable is non-decreasing in $\beta$, maximal $(\alpha,\beta)$-leakage is non-decreasing in $\beta$.\\
\textbf{Data processing inequalities:} Let random variables $X,Y,Z$ form a Markov chain, i.e., $X-Y-Z$.
Based on the expression of maximal $(\alpha,\beta)$-leakage in \eqref{result thm1} we first prove the post-processing inequality, that is
\begin{align}
        \mathcal{L}_{\alpha,\beta}(X\to Z) \leq \mathcal{L}_{\alpha,\beta}(X \to Y).
\end{align}
 For any $y \in \mathcal{Y}$, let
\begin{align}\label{def:g}
    g(y)=\left( \sum_{x} P_{\Tilde{X}}(x) P_{Y|X}(y|x)^{\alpha}\right)^{\frac{1}{\alpha}}
\end{align}
and
\begin{align}\label{def:c}
    c_z(y)=\displaystyle\frac{P_{Y}(y)\; P_{Z|Y}(z|y)}{P_{Z}(z)}
\end{align}
such that $\sum_y c_z(y)=1$. We have
\begin{align}
    &\sum_y P_Y(y)^{1-\beta} \left( \sum_{x} P_{\Tilde{X}}(x) P_{Y|X}(y|x)^{\alpha}\right)^{\frac{\beta}{\alpha}}
    \\&=\sum_y P_{Y}(y)^{1-\beta}g(y)^\beta
    \\&=\sum_{y,z} P_Y(y) P_{Z|Y}(z|y)\left(\frac{g(y)}{P_Y(y)}\right)^\beta
    \\&=\sum_{z} P_Z(z) \sum_y c_z(y)\left(\frac{g(y)}{P_Y(y)}\right)^\beta 
    \\&\ge \sum_z P_Z(z) \left(\sum_y c_z(y)\frac{g(y)}{P_Y(y)}\right)^\beta \label{jensens_application1}
    \\&=\sum_z P_{Z}(z)^{1-\beta} \left(\sum_y P_{Z|Y}(z|y)g(y)\right)^\beta\label{g_form}
\end{align}
where \eqref{jensens_application1} follows from applying Jensen's inequality to the convex function $f:\;x\to x^{p}\;(x \geq 0,\; p\geq 1)$.
Recalling the definition of $g(y)$ from \eqref{def:g}, we have
\begin{align}
    &\sum_y P_{Z|Y}(z|y) g(y)
    \\&=\sum_y P_{Z|Y}(z|y) \bigg( \sum_{x} P_{\Tilde{X}}(x)  P_{Y|X}(y|x)^{\alpha}\bigg)^{\frac{1}{\alpha}}
    \\&=\sum_y \bigg(\sum_{x}  \Big( P_{\Tilde{X}}(x)^{\frac{1}{\alpha}}P_{Z|Y}(z|y)  P_{Y|X}(y|x)\Big)^{\alpha}\bigg)^{\frac{1}{\alpha}}
    \\&\ge \bigg(\sum_{x}  \Big( \sum_y  P_{\Tilde{X}}(x)^{\frac{1}{\alpha}}P_{Z|Y}(z|y)  P_{Y|X}(y|x)\Big)^{\alpha}\bigg)^{\frac{1}{\alpha}}\label{ineq:norm}
    \\&=\left(\sum_{x} P_{\Tilde{X}}(x)  P_{Z|X}(z|x)^{\alpha}\right)^{\frac{1}{\alpha}}\label{applying Markov chain}
\end{align}
where 
\begin{itemize}
    \item  \eqref{ineq:norm} follows because $p$-norm satisfies the triangle inequality for $p\in (1,\infty)$,
    \item \eqref{applying Markov chain} follows because the Markov chain $X - Y - Z$ holds.
\end{itemize}
Applying \eqref{applying Markov chain} to \eqref{g_form}, and using the fact that for $\alpha \in (1,\infty)$ and $\beta \in [1,\infty)$, the function $f: t \to \frac{\alpha}{(\alpha-1)\beta} \log t$ is increasing in $t > 0$, gives
\begin{align}
   & \frac{\alpha}{(\alpha-1)\beta} \log \sum_y P_{Y}(y)^{1-\beta}\left( \sum_{x} P_{\Tilde{X}}(x) P_{Y|X}(y|x)^{\alpha}\right)^{\frac{\beta}{\alpha}}\nonumber\\ &\geq  \frac{\alpha}{(\alpha-1)\beta} \log  \sum_z P_{Z}(z)^{1-\beta} \left(\sum_{x} P_{\Tilde{X}}(x)  P_{Z|X}(z|x)^{\alpha}\right)^{\frac{\beta}{\alpha}}.
\end{align}
Taking suprema over $P_X$ and $P_{\Tilde{X}}$ completes the proof.

We now prove the linkage inequality, that is
\begin{align}
      \mathcal{L}_{\alpha,\beta}(X\to Z) \leq \mathcal{L}_{\alpha,\beta}(Y \to Z),
\end{align}
using the definition of maximal $(\alpha,\beta)$-leakage in \eqref{eqn:alpha,beta-leakage-original-def}. 
Let
\begin{align}
    \nonumber &f(P_{UZ})=\frac{\alpha}{\alpha-1}\\& \log  \frac{\displaystyle \max_{P_{\hat{U}|Z}} \left[\sum_z P_Z(z) \left(\sum_u P_{U|Z}(u|z) P_{\hat{U}|Z}(u|z)^{\frac{\alpha-1}{\alpha}}\right)^{\beta}\right]^{1/\beta}}{\displaystyle \max_{P_{\hat{U}}} \sum_u P_U(u)P_{\hat{U}}(u)^{\frac{\alpha-1}{\alpha}}}.
\end{align}
For the Markov chain $X-Y-Z$, we have
\begin{align}
  \mathcal{L}_{\alpha,\beta}(X\to Z)&=\sup_{P_{X}} \ \sup_{U\to X\to Z}  \ f(P_{UZ})
\\\label{expand_Markov}
 &= \sup_{P_{X}} \ \sup_{U\to X \to Y\to Z } f(P_{UZ})
\\
 &\le  \sup_{P_{X}} \ \sup_{U\to Y \to Z } f(P_{UZ})\\\label{P_x to P_y}
 &\le  \sup_{P_{Y}} \ \sup_{U\to Y\to Z} f(P_{UZ})\\\nonumber
 &=\mathcal{L}_{\alpha,\beta}(Y\to Z)
\end{align}
where \eqref{expand_Markov} follows because $P_{UZ}$ are the same under the Markov chains $U-X-Z$ and $U-X-Y-Z$, and \eqref{P_x to P_y} follows from the fact that a subset of all distributions $P_{Y}$ is reachable from the distribution $P_{X}$.
\\
\textbf{Non-negativity:} Consider the logarithmic term in \eqref{eqn:thm-alpha-beta-leakage}:
\begin{align}
 & \log 
\sum_y P_{Y|X}(y|x')^{1-\beta} \left(\sum_{x} P_{\Tilde{X}}(x) P_{Y|X}(y|x)^\alpha \right)^{\frac{\beta}{\alpha}}\\
\label{jensen1}
&\geq  \log
\sum_{y} P_{Y|X}(y|x')^{1-\beta}\bigg(\sum_{x}  P_{\Tilde{X}}(x) P_{Y|X}(y|x)\bigg)^{\beta}\\
&=\log
\sum_{y} P_{Y|X}(y|x')\left(\frac{\sum_{x} P_{\Tilde{X}}(x) P_{Y|X}(y|x)}{P_{Y|X}(y|x')}\right)^{\beta}\\
\label{jensen2}
&\geq \log
 \left(\sum_{y} P_{Y|X}(y|x') \ \frac{\sum_{x} P_{\Tilde{X}}(x) P_{Y|X}(y|x)}{P_{Y|X}(y|x')}\right)^{\beta}\\
&=\log
\bigg(\sum_{x,y} P_{\Tilde{X}}(x) P_{Y|X}(y|x)\bigg)^{\beta}=\log 1 =0
\end{align}
where both inequalities follow from applying Jensen's inequality to the convex function $f:\;x\to x^{p}\;(x \geq 0,\; p\geq1)$ and the fact that logarithmic functions are increasing. Equality holds in the first inequality if and only if for any $y \in \mathcal{Y}$, $P_{Y|X}(y|x)$ are the same for all $x\in \mathcal{X}$. Thus, we have
\begin{align}
    P_{Y|X}(y|x)=P_{Y}(y) \quad x\in \mathcal{X}, y \in \mathcal{Y}
\end{align}
which means $X$ and $Y$ are independent. This condition is also sufficient for equality in the second inequality. 
\\
\textbf{Additivity:}
We first prove additivity for $n=2$. We have $P_{X_1Y_1X_2Y_2}=P_{X_1Y_1}\cdot P_{X_2Y_2}$. To prove the additivity in \eqref{eqn:additivity}, using Theorem~\ref{theorem:alpha-beta-leakage} it suffices to show that
\begin{align}\label{eqn:additivityproof1}
   &\nonumber\sup_{P_{\tilde{X}_1,\tilde{X}_2}} \sum_{y_1,y_2} P_{Y_1Y_2|X_1X_2}(y_1,y_2|x_1',x_2')^{1-\beta}
   \\&\ \times\left(\sum_{x_1,x_2} P_{\Tilde{X}_1,\Tilde{X}_2}(x_1,x_2) P_{Y_1Y_2|X_1X_2}(y_1,y_2|x_1,x_2)^\alpha \right)^{\beta/\alpha}\\\nonumber
   &=\sup_{\substack{P_{\tilde{X}_i}\\ i\in{1,2}}}\prod_{i=1}^2\Bigg(\sum_{y_i} P_{Y|X}(y_i|x_i')^{1-\beta} \\& \ \times \left(\sum_{x_i} P_{\Tilde{X_i}}(x_i)
   P_{Y_i|X_i}(y_i|x_i)^\alpha \right)^{\beta/\alpha}\Bigg),
\end{align}
for every $x_1',x_2'$.
We simplify LHS in \eqref{eqn:additivityproof1} as 
\begin{align}
    &\nonumber \sup_{P_{\tilde{X}_1,\tilde{X}_2}} \sum_{y_1,y_2} P_{Y_1Y_2|X_1X_2}(y_1,y_2|x_1',x_2')^{1-\beta}
   \\&\ \times\Big(\sum_{x_1,x_2} P_{\Tilde{X}_1,\Tilde{X}_2}(x_1,x_2) P_{Y_1Y_2|X_1X_2}(y_1,y_2|x_1,x_2)^\alpha \Big)^{\beta/\alpha}\nonumber\\\nonumber
   =&\sup_{P_{\tilde{X}_1,\tilde{X}_2}} \sum_{y_1,y_2} P_{Y_1|X_1}(y_1|x_1')^{1-\beta}P_{Y_2|X_2}(y_2|x_2')^{1-\beta}\\&\Big(\sum_{x_1,x_2} P_{\Tilde{X}_1,\Tilde{X}_2}(x_1,x_2) P_{Y_1|X_1}(y_1|x_1)^\alpha P_{Y_2|X_2}(y_2|x_2)^\alpha \Big)^{\beta/\alpha}\label{eqn:additivityproof2}.
\end{align}
Let $k(y_1)=\sum_{x_1}P_{\tilde{X}_1}(x_1)P_{Y_1|X_1}(y_1|x_1)^\alpha$, for all $y_1$, so that we can define a set of probability distributions over $\mathcal{X}_1$ as
\begin{align}
    P_{\hat{X}_1}(x_1|y_1)=\frac{P_{\tilde{X}_1}(x_1)P_{Y_1|X_1}(y_1|x_1)^\alpha}{k(y_1)}. 
\end{align}
Thus, \eqref{eqn:additivityproof2} is equal to
\begin{align}
    &\nonumber\sup_{P_{\tilde{X}_1,\tilde{X}_2}} \sum_{y_1,y_2} P_{Y_1|X_1}(y_1|x_1')^{1-\beta}P_{Y_2|X_2}(y_2|x_2')^{1-\beta}\\\nonumber& \ \Big[\sum_{x_1,x_2} k(y_1)P_{\hat{X}_1|Y_1}(x_1|y_1)P_{\tilde{X}_2|\tilde{X}_1}(x_2|x_1)
   \\& \ P_{Y_2|X_2}(y_2|x_2)^\alpha \Big]^{\beta/\alpha}\\\nonumber
   &\leq\sup_{P_{\tilde{X}_1},P_{\tilde{X}_2|X_1}}\sum_{y_1}P_{Y_1|X_1}(y_1|x_1)^{1-\beta}\\\nonumber&\big(\sum_{x_1}P_{\tilde{X}_1}(x_1)
   P_{Y_1|X_1}(y_1|x_1)^\alpha\big)^{\frac{\beta}{\alpha}}\max_{\tilde{y}_1}\sum_{y_2}P_{Y_2|X_2}(y_2|x_2')^{1-\beta} \\&\Big(\sum_{x_1,x_2}P_{\hat{X}_1|Y_1}(x_1|\tilde{y}_1)P_{\tilde{X}_2|\tilde{X}_1}(x_2|x_1)P_{Y_2|X_2}(y_2|x_2)^\alpha \Big)^{\beta/\alpha}\\\nonumber
   &=\sup_{P_{\tilde{X}_1},P_{\tilde{X}_2|X_1}}\sum_{y_1}P_{Y_1|X_1}(y_1|x_1)^{1-\beta}\\&\ \left(\sum_{x_1}P_{\tilde{X}_1}(x_1) P_{Y_1|X_1}(y_1|x_1)^\alpha\right)^{\frac{\beta}{\alpha}}\sum_{y_2}P_{Y_2|X_2}(y_2|x_2')^{1-\beta}\nonumber\\
   & \ \Big(\sum_{x_1,x_2}P_{\hat{X}_1|Y_1}(x_1|{y}_1^*)P_{\tilde{X}_2|\tilde{X}_1}(x_2|x_1)P_{Y_2|X_2}(y_2|x_2)^\alpha \Big)^{\beta/\alpha}\label{eqn:additivityproof3}
\end{align}
We now define $$P_{\hat{X}_2}(x_2)=\sum_{x_1}P_{\hat{X}_1|Y_1}(x_1|y_1^*)P_{X_2|X_1}(x_2|x_1),$$ which is a probability distribution over $\mathcal{X}_2$. Then, \eqref{eqn:additivityproof3} is equal to
\begin{align}
    &\nonumber\sup_{\substack{P_{\tilde{X}_1},\\ P_{\hat{X}_2}}}\sum_{y_1}P_{Y_1|X_1}(y_1|x_1')^{1-\beta}\big(\sum_{x_1}P_{\tilde{X}_1}(x_1)P_{Y_1|X_1}(y_1|x_1)^\alpha\big)^{\frac{\beta}{\alpha}}\\& \ 
    \times \sum_{y_2}P_{Y_2|X_2}(y_2|x_2')^{1-\beta}\big(\sum_{x_2}P_{\hat{X}_2}(x_2)P_{Y_2|X_2}(y_2|x_2)^\alpha\big)^{\frac{\beta}{\alpha}}\\\nonumber
    &=\sup_{\substack{P_{\tilde{X}_i}\\ i\in{1,2}}}\prod_{i=1}^2\Big(\sum_{y_i} P_{Y|X}(y_i|x_i')^{1-\beta}\\& \hspace{5mm}\times  \big(\sum_{x_i} P_{\Tilde{X_i}}(x_i) P_{Y_i|X_i}(y_i|x_i)^\alpha \big)^{\beta/\alpha}\Big).
\end{align}
This proves \eqref{eqn:additivityproof1} as the lower bound part of \eqref{eqn:additivityproof1} is trivial. Thus we have 
\begin{align}\label{n=2case}
    \mathcal{L}_{\alpha,\beta}(X_1,X_2\!\rightarrow\! Y_1,Y_2)\!=\!\mathcal{L}_{\alpha,\beta}(X_1\rightarrow Y_1)\!+\!\mathcal{L}_{\alpha,\beta}(X_2\rightarrow Y_2).
\end{align}
Using \eqref{n=2case} twice, we have 
\begin{align}
    &\mathcal{L}_{\alpha,\beta}(X^3\rightarrow Y^3)\nonumber\\
    &=\mathcal{L}_{\alpha,\beta}(X^2\rightarrow Y^2)+\mathcal{L}_{\alpha,\beta}(X_3\rightarrow Y_3)\\
    &=\mathcal{L}_{\alpha,\beta}(X_1\rightarrow Y_1)+\mathcal{L}_{\alpha,\beta}(X_2\rightarrow Y_2)+\mathcal{L}_{\alpha,\beta}(X_3\rightarrow Y_3).
\end{align}
Similarly, by repeated application of \eqref{n=2case} $(n-1)$ times, we get
\eqref{eqn:additivity}.
\section{Proof of Theorem~\ref{theorem:conditional-alpha-beta-leakage}}\label{proof:thm-conditional-alpha-beta-leakage}
For $\alpha\in (1,\infty)$ and $\beta \in [1,\infty)$, we first bound $\mathcal{L}_{\alpha,\beta}(X \to Y|Z)$ from above and then, present an achievable scheme.\\
\textbf{Upper Bound:} Similarly to \eqref{eqn:optimization-result}, the numerator and denominator of \eqref{def:conditional_alpha-beta_leakage} become 
\begin{equation}
    \sum_{z,y} P_{Z,Y}(z,y) \left(\sum_u P_{U|Z,Y}(u|z,y)^\alpha\right)^{\beta/\alpha}
\end{equation}
and
\begin{equation}
    \sum_{z} P_{Z}(z) \left(\sum_u P_{U|Z}(u|z)^\alpha\right)^{\beta/\alpha},
\end{equation}
respectively. Thus, the logarithmic term in \eqref{def:conditional_alpha-beta_leakage} reduces to
\begin{align}
  \label{eqn:reduced-log}   &
 \log\frac{\displaystyle\sum_{z,y} P_{Z,Y}(z,y) \left(\sum_u P_{U|Z,Y}(u|z,y)^\alpha\right)^{\beta/\alpha}}
{\displaystyle\sum_{z} P_{Z}(z) \left(\sum_u P_{U|Z}(u|z)^\alpha\right)^{\beta/\alpha}}
\\&= 
 \log\frac{\displaystyle\sum_{z} P_{Z}(z) \sum_y P_{Y|Z}(y|z)  \left(\sum_u P_{U|Z,Y}(u|z,y)^\alpha\right)^{\beta/\alpha}}
{\displaystyle\sum_{z} P_{Z}(z) \left(\sum_u P_{U|Z}(u|z)^\alpha\right)^{\beta/\alpha}}
\\ \label{eqn:mediant-inequality}
&\le\log \ \max_{z} \ \frac{\displaystyle \sum_y P_{Y|Z}(y|z)  \left(\sum_u P_{U|Z,Y}(u
  |z,y)^\alpha\right)^{\beta/\alpha}}
{\displaystyle \left(\sum_u P_{U|Z}(u|z)^\alpha\right)^{\beta/\alpha}}
\\ \label{eq:reduced-log-term-upper-bound}
&=\max_{z} \
  \log\frac{\displaystyle \sum_y P_{Y|Z}(y|z) ^{1-\beta} \left(\sum_u P_{Y,U|Z}(y,u|z)^\alpha\right)^{\beta/\alpha}}
{\displaystyle  \left(\sum_u P_{U|Z}(u|z)^\alpha\right)^{\beta/\alpha}}
\end{align}
where \eqref{eqn:mediant-inequality} follows from the fact that for any non-negative $a_i$ and $b_i$, we have $\frac{\sum_i a_i}{\sum_i b_i}\leq \max_i \frac{a_i}{b_i}$. Moreover, we have
\begin{align}
 P_{Y,U|Z}(y,u|z)^\alpha  =P_{U|Z}(u|z)^\alpha P_{Y|U,Z}(y|u,z)^\alpha
\end{align}
and 
\begin{align}
    P_{Y|U,Z}(y|u,z)^\alpha
    &=\left(\sum_{x} P_{X|U,Z}(x|u,z) P_{Y|X,Z}(y|x,z) \right)^\alpha\\
    \label{eqn:conditional-jenson}
     &\le\sum_{x} P_{X|U,Z}(x|u,z) P_{Y|X,Z}(y|x,z)^\alpha
\end{align}
where the equality follows because $U-X-Y|Z$ holds, and the inequality follows from applying Jensen's inequality to the convex function $f:\;x\to x^{p}\;(x \geq 0,\; p> 1)$. Applying \eqref{eqn:conditional-jenson} to \eqref{eq:reduced-log-term-upper-bound}, we may bound $\mathcal{L}_{\alpha,\beta}(X \to Y|Z)$ from above by
  \begin{align}
\nonumber &\mathcal{L}_{\alpha,\beta}(X\to Y|Z) \\ 
&\nonumber\le\sup_{P_{X|Z}} \ \sup_{U\to X\to Y|Z} 
   \max_{z} \ \frac{\alpha}{(\alpha-1)\beta}\ \log \displaystyle \sum_y P_{Y|Z}(y|z) ^{1-\beta}\\& \
 \left(\frac{ \displaystyle\sum_{x} P_{Y|X,Z}(y|x,z)^\alpha  \sum_{u} P_{U|Z}(u|z)^\alpha P_{X|U,Z}(x|u,z) }
{\displaystyle \sum_u P_{U|Z}(u|z)^\alpha}\right)^{\frac{\beta}{\alpha}}\\\nonumber 
& \le\max_{z}  \sup_{P_{X|Z=z}} \sup_{\substack{P_{\Tilde{X}|Z=z} \\\ll P_{X|Z=z}}} \ \frac{\alpha}{(\alpha-1)\beta}
 \log \displaystyle \sum_y P_{Y|Z}(y|z) ^{1-\beta}\\& \ \left( \displaystyle\sum_{x} P_{Y|X,Z}(y|x,z)^\alpha P_{\Tilde{X}|Z=z}(x)\right)^{\frac{\beta}{\alpha}}\label{eqn:upper_bound-result}
\end{align}
where
\begin{equation}
      P_{\Tilde{X}|Z=z}(x)=\frac{ \displaystyle \sum_{u} P_{U|Z}(u|z)^\alpha P_{X|U,Z}(x|u,z) }
{\displaystyle \sum_u P_{U|Z}(u|z)^\alpha}.
\end{equation}
\textbf{Lower Bound:} For this proof, we use the expression in \eqref{eqn:reduced-log} as well as shattering method.  For a given conditional distribution $P_{Y|X,Z}$, let 
\begin{align}\label{def:z^*}
  \nonumber z^*&=\argmax_{z}  \sup_{P_{X|Z=z}}  \sup_{\substack{P_{\Tilde{X}|Z=z}\\\ll P_{X|Z=z}}}
\displaystyle \sum_y P_{Y|Z}(y|z) ^{1-\beta}\\& \ \left( \displaystyle\sum_{x} P_{\Tilde{X}|Z=z}(x) P_{Y|X,Z}(y|x,z)^\alpha \right)^{\frac{\beta}{\alpha}}
\end{align}
and 
\begin{align}
    \mathcal{X}_{z^*}=\{x \in \mathcal{X}: P_{X,Z}(x,z^*)>0\}.
\end{align}
Consider  a random variable $U$ whose alphabet consists of several disjoint subsets. 
For all $x \in \mathcal{X}_{z^*}$, let $\mathcal{U}_{x,z^*}$ be disjoint, finite sets. Moreover, let $\mathcal{U}_0$ be a finite set (disjoint from those above) such that $\mathcal{U}=\mathcal{U}_0 \cup \bigcup_{x\in \mathcal{X}_{z^*}}\mathcal{U}_{x,z^*}$.
We now define the conditional distribution $P_{U|X,Z}$ as 
\begin{align}
 P_{U|X,Z}(u|x,z)=\begin{cases} \displaystyle \frac{1}{|\mathcal{U}_{x,z^*}|}, & z=z^*, u \in \mathcal{U}_{x,z^*} \\
 \frac{1}{|\mathcal{U}_0|}, & z \ne z^*, u \in \mathcal{U}_0
 \\0, & \text{otherwise.}\end{cases}   
\end{align}
So the numerator of \eqref{eqn:reduced-log} reduces to
\begin{align}
 & \displaystyle\sum_{z,y} P_{Z,Y}(z,y) \left(\sum_u P_{U|Z,Y}(u|z,y)^\alpha\right)^{\beta/\alpha}\\\nonumber
  \label{eqn:total-term}
   & =\displaystyle\sum_{z\ne z^*,y} P_{Z,Y}(z,y) \left(\sum_u P_{U|Z,Y}(u|z,y)^\alpha\right)^{\beta/\alpha}\\& \ + P_Z(z^*)\displaystyle\sum_{y}  P_{Y|Z}(y|z^*)^{1-\beta} \left(\sum_u P_{U,Y|Z}(u,y|z^*)^\alpha\right)^{\beta/\alpha} .
\end{align}
where
\begin{align}
&\left(\sum_u P_{U|Z,Y}(u|z,y)^\alpha\right)^{\beta/\alpha}\\
\label{eqn:Markov-lower-bound1}
&= \left(\sum_u\left(\sum_x P_{X|Z,Y}(x|z,y)P_{U|Z,X}(u|z,x)\right)^\alpha\right)^{\beta/\alpha}\\
&= \left(|\mathcal{U}_0|\left(\sum_x P_{X|Z,Y}(x|z,y) \ \frac{1}{|\mathcal{U}_0|}\right)^\alpha\right)^{\beta/\alpha}\\ 
\label{eqn:reduced-term1}
&= \left(|\mathcal{U}_0|^{1-\alpha}\right)^{\beta/\alpha}
\end{align}
and
\begin{align}
& \left(\sum_u P_{U,Y|Z}(u,y|z^*)^\alpha\right)^{\beta/\alpha} \\\nonumber
\label{eqn:Markov-lower-bound2}
&= \bigg(\sum_u\bigg[\sum_{x'} P_{X|Z}(x'|z^*)\ \\& \hspace{3mm}\times P_{U|X,Z}(u|x',z^*)  P_{Y|X,Z}(y|x',z^*) \bigg]^\alpha\bigg)^{\beta/\alpha} 
\\\nonumber 
%
&= \bigg(\sum_x \displaystyle \sum_{u \in \mathcal{U}_{x,z^*}}\bigg[\sum_{x'} P_{X|Z}(x'|z^*)\ \\&\hspace{3mm}\times P_{U|X,Z}(u|x',z^*)  P_{Y|X,Z}(y|x',z^*) \bigg]^\alpha\bigg)^{\beta/\alpha} 
\\\nonumber 
&= \Bigg(\sum_x |\mathcal{U}_{x,z^*}|\Bigg[ P_{X|Z}(x|z^*)\\&\hspace{3mm}\times \frac{1}{|\mathcal{U}_{x,z^*}|} \ P_{Y|X,Z}(y|x,z^*) \Bigg]^\alpha\Bigg)^{\frac{\beta}{\alpha}} 
\\
\label{eqn:reduced-term2}
&= \left(\sum_x |\mathcal{U}_{x,z^*}|^{1-\alpha} P_{X|Z}(x|z^*)^\alpha\  P_{Y|X,Z}(y|x,z^*)^\alpha\right)^{\beta/\alpha}. 
\end{align}
Here \eqref{eqn:Markov-lower-bound1} and \eqref{eqn:Markov-lower-bound2} follow because $U-X-Y|Z$ holds. Applying \eqref{eqn:reduced-term1} and \eqref{eqn:reduced-term2} to \eqref{eqn:total-term}, the numerator of \eqref{eqn:reduced-log} becomes
\begin{align}
  &\nonumber\frac{1-P_Z(z^*)}{|\mathcal{U}_0|^{(1-\frac{1}{\alpha})\beta}}+P_Z(z^*)\displaystyle\sum_{y} P_{Y|Z}(y|z^*)^{1-\beta}\\& \ \left(\sum_x |\mathcal{U}_{x,z^*}|^{1-\alpha} P_{X|Z}(x|z^*)^\alpha\  P_{Y|X,Z}(y|x,z^*)^\alpha\right)^{\beta/\alpha}.
\end{align}
Similarly, the denominator of \eqref{eqn:reduced-log} becomes
\begin{align}
    \frac{1-P_Z(z^*)}{|\mathcal{U}_0|^{(1-\frac{1}{\alpha})\beta}}+P_Z(z^*)\left( \displaystyle \sum_x |\mathcal{U}_{x,z^*}|^{1-\alpha} P_{X|Z}(x|z^*)^\alpha\right)^{\beta/\alpha}.
\end{align}
Note that for $\alpha\in (1,\infty)$ and $\beta\in[1,\infty)$, $\displaystyle\frac{1-P_Z(z^*)}{|\mathcal{U}_0|^{(1-\frac{1}{\alpha})\beta}} \to 0 $  as $|\mathcal{U}_0| \to \infty$. So we may bound the conditional maximal $(\alpha,\beta)$-leakage from below by
\begin{align}
  & \nonumber \mathcal{L}_{\alpha,\beta}(X\to Y|Z) 
   \\\nonumber  & \geq\sup_{P_{X|Z=z^*}} \ \sup_{\mathcal{U}_{x,z^*}}\frac{\alpha}{(\alpha-1)\beta} \ \log \displaystyle\sum_{y} P_{Y|Z}(y|z^*)^{1-\beta} \\& \ \left( \displaystyle \sum_x P_{Y|X,Z}(y|x,z^*)^\alpha \ \frac{ |\mathcal{U}_{x,z^*}|^{1-\alpha} P_{X|Z}(x|z^*)^\alpha\  }{ \displaystyle \sum_x |\mathcal{U}_{x,z^*}|^{1-\alpha} P_{X|Z}(x|z^*)^\alpha}\right)^{\frac{\beta}{\alpha}}
   \\\nonumber   & =\sup_{P_{X|Z=z^*}} \ \sup_{P_{\Tilde{X}|Z=z^*}\ll P_{X|Z=z^*}}\frac{\alpha}{(\alpha-1)\beta} \\\nonumber & \ \log \displaystyle\sum_{y} P_{Y|Z}(y|z^*)^{1-\beta}  \Bigg( \displaystyle \sum_x P_{Y|X,Z}(y|x,z^*)^\alpha \\ & \ \times  P_{\Tilde{X}|Z=z^*}(x)\Bigg)^{\frac{\beta}{\alpha}},\label{eqn:lower-bound-p-tilde-x}
\end{align}
where here
\begin{align}
    P_{\Tilde{X}|Z=z^*}(x)= \frac{ |\mathcal{U}_{x,z^*}|^{1-\alpha} P_{X|Z}(x|z^*)^\alpha\  }{ \displaystyle \sum_x |\mathcal{U}_{x,z^*}|^{1-\alpha} P_{X|Z}(x|z^*)^\alpha}
\end{align}
and we have used the fact that any distribution $P_{\Tilde{X}|Z=z^*}(x)$ can be reached with appropriate choice of $|\mathcal{U}_{x,z^*}|$. Recalling the definition of $z^*$ from \eqref{def:z^*}, we may re-write \eqref{eqn:lower-bound-p-tilde-x} as
\begin{align}\label{eqn:lower-bound-result}
    \nonumber&\max_z  \sup_{P_{X|Z=z}}  \sup_{\substack{P_{\Tilde{X}|Z=z}\\\ll P_{X|Z=z}}}\frac{\alpha}{(\alpha-1)\beta} \ \log  \displaystyle\sum_{y} P_{Y|Z}(y|z)^{1-\beta} \\& \ \left( \displaystyle \sum_x P_{Y|X,Z}(y|x,z)^\alpha \ P_{\Tilde{X}|Z=z}(x)\right)^{\frac{\beta}{\alpha}}
\end{align}
Therefore, combining \eqref{eqn:upper_bound-result} and \eqref{eqn:lower-bound-result}, we have 
\begin{align}\label{eqn:combining-bounds}
  &\nonumber\mathcal{L}_{\alpha,\beta}(X \to Y|Z)\\\nonumber&=\max_z \ \sup_{P_{X|Z=z}} \ \sup_{\substack{P_{\Tilde{X}|Z=z}\ll \\P_{X|Z=z}}}\
  \frac{\alpha}{(\alpha-1)\beta} \ \log  \displaystyle \sum_y P_{Y|Z}(y|z)^{1-\beta}\\& \ \left( \displaystyle\sum_x P_{Y|X,Z}(y|x,z)^\alpha P_{\Tilde{X}|Z=z}(x)\right)^{\frac{\beta}{\alpha}}.
  \end{align}  
Since the choice of $P_{X|Z=z}$ only impacts $P_{Y|Z}$, and the supremum of a convex function is attained at an extreme point, we may simplify \eqref{eqn:combining-bounds} as follows.
\begin{align}
  &\nonumber\max_{z} \ \max_{x'} \ \sup_{\substack{P_{\Tilde{X}|Z=z}\ll \\P_{X|Z=z}}}  \
  \frac{\alpha}{(\alpha-1)\beta} \ \log \displaystyle \sum_y P_{Y|X,Z}(y|x',z)^{1-\beta}\\&\ \left( \displaystyle\sum_x P_{Y|X,Z}(y|x,z)^\alpha P_{\Tilde{X}|Z=z}(x)\right)^{\frac{\beta}{\alpha}}.
\end{align}
\section{Proof of Theorem~\ref{theorem:conditioning-reduces-leakage}}\label{proof:thm-conditioning-reduces-leakage}
From \eqref{eqn:thm-conditional-alpha-beta-leakage}, we have 
\begin{align}
   \nonumber &\mathcal{L}_{\alpha,\beta}(X \to Y|Z)\\\nonumber&=\max_{z} \ \max_{x'} \ \sup_{\substack{P_{\bar{X}}\ll\\ P_{X|Z=z}}} \ \frac{\alpha}{(\alpha-1)\beta} \ \log \displaystyle \sum_y P_{Y|X,Z}(y|x',z)^{1-\beta}\\& \ \left( \displaystyle\sum_x P_{Y|X,Z}(y|x,z)^\alpha P_{\bar{X}}(x)\right)^{\frac{\beta}{\alpha}}\\\nonumber
  &= \max_{z} \ \max_{x'} \ \sup_{\substack{P_{\bar{X}}\ll\\ P_{X|Z=z}}}  \
  \frac{\alpha}{(\alpha-1)\beta} \ \log \displaystyle \sum_y P_{Y|X}(y|x')^{1-\beta}\\& \ \left( \displaystyle\sum_x P_{Y|X}(y|x)^\alpha P_{\bar{X}}(x)\right)^{\frac{\beta}{\alpha}}\label{eqn:applying-Markov-conditioning}\\\nonumber
  &\le \max_{x'} \ \sup_{P_{\bar{X}}\ll P_X}  \
  \frac{\alpha}{(\alpha-1)\beta} \ \log \displaystyle \sum_y P_{Y|X}(y|x')^{1-\beta}\\& \ \left( \displaystyle\sum_x P_{Y|X}(y|x)^\alpha P_{\bar{X}}(x)\right)^{\frac{\beta}{\alpha}}\label{eqn:conditional-support}\\\nonumber
  &=\mathcal{L}_{\alpha,\beta}(X \to Y),
 \end{align}
 where \eqref{eqn:applying-Markov-conditioning} follows because the Markov chain $Z-X-Y$ holds, and \eqref{eqn:conditional-support} follows from the fact that for any $z$, the support of $P_{X|Z=z}$ is a subset of the support of $P_{X}$. The equality is achieved if for some $z \in \text{supp}(Z)$, $\text{supp}(X)=\text{supp}(X|Z=z)$.
 \section{Proof of Theorem~\ref{theorem:sub-additivity}}\label{proof:thm-sub-additivity}
We have
\begin{align}\label{eqn:x-to-y1y2}
\nonumber&\mathcal{L}_{\alpha,\beta}(X\to Y_1,Y_2|Z)\\\nonumber&=\max_{z,x'} \  \sup_{P_{\Tilde{X}|Z}}\ \frac{\alpha}{(\alpha-1)\beta} \ \log \sum_{y_1,y_2} 
 P_{Y_1,Y_2|X,Z}(y_1,y_2|x',z)^{1-\beta}\\& \  \left(\sum_{x} P_{\Tilde{X}|Z}(x|z) P_{Y_1,Y_2|X,Z}(y_1,y_2|x,z)^\alpha \right)^{\beta/\alpha}.
\end{align}
We reduce the argument of the logarithm in \eqref{eqn:x-to-y1y2} as follows.
\begin{align}
&\nonumber\sum_{y_1,y_2} P_{Y_1,Y_2|X,Z}(y_1,y_2|x',z)^{1-\beta}\\& \  \left(\sum_{x} P_{\Tilde{X}|Z}(x|z)\ P_{Y_1,Y_2|X,Z}(y_1,y_2|x,z)^\alpha \right)^{\beta/\alpha}\\\nonumber
&=\sum_{y_1,y_2} P_{Y_1|X,Z}(y_1|x',z)^{1-\beta} P_{Y_2|Y_1,X,Z}(y_2|y_1,x',z)^{1-\beta}\\\nonumber & \ \Bigg(\sum_{x} P_{\Tilde{X}|Z}(x|z) P_{Y_1|X,Z}(y_1|x,z)^\alpha \\ & \hspace{3mm} \times P_{Y_2|Y_1,X,Z}(y_2|y_1,x,z)^\alpha \Bigg)^{\beta/\alpha}.\label{eqn:composition-chain-rule}
\end{align}
Let $K(y_1,z)=\displaystyle \sum_x P_{\Tilde{X}|Z}(x|z) P_{Y_1|X,Z}(y_1|x,z)^\alpha$, for all $y_1 \in \mathcal{Y}_1$ and $z\in\mathcal{Z}$. So, we can construct a set of distributions over $\mathcal{X}$ as
\begin{equation}\label{def:x-bar}
    P_{\bar{X}|Y_1,Z}(x|y_1,z)=\frac{P_{\Tilde{X}|Z}(x|z) P_{Y_1|X,Z}(y_1|x,z)^\alpha}{K(y_1,z)}.
\end{equation}
Thus, we may rewrite the expression in \eqref{eqn:composition-chain-rule} as
\begin{align}
&\nonumber\sum_{y_1,y_2} P_{Y_1|X,Z}(y_1|x',z)^{1-\beta} P_{Y_2|Y_1,X,Z}(y_2|y_1,x',z)^{1-\beta}\\\nonumber& \ \bigg(\sum_{x} P_{\bar{X}|Y_1,Z}(x|y_1,z) \ K(y_1,z)\ \\&\  \times P_{Y_2|Y_1,X,Z}(y_2|y_1,x,z)^\alpha \Bigg)^{\frac{\beta}{\alpha}}\\\nonumber
&=\sum_{y_1,y_2}\Bigg[ P_{Y_1|X,Z}(y_1|x',z)^{1-\beta} P_{Y_2|Y_1,X,Z}(y_2|y_1,x',z)^{1-\beta}\\\nonumber& \ \times  \bigg(\displaystyle\sum_x P_{\Tilde{X}|Z}(x|z)P_{Y_1|X,Z}(y_1|x,z)^\alpha\bigg)^{\frac{\beta}{\alpha}} \\& \ \times \bigg(\sum_{x} P_{\bar{X}|Y_1,Z}(x|y_1,z) \  P_{Y_2|Y_1,X,Z}(y_2|y_1,x,z)^\alpha \bigg)^{\frac{\beta}{\alpha}}\Bigg]\\\nonumber
&=\sum_{y_1}\Bigg[ P_{Y_1|X,Z}(y_1|x',z)^{1-\beta}  \Bigg(\displaystyle \sum_x P_{\Tilde{X}|Z}(x|z) \\\nonumber& \ \times  P_{Y_1|X,Z}(y_1|x,z)^\alpha\bigg)^{\frac{\beta}{\alpha}}  \sum_{y_2} P_{Y_2|Y_1,X,Z}(y_2|y_1,x',z)^{1-\beta} \\& \ \bigg(\sum_{x} P_{\bar{X}|Y_1,Z}(x|y_1,z)\ P_{Y_2|Y_1,X,Z}(y_2|y_1,x,z)^\alpha
\bigg)^{\frac{\beta}{\alpha}}\Bigg]\\\nonumber
&\le \bigg[\sum_{y_1} P_{Y_1|X,Z}(y_1|x',z)^{1-\beta}  \ \Big(\displaystyle \sum_x P_{\Tilde{X}|Z}(x|z) \\\nonumber&\ P_{Y_1|X,Z}(y_1|x,z)^\alpha\Big)^{\frac{\beta}{\alpha}} \bigg] \bigg[\max_{y_1'\in\mathcal{Y}}\sum_{y_2} P_{Y_2|Y_1,X,Z}(y_2|y_1',x',z)^{1-\beta} \\&\  \bigg(\sum_{x} P_{\bar{X}|Y_1,Z}(x|y_1',z)  P_{Y_2|Y_1,X,Z}(y_2|y_1',x,z)^\alpha
\bigg)^{\frac{\beta}{\alpha}}\bigg].\label{eqn:reduced-argument-composition}
\end{align}
Applying \eqref{eqn:reduced-argument-composition} to \eqref{eqn:x-to-y1y2}, we have
\begin{align}
   \nonumber & \mathcal{L}_{\alpha,\beta}(X\to Y_1,Y_2|Z)\\\nonumber 
    &\nonumber\le \max_{z,x'} \  \sup_{P_{\Tilde{X}|Z}}\frac{\alpha}{(\alpha-1)\beta} \bigg[\log \sum_{y_1} P_{Y_1|X,Z}(y_1|x',z)^{1-\beta}\\\nonumber& \ \times \Big(\displaystyle \sum_x P_{\Tilde{X}|Z}(x|z) P_{Y_1|X,Z}(y_1|x,z)^\alpha\Big)^{\frac{\beta}{\alpha}}\\\nonumber&\ +\max_{y_1'} \ \log \sum_{y_2}P_{Y_2|Y_1,X,Z}(y_2|y_1',x',z)^{1-\beta}\\&\ \Big(\sum_{x} P_{\bar{X}|Y_1,Z}(x|y_1',z) P_{Y_2|Y_1,X,Z}(y_2|y_1',x,z)^\alpha
\Big)^{\frac{\beta}{\alpha}}\bigg]\\\nonumber
 &\le\max_{z,x'} \  \sup_{P_{\Tilde{X}|Z}}\frac{\alpha}{(\alpha-1)\beta}\log \sum_{y_1} P_{Y_1|X,Z}(y_1|x',z)^{1-\beta} \\\nonumber& \ \bigg( \sum_x P_{\Tilde{X}|Z}(x|z) P_{Y_1|X,Z}(y_1|x,z)^\alpha\bigg)^{\frac{\beta}{\alpha}}+ \max_{y_1',z,x'}   \sup_{P_{\bar{X}|Y_1,Z}}\\\nonumber&\ \frac{\alpha}{(\alpha-1)\beta} \log  \sum_{y_2} P_{Y_2|Y_1,X,Z}(y_2|y_1',x',z)^{1-\beta} \\& \bigg(\sum_{x} P_{\bar{X}|Y_1,Z}(x|y_1',z) P_{Y_2|Y_1,X,Z}(y_2|y_1',x,z)^\alpha
\bigg)^{\frac{\beta}{\alpha}}\\
  &=\mathcal{L}_{\alpha,\beta}(X\to Y_1|Z)+\mathcal{L}_{\alpha,\beta}(X\to Y_2|Y_1,Z).
\end{align}
\section{Proof of Proposition~\ref{prop1}}\label{App:prop1-proof}
For $\alpha\leq \beta$, maximal $(\alpha,\beta)$-leakage simplifies to
\begin{align}
&\nonumber\mathcal{L}_{\alpha,\beta}(X\to Y)\\\nonumber&=\max_{x'} \  \sup_{P_{\Tilde{X}}}\frac{\alpha}{(\alpha-1)\beta}\log 
\sum_y P_{Y|X}(y|x')^{1-\beta} \\& \ \left(\sum_{x} P_{\Tilde{X}}(x) P_{Y|X}(y|x)^\alpha \right)^{\beta/\alpha}\label{eq:convexity}
\\
&=\max_{x'} \ \max_x \frac{\alpha}{(\alpha-1)\beta}
 \log \sum_y P_{Y|X}(y|x')^{1-\beta} P_{Y|X}(y|x)^\beta, \label{eq:simplified-2}
\end{align}
where \eqref{eq:simplified-2} follows because the argument of the logarithm in \eqref{eq:convexity} is convex in $P_{\Tilde{X}}$ and so the supremum is attained at an extreme point. This quantity represents a \emph{scaled version of LRDP of order $\beta$} which is exactly equal to LRDP for $\alpha=\beta$. We now take the limits of maximal $\alpha,\beta$-leakage in \eqref{eqn:thm-alpha-beta-leakage} as $\alpha \to \infty$ and $\beta\to\infty$. We have
\begin{align}
 \lim_{\beta\to \infty}\lim_{\alpha \to \infty}  \mathcal{L}_{\alpha,\beta}(X\to Y)=
 \lim_{\beta\to \infty} \mathcal{L}_\beta(X \to Y),
\end{align}
We first bound $\displaystyle\lim_{\beta\to \infty} \mathcal{L}_\beta(X \to Y)$ from above as follows.
\begin{align}
  &\nonumber\mathcal{L}_\beta(X \to Y)\\&= \max_{x'} \ \frac{1}{\beta}  \log \sum_y P_{Y|X}(y|x')^{1-\beta} \max_x P_{Y|X}(y|x)^\beta \\
  &= \max_{x'} \ \frac{1}{\beta}  \log \sum_y P_{Y|X}(y|x') \left(\frac{\displaystyle\max_x P_{Y|X}(y|x)}{P_{Y|X}(y|x')}\right) ^\beta \\
    &\leq \max_{x'} \ \frac{1}{\beta}  \log \sum_y P_{Y|X}(y|x') \,\max_y\left(\frac{\displaystyle\max_x P_{Y|X}(y|x)}{P_{Y|X}(y|x')}\right) ^\beta \\
    &= \max_{x'} \ \frac{1}{\beta}\log \,\max_y\left(\frac{\displaystyle\max_x P_{Y|X}(y|x)}{P_{Y|X}(y|x')}\right) ^\beta   \sum_y P_{Y|X}(y|x') \\
    &= \max_{x',y} \,\log \,\left(\frac{\displaystyle\max_x P_{Y|X}(y|x)}{P_{Y|X}(y|x')}\right)\\
    &= \max_{x',y,x} \log \,\left(\frac{ P_{Y|X}(y|x)}{P_{Y|X}(y|x')}\right) \label{eq:UP-LDP}
\end{align}
   %
   %
   %
As the next step, we provide an achievable scheme. Let $x^*,y^*,x'^*=\displaystyle\argmax_{x,y,x'} \frac{P_{Y|X}(y|x)}{P_{Y|X}(y|x')}$, we have
\begin{align}
   &\nonumber\mathcal{L}_\beta(X \to Y)\\&= \max_{x'} \ \frac{1}{\beta}  \log \sum_y P_{Y|X}(y|x')^{1-\beta} \max_x P_{Y|X}(y|x)^\beta \\
   &\geq \max_{x'} \ \frac{1}{\beta}  \log \ \max_y\left[ P_{Y|X}(y|x')^{1-\beta} \max_x P_{Y|X}(y|x)^\beta\right] \\
   &=  \ \frac{1}{\beta}  \log \ \max_{y,x,x'} \ \left(\frac{P_{Y|X}(y|x)}{P_{Y|X}(y|x')} \right)^\beta
   P_{Y|X}(y|x') \\
   &\geq  \ \frac{1}{\beta}  \log  \ \left(\frac{P_{Y|X}(y^*|x^*)}{P_{Y|X}(y^*|x'^*)} \right)^\beta
   P_{Y|X}(y^*|x'^*) \\
    &=   \log  \ \frac{P_{Y|X}(y^*|x^*)}{P_{Y|X}(y^*|x'^*)}
   + \frac{1}{\beta} \log P_{Y|X}(y^*|x'^*),
\end{align}
So
\begin{align}
   \nonumber \lim_{\beta \to \infty} \mathcal{L}_{\beta}(X\to Y)&\geq \log  \ \frac{P_{Y|X}(y^*|x^*)}{P_{Y|X}(y^*|x'^*)}\\&=\max_{x,y,x'} \log  \ \frac{P_{Y|X}(y|x)}{P_{Y|X}(y|x')}.\label{eq:LP-LDP}
\end{align}
Combining \eqref{eq:UP-LDP} and \eqref{eq:LP-LDP} gives
\begin{align}
    \lim_{\beta\to \infty} \mathcal{L}_\beta(X \to Y) = \max_{x,y,x'} \log  \ \frac{P_{Y|X}(y|x)}{P_{Y|X}(y|x')},
\end{align}
which is LDP.
\section{Proof of Proposition~\ref{vector-leakage-cases}}\label{proof:vector-cases}
Applying Theorem~\ref{theorem:conditional-alpha-beta-leakage}, for finite alphabets, vector maximal $(\alpha,\beta)$-leakage defined in \eqref{def:diff-alpha-beta-leakage} simplifies to
       \begin{align}
 \nonumber&\displaystyle \mathcal{L}_{\alpha,\beta}^{\text{vec}}(X^n\to Y)= \max_{i,x_{-i},x_i'} \sup_{P_{\Tilde{X}_i|X_{-i}}}
  \frac{\alpha}{(\alpha-1)\beta} \\\nonumber& \ \log  \displaystyle \sum_y P_{Y|X_i,X_{-i}}(y|x'_i,x_{-i})^{1-\beta}\\& \ \left( \displaystyle\sum_{x_i} P_{Y|X_i,X_{-i}}(y|x_i,x_{-i})^\alpha P_{\Tilde{X}_i|X_{-i}}(x_i|x_{-i})\right)^{\frac{\beta}{\alpha}},\label{eq:vec-simplified}
    \end{align}
    where  $P_{\Tilde{X}_i|X_{-i}}$ is a distribution on the support of $P_{X_i|X_{-i}}$. The expression of vector maximal $\alpha$-leakage in \eqref{eq:vec-max-alpha-leak} can be readily obtained by setting $\beta$ equal to 1. For $\alpha\leq \beta$, applying Proposition~\ref{prop1-conditional}, we can simplify the above expression to  \begin{align}
 \displaystyle \nonumber&\mathcal{L}_{\alpha\leq\beta}^{\text{vec}}(X^n\to Y)= \max_{i,x_{-i},x_i',x_i}\  \
  \frac{\alpha}{(\alpha-1)\beta}\\& \ \log \displaystyle \sum_y P_{Y|X_i,X_{-i}}(y|x_i',x_{-i})^{1-\beta} P_{Y|X_i,X_{-i}}(y|x_i,x_{-i})^\beta.
 \end{align}
 Note that $(x'_i,x_{-i})$ and $(x_i,x_{-i})$ can be considered as two datasets which differ only in the $i$th entries. So, the maximums over $x_i, x'_i$, and $x_{-i}$ explore neighboring datasets differing in $i$th entries, and the maximum across all $i$ ensures the consideration of all possible neighboring datasets. Thus, the above expression may be rewritten to the following form:
  \begin{align}      
  \max_{x^n\sim x'^n}\frac{\alpha}{(\alpha-1)\beta} \log \displaystyle \sum_y P_{Y|X^n}(y|x'^n)^{1-\beta} P_{Y|X^n}(y|x^n)^\beta,
\end{align}
which is a scaled RDP of order $\beta$ with the scaling factor of $\displaystyle\frac{\alpha (\beta-1)}{(\alpha-1)\beta}$. Moreover, when $\alpha=\beta$, this quantity is exactly equal to RDP of order $\alpha=\beta$ which in turn recovers DP as $\alpha=\beta\to \infty$. For $\alpha\to\infty$ and an arbitrary $\beta$, applying Proposition~\ref{prop1-conditional}, we get the expression of vector maximal R\'enyi leakage.
\section{Proof of Lemma~\ref{lemma:alpha_tau_variational}}\label{proof:remark-reparameterization}
We first prove the expression \eqref{alpha_tau_variational} which provides a still other representation of the leakage measure. 

Consider any $\gamma\in(-\infty,0]\cup [1,\infty)$, and any constants $C(y)$ for $y\in\mathcal{Y}$. Furthermore, consider the optimization problem
\begin{equation}\label{general_QY_opt}
\inf_{Q_Y} \sum_y C(y) Q_Y(y)^\gamma.
\end{equation}
$\gamma$ is in the range where \eqref{general_QY_opt} is convex in $Q_Y$, so it is solved by setting the derivative of $Q_Y(y)$ to a constant:
\begin{align}
\nu=\frac{\partial}{\partial Q_Y(y)} \sum_y C(y) Q_Y(y)^\gamma=C(y) \ \gamma \ Q_Y(y)^{\gamma-1}.
\end{align}
We can see that the optimal choice is therefore
\begin{align}
Q_Y(y)=\frac{C(y)^{1/(1-\gamma)}}{\sum_{y'} C(y')^{1/(1-\gamma)}}.
\end{align}
Thus \eqref{general_QY_opt} becomes
\begin{align}\label{optimized_expression}
\frac{\sum_y C(y) C(y)^{\gamma/(1-\gamma)}}{\left(\sum_{y'} C(y')^{1/(1-\gamma)}\right)^\gamma}
=\left(\sum_y C(y)^{1/(1-\gamma)}\right)^{1-\gamma}.
\end{align}
In our case, we have $\gamma=\frac{1-\alpha}{\tau}<0$, and
\begin{align}
C(y)=\sum_x P_{\Tilde{X}}(x) P_{Y|X}(y|x)^\alpha P_{Y|X}(y|x')^{(1-\frac{1}{\tau})(1-\alpha)}.
\end{align}
Applying the result in \eqref{optimized_expression} to our case, we find that \eqref{alpha_tau_variational} is equal to
\begin{align}
&\nonumber\max_{x'}\,\sup_{P_{\Tilde{X}}}\frac{1}{\alpha-1}\log \bigg[\sum_y \bigg(\sum_x P_{\Tilde{X}}(x) P_{Y|X}(y|x)^\alpha \\& \ \times P_{Y|X}(y|x')^{(1-\frac{1}{\tau})(1-\alpha)}\bigg)^{\frac{\tau}{\tau+\alpha-1}}\bigg]^{\frac{\tau+\alpha-1}{\tau}}
\\\nonumber&=\max_{x'}\,\sup_{P_{\Tilde{X}}}\left(\frac{1}{\alpha-1}+\frac{1}{\tau}\right)
\log \sum_y P_{Y|X}(y|x')^{\frac{(\tau-1)(1-\alpha)}{\tau+\alpha-1}}\\& \ \times \left(\sum_x  P_{\Tilde{X}}(x)P_{Y|X}(y|x)^\alpha\right)^{\frac{\tau}{\tau+\alpha-1}}
\end{align}
which is precisely \eqref{alpha_tau_version}.

Moreover, we claim that maximal $(\alpha,\tau)$-leakage is non-decreasing in $\tau$  and $\alpha$ for a fixed $\alpha$   and $\tau$, respectively. Since $\beta$ is increasing in $\tau$, the first claim, that the measure is non-decreasing in $\tau$, is equivalent to it being non-decreasing in $\beta$, which we have already proved in Appendix~\ref{proof:thm-properties}. Given the expression \eqref{alpha_tau_variational}, we prove that $\mathcal{L}_{\alpha,\tau}(X\to Y)$ is non-decreasing in $\alpha$ for a fixed $\tau$ as follows. We may write the objective function in \eqref{alpha_tau_variational} as
\begin{align}
&\nonumber\frac{1}{\alpha-1} \log \sum_{x,y} P_{\Tilde{X}}(x) P_{Y|X}(y|x) \\& \ \times \left(\frac{P_{Y|X}(y|x)}{Q_Y(y)^{\frac{1}{\tau}} P_{Y|X}(y|x')^{1-\frac{1}{\tau}}}\right)^{\alpha-1}.
\end{align}
This expression is non-decreasing in $\alpha$ due to the fact that, for any distribution $P_Z$, and any constants $C(z)$,
\begin{align}
\frac{1}{\alpha-1} \log \sum_z P_Z(z) C(z)^{\alpha-1}
\end{align}
is non-decreasing in $\alpha$ for $\alpha>1$.
\begin{figure*}[b]
\normalsize
\setcounter{MYtempeqncnt}{\value{equation}}
\hrulefill
\setcounter{equation}{226}
\begin{equation}\label{eq:cont-ex-Q}
   Q_{x',P_{\tilde{X}}}(\alpha,\tau)=
    \begin{cases}
       \displaystyle\frac{1}{\tau}I(\tilde{X};Y)+\left(1-\frac{1}{\tau}\right)D_{KL}\left(P_{Y|X}(y|x)\|P_{Y|X}(y|x')|P_{\tilde{X}}(x)\right) & \text{if } \alpha=1,\tau\in[1,\infty]\\
       \displaystyle \frac{1}{\tau}  \log
\sum_y P_{Y|X}(y|x')^{1-\tau} \max_x P_{Y|X}(y|x)^\tau & \text{if } \alpha=\infty, \tau\in[1,\infty)\\
\displaystyle\frac{1}{\alpha-1} \log \sum_{x,y} P_{\tilde{X}}(x) P_{Y|X}(y|x')^{1-\alpha} P_{Y|X}(y|x)^\alpha  & \text{if}\; \alpha\in(1,\infty),\tau=\infty\\
\displaystyle\max_{x,y}\; P_{Y|X}(y|x) P_{Y|X}(y|x')^{-1}&\text{if} \; \alpha=\infty,\tau=\infty.
    \end{cases}
\end{equation} 
\setcounter{equation}{\value{MYtempeqncnt}}
\vspace*{4pt}
\end{figure*}

\section{Proof of Proposition~\ref{prop4}}\label{proof:alpha-equals-one-tau-leakage}
For a fixed $\tau$, we have
\begin{align}
&\nonumber\lim_{\alpha \to 1}\mathcal{L}_{\alpha,\tau}(X\to Y)\\\nonumber
&= \lim_{\alpha \to 1}\,\max_{x'}\,\sup_{P_{\Tilde{X}}}\,\inf_{Q_Y}\ \frac{1}{\alpha-1} \log \sum_{x,y} P_{\Tilde{X}}(x) P_{Y|X}(y|x)^\alpha \\& \ \times \left(Q_Y(y)^{\frac{1}{\tau}} P_{Y|X}(y|x')^{1-\frac{1}{\tau}}\right)^{1-\alpha}\\\nonumber
&=\max_{x'}\, \lim_{\alpha \to 1}\,\sup_{P_{\Tilde{X}}}\,\inf_{Q_Y}\ \frac{1}{\alpha-1} \log \sum_{x,y} P_{\Tilde{X}}(x) P_{Y|X}(y|x)^\alpha \\& \ \times \left(Q_Y(y)^{\frac{1}{\tau}} P_{Y|X}(y|x')^{1-\frac{1}{\tau}}\right)^{1-\alpha}\label{eq:lim-max}\\\nonumber
&=\max_{x'}\, \inf_{\alpha}\,\sup_{P_{\Tilde{X}}}\,\inf_{Q_Y}\ \frac{1}{\alpha-1} \log \sum_{x,y} P_{\Tilde{X}}(x) P_{Y|X}(y|x)^\alpha \\& \ \times \left(Q_Y(y)^{\frac{1}{\tau}} P_{Y|X}(y|x')^{1-\frac{1}{\tau}}\right)^{1-\alpha}\label{eq:increase1}\\\nonumber
&=\max_{x'}\, \inf_{\alpha}\,\inf_{Q_Y}\,\sup_{P_{\Tilde{X}}}\ \frac{1}{\alpha-1} \log\, \sum_{x,y} P_{\Tilde{X}}(x) P_{Y|X}(y|x)^\alpha \\& \ \times \left(Q_Y(y)^{\frac{1}{\tau}} P_{Y|X}(y|x')^{1-\frac{1}{\tau}}\right)^{1-\alpha}\label{eq:concave-convex}\\\nonumber 
&=\max_{x'}\,\inf_{Q_Y}\,\inf_{\alpha}\,\sup_{P_{\Tilde{X}}}\ \frac{1}{\alpha-1} \log\, \sum_{x,y} P_{\Tilde{X}}(x) P_{Y|X}(y|x)^\alpha \\& \ \times\left(Q_Y(y)^{\frac{1}{\tau}} P_{Y|X}(y|x')^{1-\frac{1}{\tau}}\right)^{1-\alpha}\\\nonumber 
&=\max_{x'}\,\inf_{Q_Y}\,\lim_{\alpha\to 1}\,\sup_{P_{\Tilde{X}}}\ \frac{1}{\alpha-1} \log\, \sum_{x,y} P_{\Tilde{X}}(x) P_{Y|X}(y|x)^\alpha \\& \ \times \left(Q_Y(y)^{\frac{1}{\tau}} P_{Y|X}(y|x')^{1-\frac{1}{\tau}}\right)^{1-\alpha}\label{increase2}\\\nonumber
&=\max_{x'}\,\inf_{Q_Y}\,\lim_{\alpha\to 1}\,\max_{x}\ \frac{1}{\alpha-1} \log\, \sum_{y} P_{Y|X}(y|x)^\alpha \\& \ \times \left(Q_Y(y)^{\frac{1}{\tau}} P_{Y|X}(y|x')^{1-\frac{1}{\tau}}\right)^{1-\alpha}\label{eq:concave-P_tilde_X}\\\nonumber
&=\max_{x'}\,\inf_{Q_Y}\,\max_{x}\,\lim_{\alpha \to 1}\, \frac{1}{\alpha-1} \log \sum_{y} P_{Y|X}(y|x)^\alpha \\& \ \times  \left(Q_Y(y)^{\frac{1}{\tau}} P_{Y|X}(y|x')^{1-\frac{1}{\tau}}\right)^{1-\alpha}\label{eq:lim-max1}\\ 
&=\max_{x'}\,\inf_{Q_Y}\,\max_{x}\, \sum_{y} P_{Y|X}(y|x)\log\frac{P_{Y|X}(y|x)}{Q_Y(y)^{\frac{1}{\tau}} P_{Y|X}(y|x')^{1-\frac{1}{\tau}}}\label{eqn:L'Hopital's rule}\\\nonumber
&=\max_{x'}\,\inf_{Q_Y}\,\sup_{P_{\tilde{X}}}\, \sum_{x,y} P_{\tilde{X}}(x)P_{Y|X}(y|x)\\ & \ \times \log\frac{P_{Y|X}(y|x)}{Q_Y(y)^{\frac{1}{\tau}} P_{Y|X}(y|x')^{1-\frac{1}{\tau}}}\label{eq:lin-in-P}\\\nonumber
&=\max_{x'}\,\sup_{P_{\tilde{X}}}\,\inf_{Q_Y}\, \sum_{x,y} P_{\tilde{X}}(x)P_{Y|X}(y|x)\\ & \ \times\log\frac{P_{Y|X}(y|x)}{Q_Y(y)^{\frac{1}{\tau}} P_{Y|X}(y|x')^{1-\frac{1}{\tau}}}\label{eq:convex-concave}\\\nonumber
&=\max_{x'}\,\sup_{P_{\tilde{X}}}\,\inf_{Q_Y}\, \sum_{x,y} P_{\tilde{X}}(x)P_{Y|X}(y|x)\\ & \ \times \log\left(\frac{P_{Y|X}(y|x) }{P_{Y|X}(y|x')}\right)^{1-\frac{1}{\tau}}\left(\frac{P_{Y|X}(y|x)}{Q_Y(y)}\right)^\frac{1}{\tau}\\\nonumber
&=\max_{x'}\,\sup_{P_{\tilde{X}}}\,\inf_{Q_Y}\, ({1-\frac{1}{\tau}})\sum_{x,y} P_{\tilde{X}}(x)P_{Y|X}(y|x)\\\nonumber & \ \times \log\left(\frac{P_{Y|X}(y|x) }{P_{Y|X}(y|x')}\right)\\ & \ +\frac{1}{\tau}\sum_{x,y} P_{\tilde{X}}(x)P_{Y|X}(y|x)\log\left(\frac{P_{Y|X}(y|x)}{Q_Y(y)}\right)\\\nonumber
&=\max_{x'}\,\sup_{P_{\tilde{X}}}\,\Bigg[ (1-\frac{1}{\tau})D_{KL}\left(P_{Y|X}(y|x)\|P_{Y|X}(y|x')|P_{\tilde{X}}(x)\right)\\& \ +\frac{1}{\tau}\,\inf_{Q_Y}\, D_{KL}(P_{\tilde{X}}(x)P_{Y|X}(y|x)\|P_{\tilde{X}}(x) Q_Y(y))\Bigg]\\\nonumber
&=\max_{x'}\,\sup_{P_{\tilde{X}}}\, (1-\frac{1}{\tau})D_{KL}\left(P_{Y|X}(y|x)\|P_{Y|X}(y|x')|P_{\tilde{X}}(x)\right)\\& \ +\frac{1}{\tau}\,I(\tilde{X};Y)\label{last-shannon-leakage},
\end{align}
where
\begin{itemize}
    \item\eqref{eq:increase1} and \eqref{increase2} follow because the objective function is non-decreasing in $\alpha$,
    \item \eqref{eq:concave-convex} follows because the quantity inside the log is linear (and thus concave) in $P_{\tilde{X}}$ and convex in $Q_Y$,
    \item \eqref{eq:concave-P_tilde_X} follows because the quantity inside the log is linear in $P_{\tilde{X}}$ and so supremum is attained at a corner point,
    \item \eqref{eqn:L'Hopital's rule} follows from L'Hopital's rule,
    \item \eqref{eq:lin-in-P} follows because the objective function in \eqref{eq:lin-in-P} is linear in $P_{\tilde{X}}$ and so the sup is attained at a corner point,
    \item \eqref{eq:convex-concave} follows because the objective is linear (thus concave) in $P_{\tilde{X}}$ and convex in $Q_Y$.
\end{itemize}
\section{Proof of Theorem~\ref{thm:continuity}}\label{proof:continuity}
We first prove the continuity of maximal $(\alpha,\tau)$-leakage, and following that, we demonstrate the continuity of M$\alpha$beL. In our analysis, we employ the extended real number line. 
\subsection{Continuity of maximal $(\alpha,\tau)$-leakage}
Let
\begin{equation}
\tau^-_\epsilon(\tau_0)=
    \begin{cases}
        \tau_0-\epsilon & \text{if } \tau_0 \in (1,\infty)\\
        \displaystyle\frac{1}{\epsilon} & \text{if } \tau_0=\infty\\
        1 & \text{if } \tau_0=1,
    \end{cases} 
    \end{equation}
    \vspace{3mm}
    \begin{equation} 
    \tau^+_\epsilon(\tau_0)=
        \begin{cases}
        \tau_0+\epsilon & \text{if } \tau_0 \in [1,\infty)\\
        \infty & \text{if } \tau_0=\infty,
    \end{cases}
\end{equation}
\vspace{3mm}
\begin{equation}
\alpha^-_\epsilon(\alpha_0)=
    \begin{cases}
        \alpha_0-\epsilon & \text{if } \alpha_0 \in (1,\infty)\\
        \displaystyle\frac{1}{\epsilon} & \text{if } \alpha_0=\infty\\
        1 & \text{if } \alpha_0=1,
    \end{cases} 
    \end{equation}
    and
    \begin{equation}
    \alpha^+_\epsilon(\alpha_0)=
        \begin{cases}
        \alpha_0+\epsilon & \text{if } \alpha_0 \in [1,\infty)\\
        \infty & \text{if } \alpha_0=\infty.
    \end{cases}
\end{equation}
To prove the continuity of maximal $(\alpha,\tau)$-leakage at $(\alpha_0,\tau_0)\in [1,\infty]\times [1,\infty]$, we define a rectangular region characterized by its corners at $\left(\alpha^-_\epsilon(\alpha_0), \tau^-_\epsilon(\tau_0)\right)$, $\left(\alpha^-_\epsilon(\alpha_0), \tau^+_\epsilon(\tau_0)\right)$, $\left(\alpha^+_\epsilon(\alpha_0), \tau^-_\epsilon(\tau_0)\right)$, and $\left(\alpha^+_\epsilon(\alpha_0), \tau^+_\epsilon(\tau_0)\right)$, with the point $(\alpha_0, \tau_0)$ lying inside or on the borders of this region, and we show that
\begin{align}
   \nonumber&\displaystyle \lim_{\epsilon\to 0} \displaystyle\inf_{\substack{\alpha\in[\alpha^-_\epsilon(\alpha_0),\alpha^+_\epsilon(\alpha_0)],\\\tau\in[\tau^-_\epsilon(\tau_0),\tau^+_\epsilon(\tau_0)]}} \mathcal{L}_{\alpha,\tau}(X \to Y)\\& =
   \displaystyle \lim_{\epsilon\to 0} \displaystyle\sup_{\substack{\alpha\in[\alpha^-_\epsilon(\alpha_0),\alpha^+_\epsilon(\alpha_0)],\\\tau\in[\tau^-_\epsilon(\tau_0),\tau^+_\epsilon(\tau_0)]}} \mathcal{L}_{\alpha,\tau}(X \to Y)\\&=\mathcal{L}_{\alpha_0,\tau_0}(X \to Y).
\end{align}
Note that the rectangular region converges to the point $(\alpha_0,\tau_0)$ as $\epsilon\to 0$. We recall that
\begin{equation}\label{eq:cont-ex-matle}
    \begin{cases}
    \mathcal{L}_{\alpha=1,\tau}(X\to Y))=\text{$\tau$-Shannon leakage}\\
       \mathcal{L}_{\alpha=\infty,\tau}(X\to Y)=\text{Maximal R\'enyi leakage of order $\tau$}\\
       \mathcal{L}_{\alpha,\tau=\infty}(X\to Y)=\mathcal{L}^{\text{LRDP}}_{\alpha}(X\to Y)\\
\mathcal{L}_{\alpha=\infty,\tau=\infty}(X\to Y)=\mathcal{L}^{\text{LDP}}(X\to Y).
    \end{cases}
\end{equation}
 \begin{figure*}[b]
\normalsize
\setcounter{MYtempeqncnt}{\value{equation}}
\hrulefill
\setcounter{equation}{246}
\begin{equation}\label{eq:cont-ex-f}
   f_{x',Q_Y,x}(\alpha,\tau)=
    \begin{cases}
     \displaystyle \sum_y P_{Y|X}(y|x)\;\log \displaystyle\frac{P_{Y|X}(y|x)}{Q_Y(y)^{\frac{1}{\tau}} P_{Y|X}(y|x')^{1-\frac{1}{\tau}}} & \text{if } \alpha=1,\tau\in[1,\infty]\vspace{3mm}\\\vspace{2mm}
      \displaystyle \max_y\; \log P_{Y|X}(y|x)\;Q_Y(y)^{-\frac{1}{\tau}}P_{Y|X}(y|x')^{-1+\frac{1}{\tau}} & \text{if } \alpha=\infty, \tau\in[1,\infty)\\
      D_{\alpha}\left(P_{Y|X}(y|x)\|P_{Y|X}(y|x')\right)   & \text{if } \alpha\in(1,\infty), \tau=\infty\vspace{3mm}\\ 
      \displaystyle\max_y P_{Y|X}(y|x) P_{Y|X}(y|x')^{-1} & \text{if}\;\alpha=\infty, \tau=\infty.
    \end{cases}
\end{equation} 
\setcounter{equation}{\value{MYtempeqncnt}}
\vspace*{4pt}
\end{figure*}
\textbf{Lower bound:}
For $(\alpha,\tau)\in(1,\infty)\times [1,\infty)$, we have 
\begin{align}\label{LintermsofQ}
    \mathcal{L}_{\alpha,\tau}(X\to Y)
=\max_{x'}\ \sup_{P_{\Tilde{X}}} \; Q_{x',P_{\tilde{X}}}(\alpha,\tau),
\end{align}
where
\begin{align}\label{def:Q}
 \nonumber&Q_{x',P_{\tilde{X}}}(\alpha,\tau)=  \left(\frac{1}{\alpha-1}+\frac{1}{\tau}\right)  \log \sum_y  P_{Y|X}(y|x')^{\frac{(\tau-1)(1-\alpha)}{\tau+\alpha-1}}
\\& \ \times \left(\sum_x P_{\Tilde{X}}(x)P_{Y|X}(y|x)^\alpha\right)^{\frac{\tau}{\tau+\alpha-1}}.   
\end{align}
$Q_{x',P_{\tilde{X}}}(\alpha,\tau)$ is continuous in $(\alpha,\tau)$, for all $(\alpha,\tau)\in (1,\infty)\times [1,\infty)$, and it may be defined by its continuous extension at $\alpha=\infty$, $\tau=\infty$, or $\alpha=1$, shown in \eqref{eq:cont-ex-Q}. \addtocounter{equation}{1}
We know that \eqref{LintermsofQ} holds for $(\alpha,\tau)\in(1,\infty)\times[1,\infty)$. Now we show that it also holds for $(\alpha,\tau)\in [1,\infty]^2\setminus (1,\infty)\times [1,\infty)$. Employing \eqref{eq:cont-ex-matle} and \eqref{eq:cont-ex-Q}, for $(\alpha,\tau)\in\{1,\infty\}\times[1,\infty]$, we can see that
\begin{align}
    \displaystyle\max_{x'}\  \sup_{P_{\Tilde{X}}}\ Q_{x',P_{\tilde{X}}}(\alpha,\tau)=\mathcal{L}_{\alpha,\tau}(X\to Y).
\end{align}
For $\alpha\in(1,\infty)$ and $\tau=\infty$, we also have
\begin{align}
    &\nonumber\displaystyle\max_{x'}\  \sup_{P_{\Tilde{X}}}\ Q_{x',P_{\tilde{X}}}(\alpha,\tau=\infty)\\&=\displaystyle\max_{x'}  \sup_{P_{\Tilde{X}}} \frac{1}{\alpha-1} \log \sum_{x,y} P_{\tilde{X}}(x) P_{Y|X}(y|x')^{1-\alpha} P_{Y|X}(y|x)^\alpha\\
    &=\displaystyle\max_{x'}\  \max_x \frac{1}{\alpha-1} \log \sum_{y}  P_{Y|X}(y|x')^{1-\alpha} P_{Y|X}(y|x)^\alpha\label{eq:lininpx}\\
    &=\mathcal{L}_{\alpha,\tau=\infty}(X\to Y),
\end{align}
where \eqref{eq:lininpx} follows because the quantity inside the logarithm is linear in $P_{\tilde{X}}$ and so the supremum is achieved at an endpoint.

 For $(\alpha_0,\tau_0)\in [1,\infty]\times [1,\infty]$, we have
\begin{align}
&\displaystyle \lim_{\epsilon\to 0} \displaystyle\inf_{\substack{\alpha\in[\alpha^-_\epsilon(\alpha_0),\alpha^+_\epsilon(\alpha_0)],\\\tau\in[\tau^-_\epsilon(\tau_0),\tau^+_\epsilon(\tau_0)]}} \mathcal{L}_{\alpha,\tau}(X \to Y)\\
 &=
\displaystyle\lim_{\epsilon \to 0}\ \mathcal{L}_{\alpha^-_\epsilon(\alpha_0),\tau^-_\epsilon(\tau_0)}(X \to Y)\label{eq:mon}\\
&=\displaystyle\lim_{\epsilon \to 0}\ \max_{x'}\ \sup_{P_{\Tilde{X}}} \; Q_{x',P_{\tilde{X}}}(\alpha^-_{\epsilon}(\alpha_0),\tau^-_{\epsilon}(\tau_0))\\
&= \max_{x'}\ \displaystyle\lim_{\epsilon \to 0}\ \sup_{P_{\Tilde{X}}} \; Q_{x',P_{\tilde{X}}}(\alpha^-_{\epsilon}(\alpha_0),\tau^-_{\epsilon}(\tau_0))\\
&= \max_{x'}\ \displaystyle\sup_{\epsilon > 0}\ \sup_{P_{\Tilde{X}}} \; Q_{x',P_{\tilde{X}}}(\alpha^-_{\epsilon}(\alpha_0),\tau^-_{\epsilon}(\tau_0))\label{eq:limtosup}\\
&= \max_{x'}\  \sup_{P_{\Tilde{X}}} \ \sup_{\epsilon > 0} \; Q_{x',P_{\tilde{X}}}(\alpha^-_{\epsilon}(\alpha_0),\tau^-_{\epsilon}(\tau_0))\\
&= \max_{x'}\  \sup_{P_{\Tilde{X}}} \ \lim_{\epsilon \to 0} \; Q_{x',P_{\tilde{X}}}(\alpha^-_{\epsilon}(\alpha_0),\tau^-_{\epsilon}(\tau_0))\label{eq:limtosup1}\\
&= \max_{x'}\  \sup_{P_{\Tilde{X}}} \  Q_{x',P_{\tilde{X}}}(\alpha^-_{0}(\alpha_0),\tau^-_{0}(\tau_0))\label{eq:continuity}\\
&= \max_{x'}\  \sup_{P_{\Tilde{X}}} \  Q_{x',P_{\tilde{X}}}(\alpha_0,\tau_0)\\
&=\mathcal{L}_{\alpha_0,\tau_0}(X \to Y)\label{result-lower}
\end{align}
where \eqref{eq:mon} follows because $\mathcal{L}_{\alpha,\tau}(X\to Y)$ is non-decreasing in $\alpha$ and $\tau$. \eqref{eq:limtosup} and \eqref{eq:limtosup1} follow because $Q_{x',P_{\tilde{X}}}(\alpha,\tau)$ is non-decreasing in $\alpha$ and $\tau$ for a fixed $\tau$ and $\alpha$, respectively. Moreover, $\alpha^-_{\epsilon}(\alpha_0)$ and $\tau^-_{\epsilon}(\tau_0)$ are non-increasing in $\epsilon$ for a fixed $\alpha_0$ and $\tau_0$, respectively. The equality \eqref{eq:continuity} follows because $Q_{x',P_{\tilde{X}}}(\alpha^-_{\epsilon}(\alpha_0),\tau^-_{\epsilon}(\tau_0))$ is continuous in $\epsilon$.

\textbf{Upper bound:} For $(\alpha,\tau)\in(1,\infty)\times [1,\infty)$, consider the expression of $\mathcal{L}_{\alpha,\tau}(X \to Y)$ in \eqref{alpha_tau_variational}:
\begin{align}
&\nonumber\mathcal{L}_{\alpha,\tau}(X\to Y)
\\\nonumber&=\max_{x'}\,\sup_{P_{\Tilde{X}}}\,\inf_{Q_Y}\ \frac{1}{\alpha-1} \log \sum_{x,y} P_{\Tilde{X}}(x) P_{Y|X}(y|x)^\alpha\\&\ \times \left(Q_Y(y)^{\frac{1}{\tau}} P_{Y|X}(y|x')^{1-\frac{1}{\tau}}\right)^{1-\alpha}\label{eq:object1}\\\nonumber
&=\max_{x'}\,\inf_{Q_Y}\,\sup_{P_{\Tilde{X}}}\ \frac{1}{\alpha-1} \log \sum_{x,y} P_{\Tilde{X}}(x) P_{Y|X}(y|x)^\alpha \\&\ \times \left(Q_Y(y)^{\frac{1}{\tau}} P_{Y|X}(y|x')^{1-\frac{1}{\tau}}\right)^{1-\alpha}\label{eq:convex-concave1}\\\nonumber
&=\max_{x'}\,\inf_{Q_Y}\,\max_x\ \frac{1}{\alpha-1} \log \sum_{y}  P_{Y|X}(y|x)^\alpha \\&\ \times\left(Q_Y(y)^{\frac{1}{\tau}} P_{Y|X}(y|x')^{1-\frac{1}{\tau}}\right)^{1-\alpha}\label{eq:convex-sup}
\end{align}
where \eqref{eq:convex-concave1} follows because the quantity inside the logarithm is convex in $Q_Y$ and linear (so concave) in $P_{\tilde{X}}$, and \eqref{eq:convex-sup} follows because the quantity inside the logarithm is linear in $P_{\tilde{X}}$ and so the supremum is achieved at an endpoint. Consider the objective in  \eqref{eq:convex-sup}:
\begin{align}
 &\nonumber f_{x',Q_Y,x}(\alpha,\tau)\\&= \frac{1}{\alpha-1} \log \sum_{y}  P_{Y|X}(y|x)^\alpha \left(Q_Y(y)^{\frac{1}{\tau}} P_{Y|X}(y|x')^{1-\frac{1}{\tau}}\right)^{1-\alpha}\\\nonumber
 &=\log \Bigg[\sum_{y} P_{Y|X}(y|x) \Bigg(P_{Y|X}(y|x) \;Q_Y(y)^{-\frac{1}{\tau}}\\&\ \times  P_{Y|X}(y|x')^{-1+\frac{1}{\tau}}\Bigg)^{\alpha-1}\Bigg]^{1/{\alpha-1}}\label{eq:norm-mon}. 
\end{align}
 $f_{x',Q_Y,x}(\alpha,\tau)$ may be defined by its continuous extension at $\alpha=\infty$, $\tau=\infty$, or $\alpha=1$, shown in \eqref{eq:cont-ex-f}. \addtocounter{equation}{1}

Looking at \eqref{eq:convex-sup}, for $(\alpha,\tau)\in(1,\infty)\times [1,\infty)$, it is clear that 
\begin{align}
    \mathcal{L}_{\alpha,\tau}(X\to Y)=\max_{x'} \;\inf_{Q_Y} \;\max_{x}\; f_{x',Q_Y,x}(\alpha,\tau).\label{Lequalsf}
\end{align}
We now show that the above equality also holds for $(\alpha,\tau)\in[1,\infty]^2\setminus(1,\infty)\times[1,\infty)$. Employing \eqref{eq:cont-ex-f}, for $\alpha\in(1,\infty)$ and $\tau=\infty$, we have 
   \begin{align}\label{eq:cont-ex1}
      \nonumber&\max_{x'} \ \inf_{Q_Y} \ \max_x  \ f_{x',Q_Y,x}(\alpha,\tau=\infty)\\&=\mathcal{L}_{\alpha}^{\text{LRDP}}(X \to Y)\\&=\mathcal{L}_{\alpha,\tau=\infty}(X\to Y).
   \end{align}
Similarly, for $\alpha=\tau=\infty$, we have 
      \begin{align}
      \nonumber&\max_{x'} \ \inf_{Q_Y} \ \max_x  \ f_{x',Q_Y,x}(\alpha=\infty,\tau=\infty)\\&=\mathcal{L}^{\text{LDP}}(X \to Y)\\&=\mathcal{L}_{\alpha=\infty,\tau=\infty}(X\to Y).
   \end{align}
    Applying similar steps to equations \eqref{eqn:L'Hopital's rule}-\eqref{last-shannon-leakage}, for $\alpha=1$ and $\tau\in[1,\infty]$, we have 
    \begin{align}\label{eq:cont-ex2}
       \nonumber&\max_{x'} \ \inf_{Q_Y} \ \max_x  \ f_{x',Q_Y,x}(\alpha=1,\tau)\\&=\mathcal{L}_{\alpha=1,\tau}(X\to Y)\\&=\text{$\tau$-Shannon leakage}. 
    \end{align}
    For $\alpha=\infty$ and $\tau\in[1,\infty)$, we have
  \begin{align}
&\nonumber\max_{x'}\,\inf_{Q_Y}\,\max_x\ f_{x',Q_Y,x}(\alpha=\infty,\tau)\\
&=\max_{x'}\,\inf_{Q_Y}\,\max_{x,y}\  \log  P_{Y|X}(y|x)\;Q_Y(y)^{-\frac{1}{\tau}}P_{Y|X}(y|x')^{-1+\frac{1}{\tau}} \\\nonumber
&=\frac{1}{\tau}\max_{x'}\ \inf_{Q_Y}\ \max_y \ \log  Q_Y(y)^{-1} \ P_{Y|X}(y|x')^{1-\tau} \\& \ \times \max_x P_{Y|X}(y|x)^{\tau}\\\nonumber
&=\frac{1}{\tau}\max_{x'}\,\inf_{Q_Y}\ \max_y \ \log   Q_Y(y)^{-1} \\\nonumber& \ \times \frac{ P_{Y|X}(y|x')^{1-\tau} \ \displaystyle\max_x P_{Y|X}(y|x)^\tau}{\displaystyle\sum_{y'} P_{Y|X}(y'|x')^{1-\tau} \ \displaystyle\max_x P_{Y|X}(y'|x)^{\tau}}\\& \ \times \displaystyle\sum_{y'} P_{Y|X}(y'|x')^{1-\tau} \ \displaystyle\max_x P_{Y|X}(y'|x)^{\tau}\\\nonumber
&=\frac{1}{\tau}\max_{x'}\Bigg[\log \displaystyle\sum_{y'} P_{Y|X}(y'|x')^{1-\tau} \ \displaystyle\max_x P_{Y|X}(y'|x)^{\tau} \\\nonumber&\ + \inf_{Q_Y}\; \log \;\max_y \; Q_Y(y)^{-1} \\& \ \times \frac{ P_{Y|X}(y|x')^{1-\tau} \ \displaystyle\max_x P_{Y|X}(y|x)^{\tau}}{\displaystyle\sum_{y'} P_{Y|X}(y'|x')^{1-\tau} \ \max_x P_{Y|X}(y'|x)^{\tau}}\Bigg].\label{collapse-mrl}
\end{align}
Now we show that
\begin{align}
   \nonumber& \inf_{Q_Y}\; \log \;\max_y \; Q_Y(y)^{-1} \\& \ \times \frac{ P_{Y|X}(y|x')^{1-\tau} \ \displaystyle\max_x P_{Y|X}(y|x)^{\tau}}{\displaystyle\sum_{y'} P_{Y|X}(y'|x')^{1-\tau} \ \max_x P_{Y|X}(y'|x)^{\tau}}=0.
\end{align}
Let $g_Y(y)=\displaystyle\frac{ P_{Y|X}(y|x')^{1-\tau} \ \displaystyle\max_x P_{Y|X}(y|x)^{\tau}}{\displaystyle\sum_{y'} P_{Y|X}(y'|x')^{1-\tau} \ \max_x P_{Y|X}(y'|x)^{\tau}}$ be a distribution on $Y$. If $Q_Y(y)=g_Y(y)$ for all $y\in\mathcal{Y}$, we have
$\log  \displaystyle\max_y Q_Y(y)^{-1} g_Y(y)=0$. So, 
\begin{align}\label{upper-div}
  \displaystyle\inf_{Q_Y}\;\log \; \max_y\;  Q_Y(y)^{-1} g_Y(y)\leq 0. 
\end{align}
Moreover, we have 
\begin{align}
  \nonumber&\inf_{Q_Y}\;  \log\;\max_y \;\frac{g_Y(y)}{Q_Y(y)}\\&=\inf_{Q_Y}\;\max_y\;\log\frac{g_Y(y)}{Q_Y(y)}\\
    &=\inf_{Q_Y}\;\sum_{y'} g_Y(y')\left(\max_y\;\log\frac{g_Y(y)}{Q_Y(y)}\right) \\
    &\geq\inf_{Q_Y}\;\sum_{y'} g_Y(y')\;\log\frac{g_Y(y')}{Q_Y(y')}\\
    &=\inf_{Q_Y} D_{KL}\left(g_Y\|Q_Y\right)=0.\label{lower-div}
\end{align}
Combining \eqref{upper-div} and \eqref{lower-div}, we get $\displaystyle \inf_{Q_Y}\log\;\max_y \;\frac{g_Y(y)}{Q_Y(y)}=0$. Therefore, \eqref{collapse-mrl} collapses to
\begin{align}
&\nonumber\frac{1}{\tau}\max_{x'} \ \log \displaystyle\sum_{y'} P_{Y|X}(y'|x')^{1-\tau} \ \displaystyle\max_x P_{Y|X}(y'|x)^{\tau}\label{new-divergence}\\&=\mathcal{L}_{\alpha=\infty,\tau}(X\to Y).
\end{align}
This completes the proof of the validity of \eqref{Lequalsf}  for all $(\alpha,\tau)\in[1,\infty]^2$.

We now investigate the monotonicity of $f_{x',Q_Y,x}(\alpha,\tau)$ in $\alpha$ and $\tau$. $f_{x',Q_Y,x}(\alpha,\tau)$ is non-decreasing in $\alpha$ for a fixed $\tau$ because the quantity inside the logarithm in \eqref{eq:norm-mon} represents $(\alpha-1)$-norm of a random variable, and $(\alpha-1)$-norm of a random variable is non-decreasing in $\alpha$. In the proof of Lemma~\ref{lemma:alpha_tau_variational}, we show that $Q_{x',P_{\tilde{X}}}(\alpha,\tau)$ defined in \eqref{def:Q} is non-decreasing in $\tau$ and so $\displaystyle\sup_{P_{\tilde{X}}}Q_{x',P_{\tilde{X}}}(\alpha,\tau)$ is non-decreasing in $\tau$.  Moreover, equations \eqref{eq:object1}-\eqref{eq:convex-sup} show that $\displaystyle\inf_{Q_Y} \max_x f_{x',Q_Y,x}(\alpha,\tau)=\displaystyle\sup_{P_{\tilde{X}}}Q_{x',P_{\tilde{X}}}(\alpha,\tau)$ (note that the infimum of the objective in \eqref{eq:object1} over $Q_Y$ is equal to $Q_{x',P_{\tilde{X}}}(\alpha,\tau)$). So, $\displaystyle\inf_{Q_Y} \max_x f_{x',Q_Y,x}(\alpha,\tau)$ is non-decreasing in $\tau$.

For $(\alpha_0,\tau_0)\in [1,\infty]^2$, we have
\begin{align}
   & \displaystyle \lim_{\epsilon\to 0} \displaystyle\sup_{\substack{\alpha\in[\alpha^-_\epsilon(\alpha_0),\alpha^+_\epsilon(\alpha_0)],\\\tau\in[\tau^-_\epsilon(\tau_0),\tau^+_\epsilon(\tau_0)]}} \mathcal{L}_{\alpha,\tau}(X \to Y)\\
   &=\displaystyle \lim_{\epsilon\to 0} \ \mathcal{L}_{\alpha^+_\epsilon(\alpha_0),\tau^+_\epsilon(\tau_0)}(X \to Y)\label{eq:mon3}\\
  &=\displaystyle \lim_{\epsilon\to 0} \ \max_{x'}\,\inf_{Q_Y}\,\max_x\ f_{x',Q_Y,x}(\alpha^+_\epsilon(\alpha_0),\tau^+_\epsilon(\tau_0))\\
  &= \max_{x'}\ \displaystyle \lim_{\epsilon\to 0} \ \inf_{Q_Y}\,\max_x\ f_{x',Q_Y,x}(\alpha^+_\epsilon(\alpha_0),\tau^+_\epsilon(\tau_0))\\
  &= \max_{x'}\ \displaystyle \inf_{\epsilon>0} \ \inf_{Q_Y}\,\max_x\ f_{x',Q_Y,x}(\alpha^+_\epsilon(\alpha_0),\tau^+_\epsilon(\tau_0))\label{eq:increas-inf}\\
  &= \max_{x'} \ \inf_{Q_Y} \ \displaystyle \inf_{\epsilon>0} \,\max_x\ f_{x',Q_Y,x}(\alpha^+_\epsilon(\alpha_0),\tau^+_\epsilon(\tau_0))\\
  &\leq \max_{x'} \ \inf_{Q_Y} \ \displaystyle \lim_{\epsilon\to 0} \,\max_x\ f_{x',Q_Y,x}(\alpha^+_\epsilon(\alpha_0),\tau^+_\epsilon(\tau_0))\label{upperbound-not=}\\
  &= \max_{x'} \ \inf_{Q_Y} \ \max_x \ \displaystyle \lim_{\epsilon\to 0} \ f_{x',Q_Y,x}(\alpha^+_\epsilon(\alpha_0),\tau^+_\epsilon(\tau_0))\\
  &= \max_{x'} \ \inf_{Q_Y} \ \max_x  \ f_{x',Q_Y,x}(\alpha^+_0(\alpha_0),\tau^+_0(\tau_0))\label{eq:continuous-in-epsilon}\\
    &= \max_{x'} \ \inf_{Q_Y} \ \max_x  \ f_{x',Q_Y,x}(\alpha_0,\tau_0)\label{ineq:lim-sup}\\
 & =\mathcal{L}_{\alpha_0,\tau_0}(X \to Y),\label{result-upper}
\end{align}
where 
\begin{itemize}
\item \eqref{eq:mon3} follows because $\mathcal{L}_{\alpha,\tau}(X \to Y)$ is non-decreasing in $\alpha$ and $\tau$.
    \item \eqref{eq:increas-inf} follows because $f_{x',Q_Y,x}(\alpha,\tau)$ is non-decreasing in $\alpha$ for a fixed $\tau$ and 
    $\displaystyle\inf_{Q_Y} \max_x f_{x',Q_Y,x}(\alpha,\tau)$ is non-decreasing in $\tau$ for a fixed $\alpha$. Moreover,  $\alpha^+_\epsilon(\alpha_0)$ and $\tau^+_\epsilon(\tau_0)$ are non-decreasing in $\epsilon$ for a fixed $\alpha_0$ and $\tau_0$, respectively. 
    \item \eqref{eq:continuous-in-epsilon} follows because $f_{x',Q_Y,x}(\alpha^+_\epsilon(\alpha_0),\tau^+_\epsilon(\tau_0))$ is continuous in $\epsilon$.
    \end{itemize}    
For $(\alpha_0,\tau_0)\in[1,\infty]^2$, combining \eqref{result-lower} and \eqref{result-upper}, we have 
\begin{align}
\nonumber\mathcal{L}_{\alpha_0,\tau_0}(X \to Y)&= \displaystyle \lim_{\epsilon\to 0} \displaystyle\inf_{\substack{\alpha\in[\alpha^-_\epsilon(\alpha_0),\alpha^+_\epsilon(\alpha_0)],\\\tau\in[\tau^-_\epsilon(\tau_0),\tau^+_\epsilon(\tau_0)]}} \mathcal{L}_{\alpha,\tau}(X \to Y)\\\nonumber&\leq
   \displaystyle \lim_{\epsilon\to 0} \displaystyle\sup_{\substack{\alpha\in[\alpha^-_\epsilon(\alpha_0),\alpha^+_\epsilon(\alpha_0)],\\\tau\in[\tau^-_\epsilon(\tau_0),\tau^+_\epsilon(\tau_0)]}} \mathcal{L}_{\alpha,\tau}(X \to Y)\\&\leq\mathcal{L}_{\alpha_0,\tau_0}(X \to Y).
\end{align}
which implies
\begin{align}
   &\nonumber \displaystyle \lim_{\epsilon\to 0} \displaystyle\inf_{\substack{\alpha\in[\alpha^-_\epsilon(\alpha_0),\alpha^+_\epsilon(\alpha_0)],\\\tau\in[\tau^-_\epsilon(\tau_0),\tau^+_\epsilon(\tau_0)]}} \mathcal{L}_{\alpha,\tau}(X \to Y)\\&=
   \displaystyle \lim_{\epsilon\to 0} \displaystyle\sup_{\substack{\alpha\in[\alpha^-_\epsilon(\alpha_0),\alpha^+_\epsilon(\alpha_0)],\\\tau\in[\tau^-_\epsilon(\tau_0),\tau^+_\epsilon(\tau_0)]}} \mathcal{L}_{\alpha,\tau}(X \to Y)\\&=\mathcal{L}_{\alpha_0,\tau_0}(X \to Y).
\end{align}
This completes the proof of the continuity of maximal $(\alpha,\tau)$-leakage.
\subsection{Continuity of $\mathcal{L}_{\alpha,\beta}(X\to Y)$ at $(\alpha,\beta)\in[1,\infty]^2\setminus\{(1,1)\}$}
Let $$B_{XY}(\alpha,\beta)=\mathcal{L}_{\alpha,\beta}(X \to Y)$$ and $$T_{XY}(\alpha,\tau)=\mathcal{L}_{\alpha,\tau}(X \to Y).$$
For $\alpha>\beta$, we have
\begin{align}
    T_{XY}(\alpha,\tau)=B_{XY}\left(\alpha,\displaystyle\frac{\alpha \tau}{\tau+\alpha-1}\right)  
\end{align}
and
\begin{align}
  B_{XY}(\alpha,\beta)=T_{XY}\left(\alpha,\displaystyle\frac{\beta (\alpha-1)}{\alpha-\beta}\right).  
\end{align}
To demonstrate the continuity of M$\alpha$beL, we employ the sequential continuity theorem, that is, a function $f$ is continuous at $a$ if and only if $f(x_n)\to f(a)$ for all sequences $x_n \to a$.

\textbf{Continuity at $(\alpha_0,\beta_0)$ with $\alpha_0>\beta_0$:} If $\displaystyle\lim_{n \to \infty} (\alpha_n,\beta_n)=(\alpha_0,\beta_0)$, then for sufficiently large $n$, we have $\alpha_n>\beta_n$ and so
\begin{align}
 &\nonumber \lim_{n\to \infty}B_{XY}(\alpha_n,\beta_n)\\&= \lim_{n\to \infty} T_{XY}\left(\alpha_n,\displaystyle\frac{\beta_n (\alpha_n-1)}{\alpha_n-\beta_n}\right)\\
&=T_{XY}\left(\lim_{n\to \infty}\alpha_n,\lim_{n\to \infty}\displaystyle\frac{\beta_n (\alpha_n-1)}{\alpha_n-\beta_n}\right)\label{eq:cont}\\
&=T_{XY}\left(\alpha_0,\displaystyle\frac{\beta_0(\alpha_0-1)}{\alpha_0-\beta_0}\right)\label{eq:cont-mapping}\\
&=B_{XY}(\alpha_0,\beta_0),
\end{align}
where \eqref{eq:cont} and \eqref{eq:cont-mapping} follow from the continuity of maximal $(\alpha,\tau)$-leakage and the continuity of $\displaystyle\frac{\beta (\alpha-1)}{\alpha-\beta}$ for $\alpha>\beta$, respectively. 

\textbf{Continuity at $(\alpha_0,\beta_0)$ with $1\neq\alpha_0<\beta_0$:} We first prove the continuity of $\mathcal{L}_{\beta}^{\text{LRDP}}(X \to Y)$ at $\beta> 1$ as follows. If $\displaystyle\lim_{n\to \infty}\beta_n=\beta\neq 1$, then we have \begin{align}
    &\nonumber\lim_{n\to \infty}\mathcal{L}^{\text{LRDP}}_{\beta_n}(X\rightarrow Y)
\\\nonumber&=\lim_{n\to \infty}\max_{x,x^\prime\in\mathcal{X}}\frac{1}{\beta_n-1}
 \log \sum_y P_{Y|X}(y|x')^{1-\beta_n} \\& \ \times P_{Y|X}(y|x)^{\beta_n}\\\nonumber
 &=\max_{x,x^\prime\in\mathcal{X}} \lim_{n\to \infty} \frac{1}{\beta_n-1}
 \log \sum_y P_{Y|X}(y|x')^{1-\beta_n} \\ & \ \times P_{Y|X}(y|x)^{\beta_n}\\
 &=\max_{x,x^\prime\in\mathcal{X}}  \frac{1}{\beta-1}
 \log \sum_y P_{Y|X}(y|x')^{1-\beta} P_{Y|X}(y|x)^{\beta}\label{eq:con-obj}\\
 &=\mathcal{L}^{\text{LRDP}}_{\beta}(X \to Y),
\end{align}
where \eqref{eq:con-obj} follows because the objective is continuous at $\beta>1$. Note that the objective is defined by its continuous extension, i.e., KL divergence, at $\beta=1$.
If $\displaystyle\lim_{n \to \infty} (\alpha_n,\beta_n)=(\alpha_0,\beta_0)$ with $1\neq \alpha_0<\beta_0$, then for sufficiently large $n$, we have $1\neq\alpha_n<\beta_n$. So,
\begin{align}
 &\nonumber\lim_{n\to \infty}\mathcal{L}_{\alpha_n,\beta_n}(X\to Y)\\&=\lim_{n\to \infty} \displaystyle\frac{\alpha_n (\beta_n-1)}{(\alpha_n-1)\beta_n}\     \mathcal{L}^{\text{LRDP}}_{\beta_n}(X\rightarrow Y)\\
&=\displaystyle\frac{\alpha_0 (\beta_0-1)}{(\alpha_0-1)\beta_0} \   \mathcal{L}^{\text{LRDP}}_{\beta_0}(X\rightarrow Y)\label{eq:cont-reyni}\\
&=\mathcal{L}_{\alpha_0,\beta_0}(X \to Y), 
\end{align}
where \eqref{eq:cont-reyni} follows from the continuity of $\mathcal{L}_{\beta}^{\text{LRDP}}(X\to Y)$ and $\displaystyle\frac{\alpha (\beta-1)}{(\alpha-1)\beta}$ for $\alpha\neq 1$.

\textbf{Continuity at $(\alpha_0,\beta_0)$ with $1=\alpha_0<\beta_0$:}

Let $\displaystyle\lim_{n\to \infty} (\alpha_n,\beta_n)=(1,\beta_0\neq1)$. If $\mathcal{L}_{\beta_0}^{\text{LRDP}}(X\to Y)\neq 0$, for sufficiently large $n$, we have 
\begin{align}
    \nonumber & \lim_{n\to \infty}\mathcal{L}_{\alpha_n,\beta_n}(X\to Y)\\&=\lim_{n\to \infty} \displaystyle\frac{\alpha_n (\beta_n-1)}{(\alpha_n-1)\beta_n}\     \mathcal{L}^{\text{LRDP}}_{\beta_n}(X\rightarrow Y)\\&=\infty.
\end{align}
Here, maximal $(\alpha,\beta)$-leakage is continuous in the sense that the limit points of all the sequences $\{\mathcal{L}_{\alpha_n,\beta_n}(X\to Y)\}_{n\in\mathbb{N}}$ are  equal to $\infty$. Moreover, 
if $\mathcal{L}_{\beta_0}^{\text{LRDP}}(X\to Y)= 0$, then $X$ and $Y$ are independent. So,  $\mathcal{L}_{\alpha,\beta}(X \to Y)=0$ everywhere and is continuous. 

\textbf{Continuity at $(\alpha_0,\beta_0)$ with $\alpha_0=\beta_0\neq 1$:}

Let $$\displaystyle\lim_{n\to \infty}(\alpha_n,\beta_n)=(\beta_0,\beta_0)\neq (1,1),$$
and let  partition $\{(\alpha_n,\beta_n)\}_{n\in\mathbb{N}}$ into two distinct subsequences
$\{(\alpha_{n_k},\beta_{n_k})\}_{k\in S_1}$ with $\alpha_{n_k}\leq \beta_{n_k}$ and $\{(\alpha_{\bar{n}_k},\beta_{\bar{n}_k})\}_{k\in S_2}$ with $\alpha_{\bar{n}_k}> \beta_{\bar{n}_k},$
where $\{n_k\}_{k\in S_1}\cup \{\bar{n}_k\}_{k\in S_2}=\mathbb{N}$ and $S_1,S_2\subseteq \mathbb{N}$. If either of the subsequences is finite, then the sequence $ \{(\alpha_n,\beta_n)\}_{n>N}$
with $N>\min\{\displaystyle\max_{k\in S_1} n_k,\max_{k\in S_2}\bar{n}_k\}$ consists entirely of elements from the other subsequence. As a result, this finite subsequence does not impact the convergence of the original sequence $\{(\alpha_n,\beta_n)\}_{n\in\mathbb{N}}$. Here, we consider the scenario where $S_1=S_2=\mathbb{N}$. Since every subsequence of a convergent sequence converges to the same limit as the original sequence, we have
$\displaystyle\lim_{k\to \infty} (\alpha_{n_k},\beta_{n_k})=(\beta_0,\beta_0)$ and  $\displaystyle\lim_{k\to \infty} (\alpha_{\bar{n_k}},\beta_{\bar{n_k}})=(\beta_0,\beta_0).$ For $\{(\alpha_{n_k},\beta_{n_k})\}_{k\in \mathbb{N}}$ with $\alpha_{n_k}\leq \beta_{n_k}$, we have
\begin{align}
   \nonumber& \displaystyle\lim_{k \to \infty} \mathcal{L}_{\alpha_{n_k},\beta_{n_k}}(X\to Y)\\&=\displaystyle\lim_{k\to \infty} \frac{\alpha_{n_k}(\beta_{n_k}-1)}{(\alpha_{n_k}-1)\beta_{n_k}} \ \mathcal{L}^{\text{LRDP}}_{\beta_{n_k}}(X \to Y)\\
&=\mathcal{L}^{\text{LRDP}}_{\beta_0}(X \to Y)\label{cont-pr}.
\end{align}
The last equality follows from the continuity of $\mathcal{L}_{\beta}^{\text{LRDP}}(X\to Y)$ and $\displaystyle\frac{\alpha (\beta-1)}{(\alpha-1)\beta}$ for $\alpha\neq 1$. Furthermore, recalling the functions $B_{XY}(\alpha,\beta)$ and $T_{XY}(\alpha,\tau)$, for $\{(\alpha_{\bar{n}_k},\beta_{\bar{n}_k})\}_{k\in \mathbb{N}}$ with $\alpha_{\bar{n}_k}>\beta_{\bar{n}_k}$, we have
\begin{align}
    &\nonumber\displaystyle\lim_{k \to \infty} \mathcal{L}_{\alpha_{\bar{n}_k},\beta_{\bar{n}_k}}(X\to Y)\\&=\displaystyle\lim_{k \to \infty} B_{XY}(\alpha_{\bar{n}_k},\beta_{\bar{n}_k})\\
    &=\displaystyle\lim_{k \to \infty} T_{XY}\left(\alpha_{\bar{n}_k},\frac{\beta_{\bar{n}_k}(\alpha_{\bar{n}_k}-1)}{\alpha_{\bar{n}_k}-\beta_{\bar{n}_k}}\right)\\
    &= T_{XY}\left(\displaystyle\lim_{k \to \infty}\alpha_{\bar{n}_k},\displaystyle\lim_{k \to \infty}\frac{\beta_{\bar{n_k}}(\alpha_{\bar{n}_k}-1)}{\alpha_{\bar{n}_k}-\beta_{\bar{n}_k}}\right)\label{eq:cont2}\\
    &=T_{XY}(\beta_0,\infty)\\
    &=\mathcal{L}^{\text{LRDP}}_{\beta_0}(X\to Y)\label{eq:lim-lrdp}
\end{align}
where \eqref{eq:cont2} follows from the continuity of $\mathcal{L}_{\alpha,\tau}(X \to Y)$ and \eqref{eq:lim-lrdp} follows because $$\displaystyle\lim_{\tau\to \infty}\mathcal{L}_{\alpha,\tau}(X \to Y)=\mathcal{L}^{\text{LRDP}}_{\alpha}(X\to Y).$$ Combining \eqref{cont-pr} and \eqref{eq:lim-lrdp}, we get 
$$\displaystyle\lim_{n\to \infty} \mathcal{L}_{\alpha_n,\beta_n}(X \to Y)=\mathcal{L}^{\text{LRDP}}_{\beta_0}(X\to Y).$$
This completes the proof of continuity.
\section{Proof of Theorem~\ref{thm:continuous-alpha-beta-leakage}}\label{proof:thm-continuous-alpha-beta-leakage} 
Here, we extend our results to continuous real random variables through the Riemann integral. The results can also be extended to higher dimensions through Lebesgue integral. We first consider a case in which $X$ still has a finite alphabet $\mathcal{X}$ but $Y$ takes value from a continuous alphabet $\mathcal{Y}\subseteq \mathbb{R}$.
\begin{lemma}\label{lemma:discrete-x-continuous-Y}
    When $X$ has a finite alphabet and $Y$ is continuous, maximal $(\alpha,\beta)$-leakage defined in \eqref{eqn:alpha,beta-leakage-original-def} simplifies to 
    \begin{align}
        \nonumber \mathcal{L}_{\alpha,\beta}(X\to Y)&= \sup_{P_{X}}\ \sup_{P_{\Tilde{X}}} \ \frac{\alpha}{(\alpha-1)\beta}
\log 
\int_{\mathcal{Y}} f_{Y}(y)^{1-\beta} \\& \ \left( \sum_x P_{\Tilde{X}}(x) f_{Y|X}(y|x)^\alpha \right)^{\frac{\beta}{\alpha}}dy
    \end{align}
  where $P_{\Tilde{X}}$ is a distribution on the support of $P_X$.  
\end{lemma}
The proof of lemma~\ref{lemma:discrete-x-continuous-Y} follows similar steps to the proof of Theorem~\ref{theorem:alpha-beta-leakage} along the lines of \cite[Proof of Theorem~7]{IssaWK2020}. Using lemma~\ref{lemma:discrete-x-continuous-Y}, we prove Theorem~\ref{thm:continuous-alpha-beta-leakage} as follows. 
\\
\textbf{Upper bound:} Applying similar steps to the proof of upper bound for Theorem~\ref{theorem:alpha-beta-leakage}, we may get 
\begin{align}
\nonumber\mathcal{L}_{\alpha,\beta}(X\to Y)&\leq\sup_{f_{X}}\ \sup_{f_{\Tilde{X}}} \ \frac{\alpha}{(\alpha-1)\beta}
\log 
\int_{\mathcal{Y}} f_{Y}(y)^{1-\beta} \\&\left( \int_{\mathcal{X}} f_{\Tilde{X}}(x) f_{Y|X}(y|x)^\alpha dx\right)^{\frac{\beta}{\alpha}}dy.\label{upper-continuous}
\end{align}
\textbf{Lower bound:} 
Fix $n_1,n_2\in \mathbb{N}$ and $a,b\in \mathbb{R}$ such that $a,b>0$. We partition the intervals $[-a,a]$ and $[-b,b]$ into subintervals with equal lengths $\Delta_1=\displaystyle\frac{2a}{n_1}$ and $\Delta_2=\displaystyle\frac{2b}{n_2}$, respectively. Let 
$$\bar{X}=\displaystyle\sum_{i=1}^{n_1} x^*_i \ \mathbbm{1}\{X\in [-a+(i-1)\Delta_1,-a+i\Delta_1] \}$$
for $x \in [-a,a]$, and $\bar{X}=x_0^*$,  otherwise. Moreover, let 
$$\bar{Y}=\displaystyle\sum_{j=1}^{n_2} y^*_j \ \mathbbm{1}\{Y\in [-b+(j-1)\Delta_2,-b+j\Delta_2] \}$$
for $y \in [-b,b]$, and $\bar{Y}=y_0^*$, otherwise. Here, $\mathbbm{1}\{\cdot\} $ is the indicator function. Furthermore, $x^*_i\in [-a+(i-1)\Delta_1,-a+i\Delta_1]$  for $i\in \{1,\cdots,n_1\}$ and $y^*_j\in [-b+(j-1)\Delta_2,-b+j\Delta_2]$  for $j\in \{1,\cdots,n_2\}$ are fixed points, and $x_0^* $ and $y_0^*$ are fixed symbols. For continuous random variables $X$ and $Y$ and their quantized versions $\bar{X}$ and $\bar{Y}$, we now prove that since the Markov chain $\bar{X}-X-Y-\bar{Y}$ holds, we have $\mathcal{L}_{\alpha,\beta}(X \to Y) \geq \mathcal{L}_{\alpha,\beta} (\bar{X}\to \bar{Y})$. To do so, we first prove $\mathcal{L}_{\alpha,\beta}(X \to Y) \geq \mathcal{L}_{\alpha,\beta}(\bar{X} \to Y) $ and then, we show $\mathcal{L}_{\alpha,\beta}(\bar{X} \to Y) \geq \mathcal{L}_{\alpha,\beta}(\bar{X} \to \bar{Y})$. At the end, we bound $\mathcal{L}_{\alpha,\beta}(\bar{X}\to \bar{Y})$ from below. 
Applying similar steps to the proof of the linkage inequality for random variables with finite alphabets, i.e., \eqref{data-processing1}, we may prove the following lemma. 
\begin{lemma}
    Let $X$ have a finite alphabet and $Y$ and $Z$ be continuous random variables. If the Markov chain $X-Y-Z$ holds then $\mathcal{L}_{\alpha,\beta}(Y\to Z) \geq \mathcal{L}_{\alpha,\beta}(X\to Z)$, \text{which means}
\begin{align}\label{ineq:linkage-continuous}
  \mathcal{L}_{\alpha,\beta}(X \to Y) \geq \mathcal{L}_{\alpha,\beta}(\bar{X} \to Y).  
\end{align}
\end{lemma} 
Using lemma~\ref{lemma:discrete-x-continuous-Y}, we prove the following post-processing inequality.
\begin{lemma}\label{ineq:post-process}
   Let $X$ and $Z$ have finite alphabets and $Y$ be a continuous random variable. If the Markov chain $X-Y-Z$ holds then $\mathcal{L}_{\alpha,\beta}(X\to Y) \geq \mathcal{L}_{\alpha,\beta}(X\to Z)$. 
\end{lemma}
\begin{proof}
 For any $y \in \mathcal{Y}$, let
\begin{align}\label{def:g-continuous}
    g(y)=\left( \sum_{x} P_{\Tilde{X}}(x) f_{Y|X}(y|x)^{\alpha}\right)^{\frac{1}{\alpha}}
\end{align}
and
\begin{align}\label{def:c-continuous}
    c_z(y)=\displaystyle\frac{f_{Y}(y)\; P_{Z|Y}(z|y)}{P_{Z}(z)}
\end{align}
such that $\int_{\mathcal{Y}} c_z(y)\ dy=1$. We have
\begin{align}
    &\int_{\mathcal{Y}} f_Y(y)^{1-\beta} \left( \sum_{x} P_{\Tilde{X}}(x) f_{Y|X}(y|x)^{\alpha}\right)^{\frac{\beta}{\alpha}}dy
    \\&=\int_{\mathcal{Y}} f_{Y}(y)^{1-\beta}g(y)^\beta dy
    \\&=\int_{\mathcal{Y}}\sum_{z} f_Y(y) P_{Z|Y}(z|y)\left(\frac{g(y)}{f_Y(y)}\right)^\beta dy
    \\&=\sum_{z} P_Z(z) \int_{\mathcal{Y}} c_z(y)\left(\frac{g(y)}{f_Y(y)}\right)^\beta dy
    \\&\ge \sum_z P_Z(z) \left(\int_{\mathcal{Y}} c_z(y)\frac{g(y)}{f_Y(y)} \ dy\right)^\beta \label{jensens_application}
    \\&=\sum_z P_{Z}(z)^{1-\beta} \left(\int_{\mathcal{Y}} P_{Z|Y}(z|y)g(y) \ dy\right)^\beta\label{g_form1}
\end{align}
where \eqref{jensens_application} follows from applying Jensen's inequality to the convex function $f:\;x\to x^{p}\;(x \geq 0,\; p\geq 1)$. Recalling the definition of $g(y)$ from \eqref{def:g-continuous}, we have
\begin{align}
    &\int_{\mathcal{Y}} P_{Z|Y}(z|y) g(y)\ dy
    \\&=\int_{\mathcal{Y}} P_{Z|Y}(z|y) \bigg( \sum_{x} P_{\Tilde{X}}(x)  f_{Y|X}(y|x)^{\alpha}\bigg)^{\frac{1}{\alpha}}dy
    \\&=\int_{\mathcal{Y}} \bigg(\sum_{x}  \Big( P_{\Tilde{X}}(x)^{\frac{1}{\alpha}}P_{Z|Y}(z|y)  f_{Y|X}(y|x)\Big)^{\alpha}\bigg)^{\frac{1}{\alpha}}dy
    \\&\ge \bigg(\sum_{x}  \Big( \int_{\mathcal{Y}}  P_{\Tilde{X}}(x)^{\frac{1}{\alpha}}P_{Z|Y}(z|y)  f_{Y|X}(y|x) \ dy\Big)^{\alpha}\bigg)^{\frac{1}{\alpha}}\label{ineq:Minkowski}
    \\&=\left(\sum_{x} P_{\Tilde{X}}(x)  P_{Z|X}(z|x)^{\alpha}\right)^{\frac{1}{\alpha}}\label{applying Markov chain1}
\end{align}
where 
\begin{itemize}
    \item  \eqref{ineq:Minkowski} follows from Minkowski's integral inequality,
    \item \eqref{applying Markov chain1} follows because the Markov chain $X - Y - Z$ holds.
\end{itemize} 
Applying \eqref{applying Markov chain1} to \eqref{g_form1}, and using the fact that for $\alpha \in (1,\infty)$ and $\beta \in [1,\infty)$, the function $f: t \to \frac{\alpha}{(\alpha-1)\beta} \log t$ is increasing in $t > 0$, give
\begin{align}
   & \frac{\alpha}{(\alpha-1)\beta} \log \int_{\mathcal{Y}} f_{Y}(y)^{1-\beta}\left( \sum_{x} P_{\Tilde{X}}(x) f_{Y|X}(y|x)^{\alpha}\right)^{\frac{\beta}{\alpha}}dy\nonumber\\ &\geq  \frac{\alpha}{(\alpha-1)\beta} \log  \sum_z P_{Z}(z)^{1-\beta} \left(\sum_{x} P_{\Tilde{X}}(x)  P_{Z|X}(z|x)^{\alpha}\right)^{\frac{\beta}{\alpha}}.
\end{align}
Taking suprema over $P_X$ and $P_{\Tilde{X}}$ completes the proof.
\end{proof}
From lemma~\ref{ineq:post-process}, we get $\mathcal{L}_{\alpha,\beta}(\bar{X}\to Y)\geq \mathcal{L}_{\alpha,\beta}(\bar{X}\to \bar{Y})$. Combining this result with \eqref{ineq:linkage-continuous}, 
we have $\mathcal{L}_{\alpha,\beta}(X\to Y)\geq \mathcal{L}_{\alpha,\beta}(\bar{X}\to \bar{Y})$. 
Since $\mathcal{L}_{\alpha,\beta}(X\to Y)\geq \mathcal{L}_{\alpha,\beta}(\bar{X}\to \bar{Y})$ holds for all $a,b>0$ and  $n_1, n_2\in\mathbb{N}$, we get $\mathcal{L}_{\alpha,\beta}(X\to Y)\geq \displaystyle\sup_{a,b,n_1,n_2}\mathcal{L}_{\alpha,\beta}(\bar{X}\to \bar{Y})$. Moreover, since $\bar{X}$ and $\bar{Y}$ have finite alphabets, we may use the result of lower bound for Theorem~\ref{theorem:alpha-beta-leakage}, equation~\eqref{eq:result-lower-finite}, and so we have 
\begin{align}
\nonumber&\mathcal{L}_{\alpha,\beta}(X\to Y)\\\nonumber&\geq \sup_{a,b,n_1,n_2} \sup_{P_{\bar{X}},P_{\Tilde{X}}} \ \frac{\alpha}{(\alpha-1)\beta} \
\log
\sum_{j=0}^{n_2} P_{\bar{Y}}(y^*_j)^{1-\beta}\\& \  \left(\sum_{i=0}^{n_1} P_{\Tilde{X}}(x^*_i) P_{\bar{Y}|\bar{X}}(y^*_j|x^*_i)^\alpha \right)^{\frac{\beta}{\alpha}}\\\nonumber
&\geq \sup_{P_{\bar{X}},P_{\Tilde{X}}}  \ \frac{\alpha}{(\alpha-1)\beta}\
\log \sup_{a,b,n_1,n_2}\
\sum_{j=1}^{n_2} P_{\bar{Y}}(y^*_j)^{1-\beta} \\& \ \left(\sum_{i=1}^{n_1} P_{\Tilde{X}}(x^*_i) P_{\bar{Y}|\bar{X}}(y^*_j|x^*_i)^\alpha \right)^{\frac{\beta}{\alpha}}\\\nonumber
&\geq \sup_{P_{\bar{X}},P_{\Tilde{X}}}  \ \frac{\alpha}{(\alpha-1)\beta}
\log \ \sup_{b,n_2}\
\sum_{j=1}^{n_2} P_{\bar{Y}}(y^*_j)^{1-\beta}\\& \  \left(\lim_{a \to \infty} \ \lim_{n_1 \to \infty} \sum_{i=1}^{n_1} P_{\Tilde{X}}(x^*_i) P_{\bar{Y}|\bar{X}}(y^*_j|x^*_i)^\alpha \right)^{\frac{\beta}{\alpha}}\\\nonumber
&= \sup_{f_{\bar{X}},f_{\Tilde{X}}}  \ \frac{\alpha}{(\alpha-1)\beta}
\log \ \sup_{b,n_2}\
\sum_{j=1}^{n_2} P_{\bar{Y}}(y^*_j)^{1-\beta} \\& \ \left(\int_{\mathcal{X}} f_{\Tilde{X}}(x) P_{\bar{Y}|X}(y^*_j|x)^\alpha dx\right)^{\frac{\beta}{\alpha}}\label{eq:Riemann-x}\\\nonumber
&\geq \sup_{f_{\bar{X}},f_{\Tilde{X}}}  \ \frac{\alpha}{(\alpha-1)\beta}
\log \lim_{b\to \infty} \lim_{n_2 \to \infty}
\sum_{\substack{j=1:\\P_{\bar{Y}}(y^*_j)>0}}^{n_2} P_{\bar{Y}}(y^*_j) \\& \ \left(\int_{\mathcal{X}} f_{\Tilde{X}}(x) \left(\frac{P_{\bar{Y}|X}(y^*_j|x)}{P_{\bar{Y}}(y^*_j)}\right)^\alpha dx\right)^{\frac{\beta}{\alpha}}\\\nonumber
&= \sup_{f_{\bar{X}},f_{\Tilde{X}}}  \ \frac{\alpha}{(\alpha-1)\beta}
\log \int_{\mathcal{Y}} f_{Y}(y)\\& \ \left(\int_{\mathcal{X}} f_{\Tilde{X}}(x) \left(\frac{f_{Y|X}(y|x)}{f_{Y}(y)}\right)^\alpha dx\right)^{\frac{\beta}{\alpha}}dy\label{eq:Riemann-y}\\\nonumber
&= \sup_{f_{\bar{X}},f_{\Tilde{X}}}  \ \frac{\alpha}{(\alpha-1)\beta}
\log \int_{\mathcal{Y}} f_{Y}(y)^{1-\beta} \\& \ \left(\int_{\mathcal{X}} f_{\Tilde{X}}(x) f_{Y|X}(y|x)^\alpha dx\right)^{\frac{\beta}{\alpha}}dy\label{lower-continuous}
\end{align}
where $P_{\tilde{X}}$ is a distribution on the support of $\bar{X}$ , $f_{X}$ is the marginal pdf of $X$ and $f_{\tilde{X}}$ is a pdf on the support of $X$. Furthermore, \eqref{eq:Riemann-x} and \eqref{eq:Riemann-y} follow from the definition of the Riemann integral. Combining \eqref{upper-continuous} and \eqref{lower-continuous} gives
\begin{align}
  \nonumber& \mathcal{L}_{\alpha,\beta}(X\to Y)\\\nonumber&= \sup_{f_{X}}\ \sup_{f_{\Tilde{X}}} \ \frac{\alpha}{(\alpha-1)\beta}
\log 
\int_{\mathcal{Y}} f_{Y}(y)^{1-\beta} \\& \ \left( \int_{\mathcal{X}} f_{\Tilde{X}}(x) f_{Y|X}(y|x)^\alpha dx\right)^{\frac{\beta}{\alpha}}dy\\\nonumber
&=\max_{x': f_X(x')>0} \  \sup_{f_{\Tilde{X}}}\ \frac{\alpha}{(\alpha-1)\beta}  \log 
\int_{\mathcal{Y}} f_{Y|X}(y|x')^{1-\beta} \\& \ \left(\int_{\mathcal{X}} f_{\Tilde{X}}(x) f_{Y|X}(y|x)^\alpha dx\right)^{\beta/\alpha}dy
\end{align}
where the latter equality follows because the quantity inside the log is convex in $f_X$. Similarly, we may prove the following lemma.
\begin{lemma}
  When X is continuous and Y has a finite alphabet, maximal $(\alpha,\beta)$-leakage defined in \eqref{eqn:alpha,beta-leakage-original-def} simplifies to 
  \begin{align}
  \nonumber& \mathcal{L}_{\alpha,\beta}(X\to Y)\\\nonumber&=\max_{x': f_X(x')>0} \  \sup_{f_{\Tilde{X}}}\ \frac{\alpha}{(\alpha-1)\beta}  \log 
\sum_y P_{Y|X}(y|x')^{1-\beta} \\& \ \left(\int_{\mathcal{X}} f_{\Tilde{X}}(x) P_{Y|X}(y|x)^\alpha dx\right)^{\beta/\alpha}
\end{align}
where $f_{\tilde{X}}$ is a pdf on $\mathcal{X}$.
\end{lemma}
\section{Results for Known Mechanisms~\ref{sec:Illustration of Results}}\label{proof:lemma-Laplacian-mechanism}
Here, we obtain an upper bound on vector maximal Renyi leakage under a Laplacian mechanism. The proof for a Gaussian mechanism follows similarly.\\
Let $c_{i,x_{-i}}=\displaystyle\max_{x_i} \ h(x_{-i},x_i)$, $a_{i,x_{-i}}=\displaystyle\min_{x_i} \ h(x_{-i},x_i)$, and $Y=\mathcal{M}(X^n)$. Since  $|h(x_{-i},x_i)-h(x_{-i},\tilde{x}_i)|\leq \delta$, we have $c_{i,x_{-i}}-a_{i,x_{-i}}\leq \delta$. Moreover, for continuous alphabets we have
\begin{align}
   \nonumber&\max_i \ \mathcal{L}_{\infty,\beta}(X_i \to Y|X_{-i})\\\nonumber&=\max_i \ \max_{\tilde{x_i},x_{-i}} \ \frac{1}{\beta} \log
\int_{-\infty}^{\infty} f_{Y|X_i,X_{-i}}(y|\tilde{x_i},x_{-i})^{1-\beta} \ \\& \ \times \max_{x_i} \ f_{Y|X_i,X_{-i}}(y|x_i,x_{-i})^\beta \ dy,
\end{align}
where $$ f_{Y|X_i,X_{-i}}(y|x_i,x_{-i})=f_N\left(y-h(x_{-i},x_i)\right)$$ and $$ f_{Y|X_i,X_{-i}}(y|\tilde{x_i},x_{-i})=f_N\left(y-h(x_{-i},\tilde{x_i})\right).$$For fixed $i, \tilde{x}_i$ and $x_{-i}$, consider the quantity inside the logarithm:
\begin{align}
 &\nonumber\int_{-\infty}^{\infty} f_{Y|X_i,X_{-i}}(y|\tilde{x_i},x_{-i})^{1-\beta} \ \\& \ \times \max_{x_i} f_{Y|X_i,X_{-i}}(y|x_i,x_{-i})^\beta \ dy\\\nonumber
 &=(\frac{1}{2b})^\beta \int_{-\infty}^{\infty} f_{Y|X_i,X_{-i}}(y|\tilde{x_i},x_{-i})^{1-\beta} \ \\& \ \times \max_{x_i} \  \exp{-\frac{\beta|y-h(x_{-i},x_i)|}{b}} \ dy\\\nonumber
&=(\frac{1}{2b})^\beta \bigg[\int_{-\infty}^{a_{i,x_{-i}}} f_{Y|X_i,X_{-i}}(y|\tilde{x_i},x_{-i})^{1-\beta} \\\nonumber& \ \times  \max_{x_i} \  \exp{-\frac{\beta|y-h(x_{-i},x_i)|}{b}} \ dy\\\nonumber& \ + \int_{a_{i,x_{-i}}}^{c_{i,x_{-i}}} f_{Y|X_i,X_{-i}}(y|\tilde{x_i},x_{-i})^{1-\beta} \\\nonumber& \ \times \max_{x_i} \  \exp{-\frac{\beta|y-h(x_{-i},x_i)|}{b}} \ dy \\\nonumber& \ + \int_{c_{i,x_{-i}}}^\infty f_{Y|X_i,X_{-i}}(y|\tilde{x_i},x_{-i})^{1-\beta} \\& \ \times \max_{x_i} \  \exp{-\frac{\beta|y-h(x_{-i},x_i)|}{b}} \ dy\bigg]\\\nonumber
&\leq (\frac{1}{2b})^\beta \bigg[\int_{-\infty}^{a_{i,x_{-i}}} f_{Y|X_i,X_{-i}}(y|\tilde{x_i},x_{-i})^{1-\beta} \\\nonumber& \  \times  \exp{-\frac{\beta|y-a_{i,x_{-i}}|}{b}} \ dy\\\nonumber & \ + \int_{a_{i,x_{-i}}}^{c_{i,x_{-i}}} f_{Y|X_i,X_{-i}}(y|\tilde{x_i},x_{-i})^{1-\beta}\ dy \\&+ \int_{c_{i,x_{-i}}}^\infty f_{Y|X_i,X_{-i}}(y|\tilde{x_i},x_{-i})^{1-\beta} \  \exp{-\frac{\beta|y-c_{i,x_{-i}}|}{b}} \ dy\bigg]\\\nonumber
&=\frac{1}{2b} \bigg[\int_{-\infty}^{a_{i,x_{-i}}} \exp{\frac{(\beta-1)|y-h(x_{-i},\tilde{x_i})|}{b}} \\\nonumber & \ \times \exp{-\frac{\beta|y-a_{i,x_{-i}}|}{b}} \ dy\\\nonumber &\ + \int_{a_{i,x_{-i}}}^{c_{i,x_{-i}}} \exp{\frac{(\beta-1)|y-h(x_{-i},\tilde{x_i})|}{b}}\ dy \\\nonumber&+ \int_{c_{i,x_{-i}}}^\infty \exp{\frac{(\beta-1)|y-h(x_{-i},\tilde{x_i})|}{b}} \\& \ \times  \exp{-\frac{\beta|y-c_{i,x_{-i}}|}{b}} \ dy\bigg]\\\nonumber
&=\frac{1}{2b} \bigg[\int_{-\infty}^{a_{i,x_{-i}}} \exp{\frac{(1-\beta)(y-h(x_{-i},\tilde{x_i}))}{b}} \\\nonumber &\ \times \exp{\frac{\beta(y-a_{i,x_{-i}})}{b}} \ dy\\\nonumber & \ + \int_{a_{i,x_{-i}}}^{h(x_{-i},\tilde{x_i})} \exp{\frac{(1-\beta)(y-h(x_{-i},\tilde{x_i}))}{b}}\ dy \\\nonumber& \ + \int_{h(x_{-i},\tilde{x_i})}^{c_{i,x_{-i}}} \exp{\frac{(\beta-1)(y-h(x_{-i},\tilde{x_i}))}{b}}\ dy\\\nonumber& \ + \int_{c_{i,x_{-i}}}^\infty \exp{\frac{(\beta-1)(y-h(x_{-i},\tilde{x_i}))}{b}} \\& \ \times  \exp{-\frac{\beta(y-c_{i,x_{-i}})}{b}} \ dy\bigg]\\\nonumber
&=\frac{1}{2} \exp\left( \frac{(\beta-1)(h(x_{-i},\tilde{x_i})-a_{i,x_{-i}})}{b}\right)\\\nonumber & \ +\frac{1}{2}\exp\left( \frac{(\beta-1)(c_{i,x_{-i}}-h(x_{-i},\tilde{x_i}))}{b}\right)\\\nonumber & \ +\frac{1}{2(\beta-1)}\Bigg[\exp\left( \frac{(\beta-1)(h(x_{-i},\tilde{x_i})-a_{i,x_{-i}})}{b}\right)\\ & \ +\exp\left( \frac{(\beta-1)(c_{i,x_{-i}}-h(x_{-i},\tilde{x_i}))}{b}\right)-2\Bigg]
\end{align}
Thus, we have
\begin{align}
 \nonumber  &\max_i \ \mathcal{L}_{\infty,\beta}(X_i \to Y|X_{-i})\\\nonumber
 &\leq\max_{i,\tilde{x}_i,x_{-i}} \ \frac{1}{\beta} \log
\Bigg[\frac{1}{2} \exp\left( \frac{(\beta-1)(h(x_{-i},\tilde{x_i})-a_{i,x_{-i}})}{b}\right)\\\nonumber& \ +\frac{1}{2}\exp\left( \frac{(\beta-1)(c_{i,x_{-i}}-h(x_{-i},\tilde{x_i}))}{b}\right)\\\nonumber& \ +\frac{1}{2(\beta-1)}\Bigg(\exp\left( \frac{(\beta-1)(h(x_{-i},\tilde{x_i})-a_{i,x_{-i}})}{b}\right)\\& \ +\exp\left( \frac{(\beta-1)(c_{i,x_{-i}}-h(x_{-i},\tilde{x_i}))}{b}\right)-2\Bigg)\Bigg]\\\nonumber
&=\max_{i,x_{-i}} \ \max_{h(x_{-i},\tilde{x}_i)\in[a_{i,x_{-i}},c_{i,x_{-i}}]}\ \frac{1}{\beta} \\\nonumber& \ \log
\Bigg[\frac{1}{2} \exp\left( \frac{(\beta-1)(h(x_{-i},\tilde{x_i})-a_{i,x_{-i}})}{b}\right)\\\nonumber& \ +\frac{1}{2}\exp\left( \frac{(\beta-1)(c_{i,x_{-i}}-h(x_{-i},\tilde{x_i}))}{b}\right)\\\nonumber& \ +\frac{1}{2(\beta-1)}\Bigg(\exp\left( \frac{(\beta-1)(h(x_{-i},\tilde{x_i})-a_{i,x_{-i}})}{b}\right)\\& \ +\exp\left( \frac{(\beta-1)(c_{i,x_{-i}}-h(x_{-i},\tilde{x_i}))}{b}\right)-2\Bigg)\Bigg]\\\nonumber
&=\max_{i,x_{-i}}\  \frac{1}{\beta} \ \log \Bigg[  \frac{1}{2}-\frac{1}{2(\beta-1)}+\left(\frac{1}{2}+\frac{1}{2(\beta-1)}\right) \\& \  \exp \left(\frac{(\beta-1)(c_{i,x_{-i}}-a_{i,x_{-i}})}{b}\right)\Bigg]\label{eq:convex in h}\\\nonumber
 &\leq \frac{1}{\beta} \ \log \Bigg[  \frac{1}{2}-\frac{1}{2(\beta-1)}+\left(\frac{1}{2}+\frac{1}{2(\beta-1)}\right) \\& \ \exp \left(\frac{(\beta-1)\delta}{b}\right)\Bigg],\label{ineq:c-a<delta}
\end{align}
 where \eqref{eq:convex in h} follows because the quantity inside the logarithm is convex in $h(x_{-i},\tilde{x_i})$ and so we have $h(x_{-i},\tilde{x_i})=a$ or $h(x_{-i},\tilde{x_i})=c$. It is easy to show that both values of $h(x_{-i},\tilde{x_i})$ give the expression \eqref{eq:convex in h}. Moreover, \eqref{ineq:c-a<delta} follows because $c_{i,x_{-i}}-a_{i,x_{-i}}\leq \delta$. The equality is achieved if there exist $i$ and $x_{-i}$ such that the function $h(x_{-i},x_i)$ from $\mathcal{X}$ to $[a_{i,x_{-i}},c_{i,x_{-i}}]$ is surjective and $c_{i,x_{-i}}-a_{i,x_{-i}}=\delta$.    
\bibliographystyle{IEEEtran}
\bibliography{Bibliography}

\begin{IEEEbiographynophoto}
{Atefeh Gilani} (Student Member, IEEE) received the B.S. degree in electrical engineering from K. N. Toosi University of Technology, Iran in 2016, and the M.S. degree in electrical engineering from University of Tehran, Iran in 2019. She is currently pursuing the Ph.D. degree in the School of Electrical, Computer, and Energy Engineering at Arizona State University. Her research interests include privacy, information theory and machine learning.
\end{IEEEbiographynophoto}
\begin{IEEEbiographynophoto}
{Gowtham R. Kurri} (Member, IEEE) graduated from the International Institute of Information Technology (IIIT), Hyderabad, India, with a B.\ Tech.\ degree in Electronics and Communication Engineering, in 2011. He received his M.Sc. and Ph.D. degrees from the Tata Institute of Fundamental Research, Mumbai, India in 2020. From 2020-2023, he was a Post-Doctoral Researcher at the School of Electrical, Computer and Energy Engineering at Arizona State University. Since February 2023, he has been an Assistant Professor with IIIT Hyderabad, where he is affiliated to the Signal Processing and Communications Research Centre. His research interests are in information theory and statistical machine learning.

From 2011-2012, he worked as an Associate Engineer at Qualcomm India Private Limited, Hyderabad, India. From July to October, 2019, he was a Research Intern in the Blockchain Technology Group at IBM Research, Bangalore, India. 
\end{IEEEbiographynophoto}
\begin{IEEEbiographynophoto}{Oliver Kosut} (S'06--M'10--SM'22)
     received B.S. degrees in electrical engineering and mathematics from the Massachusetts Institute of Technology, Cambridge, MA, USA in 2004, and the Ph.D. degree in electrical and computer engineering from Cornell University, Ithaca, NY, USA in 2010.

Since 2012, he has been a faculty member in the School of Electrical, Computer and Energy Engineering at Arizona State University, Tempe, AZ, USA, where he is an Associate Professor. Previously, he was a Postdoctoral Research Associate at MIT from 2010 to 2012. His research interests include information theory---particularly with applications to security and machine learning---and power systems.

Prof.~Kosut received the NSF CAREER award in 2015. He is an associate editor for the \emph{IEEE Transactions on Information Forensics and Security}. He is an IEEE Information Theory Society Distinguished Lecturer, 2023--2024. He has also served as a co-lead editor of an issue on  Information-Theoretic Methods for Trustworthy and Reliable Machine Learning for the \emph{IEEE Journal on Selected Areas in Information Theory}.
\end{IEEEbiographynophoto}
\begin{IEEEbiographynophoto}{Lalitha Sankar} (S'02--M'07--SM'13) received the B. Tech. degree from the Indian Institute of Technology, Bombay, the M.S. degree from the University of Maryland, and the Ph.D. degree from Rutgers University. 

She is currently a Professor in the School of Electrical, Computer, and Energy Engineering at Arizona State University. Her research interests include applying information theory and data science to study reliable, responsible, and privacy-protected machine learning.

Prof. Sankar received the National Science Foundation CAREER Award in 2014, the IEEE Globecom 2011 Best Paper Award for her work on privacy of side-information in multi-user data systems, and the Academic Excellence award from Rutgers in 2008. She is presently an Associate Editor for the \emph{IEEE Transactions on Information Forensics and Security} and the \emph{IEEE Transactions on Information Theory} and has served as Associate Editor for the \emph{IEEE BITS Magazine}. She has also served as a co-lead editor of a special issue on Information-Theoretic Methods for Trustworthy and Reliable Machine Learning for the \emph{IEEE Journal on Selected Areas in Information Theory}.

\end{IEEEbiographynophoto}

\end{document}